\documentclass[11pt,letterpaper]{article}
\usepackage{hyperref}
\hypersetup{
     colorlinks   = true,
     linkcolor = black,
     urlcolor    = teal,
	 citecolor = teal
}
\usepackage{authblk}
\usepackage{natbib}
\usepackage{amsmath,amsfonts,amssymb}
\usepackage{amsthm}
\usepackage{blkarray}
\usepackage{algpseudocode}
\usepackage{thmtools,thm-restate}
\usepackage{caption}
\usepackage{subcaption}

\usepackage{bm}
\usepackage[margin=1.2in]{geometry}
\usepackage{enumitem,array,booktabs}
\usepackage{algorithm}
\usepackage{float}
\usepackage{graphicx}
\usepackage{bbm}
\usepackage{tabulary,multirow}
\usepackage{dsfont}

\makeatletter
\newenvironment{protocol}[1][htb]{%
    \renewcommand{\ALG@name}{Protocol}
    \begin{algorithm}[#1]%
    }{\end{algorithm}
}
\makeatother

\usepackage[textsize=scriptsize]{todonotes} 
\usepackage{xspace}
\definecolor{misocolor}{rgb}{0.16,0.27,0.86}

\definecolor{graphicbackground}{rgb}{0.96,0.96,0.8}
\definecolor{rouge1}{RGB}{226,0,38}  
\definecolor{orange1}{RGB}{243,154,38}  
\definecolor{jaune}{RGB}{254,205,27}  
\definecolor{blanc}{RGB}{255,255,255} 
\definecolor{rouge2}{RGB}{230,68,57}  
\definecolor{orange2}{RGB}{236,117,40}  
\definecolor{taupe}{RGB}{134,113,127} 
\definecolor{gris}{RGB}{91,94,111} 
\definecolor{bleu1}{RGB}{38,109,131} 
\definecolor{bleu2}{RGB}{28,50,114} 
\definecolor{vert1}{RGB}{133,146,66} 
\definecolor{vert3}{RGB}{20,200,66} 
\definecolor{vert2}{RGB}{157,193,7} 
\definecolor{darkyellow}{RGB}{233,165,0}  
\definecolor{lightgray}{rgb}{0.9,0.9,0.9}
\definecolor{darkgray}{rgb}{0.6,0.6,0.6}
\definecolor{babyblue}{rgb}{0.54, 0.81, 0.94}
\definecolor{citrine}{rgb}{0.89, 0.82, 0.04}
\definecolor{misogreen}{rgb}{0.25,0.6,0.0}
\definecolor{darkmagenta}{rgb}{0.5,0,0.5}

\DeclareMathOperator*{\argmax}{arg\,max}
\DeclareMathOperator*{\argmin}{arg\,min}

\newcommand{\floor}[1]{\left\lfloor#1\right\rfloor}

\newcommand{\Ber}{\mathrm{Ber}}

\newtheorem{definition}{Definition}
\newtheorem{lemma}{Lemma}
\newtheorem{theorem}{Theorem}

\newtheorem{remark}{Remark}
\newtheorem{assumption}{Assumption}

\newcommand{\NN}{{\mathbb N}}

\newcommand{\ind}[1]{\mathds{1}[#1]}

\newcommand{\E}{\mathbb{E}}

\newcommand{\EEs}[2]{\mathbb{E}_{#1}\left[#2\right]}
\newcommand{\EEc}[2]{\mathbb{E}\left[#1\left|#2\right.\right]}

\newcommand{\norm}[1]{\left\|#1\right\|}

\newcommand{\abs}[1]{\left|#1\right|}

\newcommand{\cA}{\mathcal{A}}
\newcommand{\cB}{\mathcal{B}}

\newcommand{\cE}{\mathcal{E}}

\newcommand{\cF}{\mathcal{F}}

\newcommand{\cL}{\mathcal{L}}

\newcommand{\cO}{\mathcal{O}}

\newcommand{\cP}{\mathcal{P}}

\newcommand{\cS}{\mathcal{S}}

\newcommand{\cY}{\mathcal{Y}}

\renewcommand{\epsilon}{\varepsilon}
\renewcommand{\hat}{\widehat}
\renewcommand{\tilde}{\widetilde}
\renewcommand{\bar}{\overline}

\newcommand{\nothere}[1]{}

\newcommand{\conv}{\mathrm{Conv}}

\newcommand{\A}{\mathcal{A}}

\newcommand{\pip}{\pi}

\newcommand{\brp}{p^*}
\newcommand{\brr}{a^*}

\newcommand{\ir}{\textrm{SwapReg}}
\newcommand{\nr}{\textrm{NegReg}}

\newcommand{\eint}{\epsilon_\textrm{swap}}

\newcommand{\eneg}{\epsilon_\textrm{neg}}

\newcommand{\sig}{\varphi}

\newcommand{\pit}{p_t^\textrm{optimistic}}
\newcommand{\rit}{a_t^\textrm{optimistic}}

\newcommand{\ugap}{\textrm{UGap}}

\allowdisplaybreaks
\begin{document}

\title{Efficient Prior-Free Mechanisms for No-Regret Agents}
\date{}

\author[1]{Natalie Collina\thanks{Supported in part by an AWS AI Gift for Research in Trustworthy AI.}}
\author[1]{Aaron Roth\thanks{Supported in part by the Simons Collaboration on the Theory of Algorithmic Fairness, and NSF grants FAI-2147212 and CCF-2217062.}}
\author[2]{Han Shao\thanks{Supported in part by the National Science Foundation under grants 2212968 and 2216899, by the Simons Foundation under the Simons Collaboration on the Theory of Algorithmic Fairness, by the Defense Advanced Research Projects Agency under cooperative agreement HR00112020003.}}
\affil[1]{University of Pennsylvania Department of Computer and Information Sciences}
\affil[2]{Toyota Technological Institute at Chicago}

\maketitle

 \begin{abstract}
We study a repeated Principal Agent problem between a long lived Principal and Agent pair in a prior free setting.  In our setting, the sequence of realized states of nature may be adversarially chosen, the Agent is non-myopic, and the Principal aims for a strong form of policy regret. 
Following \cite{camara2020mechanisms}, we model the Agent's long-run behavior with behavioral assumptions that relax the common prior assumption (for example, that the Agent has no swap regret). Within this framework, we revisit the mechanism proposed by \cite{camara2020mechanisms}, which informally uses calibrated forecasts of the unknown states of nature in place of a common prior. We give two main improvements. First, we give a mechanism that has an exponentially improved dependence (in terms of both running time and regret bounds) on the number of distinct states of nature. To do this, we show that our mechanism does not require truly calibrated forecasts, but rather forecasts that are unbiased subject to only a polynomially sized collection of events --- which can be produced with polynomial overhead. Second, in several important special cases---including the focal linear contracting setting---we show how to remove strong ``Alignment'' assumptions (which informally require that near-ties are always broken in favor of the Principal) by specifically deploying ``stable'' policies that do not have any near ties that are payoff relevant to the Principal. Taken together, our new mechanism makes the compelling framework proposed by \cite{camara2020mechanisms} much more powerful, now able to be realized  over polynomially sized state spaces, and while requiring only mild assumptions on Agent behavior. 
 \end{abstract}

 \thispagestyle{empty} \setcounter{page}{0}
 \clearpage

\tableofcontents

 \thispagestyle{empty} \setcounter{page}{0}
 \clearpage

\section{Introduction}
Many mechanism design settings can be cast as \emph{Principal/Agent Problems}. These are Stackelberg games of incomplete information, in which the \emph{Principal} first commits to some policy, and then the \emph{Agent} chooses an action by best responding. The utility for both the Principal and the Agent can depend on the actions they each choose, as well as some underlying and unknown state of nature. The fact that the state of nature is unknown is a crucial modeling aspect of Principal Agent problems. Two canonical examples of Principal Agent problems will be instructive: a simple example of a contract theory problem (see e.g. \cite{carroll2021contract}) and of a Bayesian Persuasion Problem \citep{kamenica2011bayesian}.
\begin{enumerate}
    \item \textbf{Contract Theory}: Consider a Principal (say a university endowment office) that has capital that they would like to invest, but who does not themselves have the expertise to invest it effectively. Instead they would like to contract with an Agent (say a hedge fund) so as to maximize their returns. The Agent will choose a strategy (say by dividing funds across a particular portfolio of investments), but the return of the strategy will be unknown at the time that they choose it---it depends on the unknown-at-the-time-of-action returns of each investment. Moreover they may be able to choose a better strategy by investing more time, effort, and money (for example, by hiring talented fund managers away from competing hedge funds). But should they? It is in the Principal's interest that their returns (minus their fees) should be maximized, but it is in the Agent's interest that their fees (minus their costs) should be maximized. How should the Principal design the contract (i.e. a mapping from outcomes to payments to the Agent) so that their utility is maximized when the Agent best responds? 
    \item \textbf{Bayesian Persuasion}: Consider a Principal (say a pharmaceutical company) that manufactures drugs that they need to get approved by an Agent (say a regulatory agency like the FDA) before they can be sold. The drugs will have various properties which we can think of as an underlying state comprising effectiveness, safety, etc. that are initially unknown. But drug trials (that may be at least partially designed by the Principal) will be run that will provide a noisy signal about the qualities of the drug, that the Agent will use to form a belief about the state, and as a result, either approve the drug or not. It is in the Principal's interest that as many drugs as possible should be approved --- but the Agent will approve only those drugs that it believes are safe. How should the Principal design the drug trial (i.e. a stochastic mapping from state to observable signal) so that as many drugs as possible are approved when the Agent best responds? 
\end{enumerate}
The classical economic literature answers these questions in a conceptually straightforward manner (although the structure of the solution can be intricate and rich): The Principal should commit to a strategy such that her payoff will be maximized after the Agent best responds. But given that the state is unknown, how will the Agent choose to best respond, and how will the Principal anticipate the Agent's choice? The classical answer is that the Principal and the Agent share a common \emph{prior distribution} on the unknown state of the world: the Agent best-responds so as to maximize his utility in expectation over this Prior, and the Principal, also being in possession of the same beliefs, anticipates this. There are some assumptions that are traditionally made about tie-breaking (that it is done in favor of the Principal) that we will interrogate, but the reader can ignore these for now. A strong general critique of the foundations of this literature asks: In a complex, dynamic world, where does this prior belief come from, and why is it reasonable to assume it is shared? 

Recently, \cite{camara2020mechanisms} gave an elegant framework for addressing this critique head on. They study a repeated Principal Agent problem (where two long-lived parties interact with each other repeatedly) and dispense with the common prior assumption entirely. In fact, there are no distributional assumptions at all in their model: the sequence of realized states of nature can be arbitrary or even adversarially chosen. Instead, it is assumed that the Agent behaves in a way that is consistent with various efficiently obtainable online-learning desiderata, which are elaborations on the goal that they should have no \emph{swap regret} \citep{blum2007external}, and that they don't have too much ``additional information'' about the state sequence compared to the Principal (this can be formalized in various ways that we shall discuss). These are assumptions that would be satisfied were there a common prior that both Agents were optimizing under --- but can be reasonably assumed (because they can be efficiently algorithmically obtained) without this assumption. Under a collection of such behavioral assumptions --- and other assumptions on the structure of the game --- \cite{camara2020mechanisms} show that a Principal who maintains \emph{calibrated forecasts} for the unknown states of nature, and acts by treating these forecasts as if they were a common prior --- is able to guarantee themselves a strong form of \emph{policy regret}. That is, they are guaranteed to obtain utility nearly as high as they would have had they instead played any fixed policy in some benchmark class, \emph{even accounting for how the Agent would have acted under this counter-factual policy}. Moreover, it has been known since \cite{foster1998asymptotic} that it is possible to produce calibrated forecasts of an arbitrary finite dimensional state, even if the state sequence is chosen adversarially --- so the mechanism proposed by \cite{camara2020mechanisms} could in principle be implemented in their model. This makes the model of \cite{camara2020mechanisms} a compelling alternative to common prior assumptions. Nevertheless, there remain some difficulties with the mechanism they propose within this framework:
\begin{enumerate}
    \item \textbf{Computational and Statistical Complexity}: Informally speaking, a method of producing forecasts $\hat s \in \mathbb{R}^d$ of a $d$-dimensional state $s \in \mathbb{R}^d$ is \emph{calibrated} if the forecasts are unbiased, not just overall, but conditional on the forecast itself: $\E_{s, \hat s}[s | \hat s] = \hat s$, for all values of $\hat s$. When we are forecasting probability distributions over a finite collection of states of nature $\cY$, the forecasts are probability distributions represented as $|\cY|$-dimensional  vectors $\hat s \in \Delta (\cY)$. Under any reasonable discretization, there are $\Omega\left(2^{|\cY|}\right)$ many such vectors, and algorithms for maintaining calibrated forecasts in this space have both computational and statistical complexity scaling exponentially with $|\cY|$. The mechanism proposed by \cite{camara2020mechanisms} inherits these limitations: and as a result has both running time and regret bounds that suffer \emph{exponential} dependencies on the cardinality of the state space $|\cY|$. Thus these mechanisms are reasonable only for very small constant sized state spaces. 
    \item \textbf{Strong ``Alignment'' Assumptions}: Even in the classical model in which the Agent ``best responds'' to the policy of the Principal, using their prior beliefs over the state of nature, there can be ambiguity in how the Agent will act. In particular, what if their set of best responses is not a singleton set: there are multiple actions that they can take that yield the same utility for the Agent---which action will they take? This is an important detail, because even when the Agent's utilities are tied over this set, each action may yield very different utility for the Principal. The traditional assumption is that the Agent breaks ties in favor of the Principal---which although optimistic can perhaps be viewed as a mild assumption because it concerns only exact ties. However, when there is doubt or imprecision about the Agent's beliefs, this problem is exacerbated: one could assume that \emph{near} ties are broken in favor of the Principal, but it is much less reasonable to assume that the Agent will forgo small gains so as to benefit the Principal; a similar phenomenon arises with the mechanism of  \cite{camara2020mechanisms} because sequential forecasts will never be exactly, but only approximately calibrated. \cite{camara2020mechanisms} deal with this issue by making strong ``alignment'' assumptions, which informally require that with respect to all possible prior distributions, the difference in Agent utilities between a pair of actions is comparable to the corresponding change in Principal utilities. This has the effect of making approximate tie-breaking (almost) irrelevant for the Principal. Unlike the behavioral assumptions placed on the Agent, which generalize the common prior assumption, however, these Alignment assumptions are restrictive and not commonly satisfied. It would be preferable to be able to remove them: whenever they can be removed entirely, the model makes strictly weaker assumptions than a common prior. 
\end{enumerate}

\subsection{Our Results}

In this paper we revisit the framework of \cite{camara2020mechanisms} and derive new mechanisms which address these issues. Our mechanisms obtain strong policy regret guarantees, but  are exponentially more efficient (in their dependence on the cardinality of the state space) in terms of both their running time and their regret bounds. Moreover, in a subset of instances (which we show includes linear contracting, that has been the exclusive focus of a large fraction of recent computational work in contract theory) our mechanisms entirely eliminate the need for alignment assumptions.

\paragraph{Computational and Statistical Efficiency---Beyond Calibration:} We show how to obtain both policy regret bounds and running time bounds that scale polynomially with the cardinality of the state space $|\cY|$, rather than exponentially (as in \cite{camara2020mechanisms}). To do this, we need to give mechanisms that do not rely on fully calibrated forecasts of the state of nature. Instead, we give mechanisms that use forecasts of the state of nature that are \emph{statistically unbiased} subject only to a polynomial number of events: informally, the events that the forecasts themselves (were they used as a common prior) would lead the Principal to propose each particular policy, and anticipate each particular action in response by the Agent. Calibration requires unbiasedness subject to exponentially many (in the cardinality of the state space $|\cY|$) events; here we require unbiasedness with respect to only quadratically many events (in the cardinality of the action space of the Principal and the Agent). Using a recent algorithm of \cite{NRRX23}, we are able to produce forecasts with these properties with running time that is polynomial in the cardinality of the state space $|\cY|$ and the action spaces of the Agent and the Principal. Under similar behavioral assumptions as \cite{camara2020mechanisms} (which strictly generalize the common prior assumption), we show that our mechanism obtains policy regret bounds that scale linearly with $|\cY|$ (again, compared to exponentially with $|\cY|$ in \cite{camara2020mechanisms}). 

\paragraph{Stable Policy Oracles---Avoiding Alignment Assumptions:} As discussed above, Alignment assumptions are needed in \cite{camara2020mechanisms} to address, informally, the problem of the mechanism's proposed policy inducing ``near-ties'' in the Agent's utility that  nevertheless lead to very different Principal utility. In contrast, we define a policy to be \emph{stable} with respect to a state distribution $\pi$ if when compared to the Agent's best response to the policy under $\pi$, every other action \emph{either} leads to substantially lower utility for the Agent (in expectation over the distribution), or else leads to nearly the same utility for the Principal. We show that if our mechanism has the ability to construct stable policies that also lead to near optimal utility for the Principal under the Principal's current state forecast, then she can obtain strong policy regret bounds without the need for an Alignment assumption. We then turn to the task of constructing near-optimal stable policies. We show by example that this is not possible for all Principal-Agent games within the framework we consider; but show how to do it in two important special cases. The first is the \emph{linear contracting} setting---the special case of contract theory in which the contract space is restricted to be a linear function of a real valued outcome (e.g. ``The Agent receives payment equal to 10\% of the revenue of the Principal''). Linear contracts are focal within the contract theory literature because they have a variety of robustness properties (see e.g. \cite{carroll2015robustness,dutting19})---and because they are the most commonly used type of contract in practice. As a result they have been the focus of a large fraction of the recent computational work in contract theory (see our discussion in the Related Work section). The second is the Bayesian Persuasion setting when the underlying state of nature is binary: e.g. drugs that are either effective or not, or defendants that are either innocent or guilty. This captures some of the best studied Bayesian Persuasion instances.

\paragraph{Guide to the Paper}
In Section \ref{sec:model} we define the model that we will be working under, following \cite{camara2020mechanisms}. In Section \ref{sec:behavior}, we state and discuss the behavioral assumptions that we make on the Agent throughout this paper. In Section \ref{sec:stable}, we derive our results when we have access to a \emph{stable policy oracle}---in this case, we do not need to make any ``alignment'' assumptions on the underlying game. In Section~\ref{sec:oracles} we show how to derive optimal stable policy oracles for linear contracting problems and for binary state Bayesian Persuasion problems. In Section \ref{sec:general} we consider the general case, in which we do not have the ability to construct stable policies. Here, like \cite{camara2020mechanisms}, we also need to make an alignment assumption. In Section \ref{sec:imposs} we interrogate the need for our assumptions and show several impossibility results that arise from not making them. In particular, we give an example of a game in which there is no stable policy oracle, which demonstrates that our approach for removing alignment assumptions cannot be generalized to all Principal Agent problems within the framework we study.

\subsection{Additional Related Work}
The foundations of principal agent problems and contract theory (in the standard setting with common priors) date back to \cite{holmstrom1979moral} and \cite{grossman1992analysis}. This literature is far too large to survey --- we refer the reader to \cite{bolton2004contract} for a textbook introduction, and here focus on only the most relevant work. 

Optimal contracts under a common prior assumption can be very complicated, and do not reflect structure seen in real world contracts. This criticism goes back to at least \cite{holmstrom1987aggregation}, who show a dynamic setting in which optimal contracts are linear. Recently, linear contracts have become an object of intense study, with work showing that they are optimal in various worst-case settings. In the classical common prior setting, \cite{carroll2015robustness} shows that linear contracts are minimax optimal for a Principal who knows \emph{some} but not \emph{all} of the Agent's actions. Similarly, \cite{dutting19} shows that if the Principal only knows the costs and expected rewards for each Agent action, then linear contracts are minimax optimal over the set of all reward distributions with the given expectation. \cite{dutting2022combinatorial} extends this robustness result to a combinatorial setting. \cite{dutting19} also show linear contracts are bounded approximations to optimal contracts, where the approximation factor can be bounded in terms of various quantities (e.g. the number of agent actions, or the ratio of the largest to smallest reward, or the ratio of the largest to smallest cost, etc). 
\cite{castiglioni2021bayesian} studies linear contracts in Bayesian settings (when the Principal knows a distribution over types from which the Agent's type is drawn) and studies how well linear contracts can approximate optimal contracts. In this setting, optimal contracts can be computationally hard to construct, and show that linear contracts obtain optimal approximations amongst tractable contracts. 

There is also a more recent tradition of studying sequential (repeated) principle agent games. \cite{ho2014adaptive} study online contract design by approaching it as a bandit problem in which an unknown distribution over myopic agents arrive and respond to an offered Principal contract by optimizing their expected utility with respect to a known prior. \cite{cohen2022learning} extend this to the case in which the Agent has bounded risk aversion. 
\cite{zhu2022sample} revisit this problem and characterize the sample complexity of online contract design in general (with nearly matching upper and lower bounds) and for the special case of linear contracts (with exactly matching upper and lower bounds). In contrast to this line of work, our Agent is not myopic --- a primary challenge is that we need to manage their long-term incentives --- and we make no distributional assumptions at all, either about the actual realizations nor about agent beliefs. 

\cite{chassang2013calibrated} studies a repeated interaction between a Principal and a long-lived Agent, with a focus on the \emph{limited liability} problem. As discussed, linear contracts have many attractive robustness properties, but can require negative payments from the Agent, which are difficult to implement. A limited liability contract, in contrast, never requires negative payments. Using a Blackwell-approachability argument, \cite{chassang2013calibrated} shows how to repeatedly contract with a single Agent (or instead to use a free outside option) so that the aggregate payments made to the agent is the same as they would have been under a linear contract, but negative payments are never required, and the Principal has no regret to either always contracting with the agent or always using the outside option.

The Bayesian Persuasian problem was introduced by \cite{kamenica2011bayesian} and has been studied from a computational perspective since \cite{dughmi2016algorithmic}. It has been applied to various problems, including incentivizing exploration in multiarmed bandit problems \citep{cohen2019optimal,sellke2021price,mansour2022bayesian}. A recent literature has studied sequential Bayesian Persuasian problems. \cite{zu2021learning} and \cite{bernasconi2022sequential} study a sequential Bayesian Persuasian problem in which the Principal does \emph{not initially} know the underlying distribution on the state space, and needs to learn it while acting in the game. \cite{wu2022sequential} study a sequential problem in which a Principal repeatedly interacts with myopic agents, using tools from reinforcement learning. \cite{gan2022bayesian} study a sequential Bayesian Persuasian problem in which the state evolves according to a Markov Decision Process, and show that for a myopic agent, the optimal signalling scheme can be computed efficiently, but that it is computationally hard for a non-myopic agent. \cite{bernasconi2023optimal} study regret bounds for a Principal in a sequential Bayesian Persuasian problem facing a sequence of myopic Agents, whose utility functions can be chosen by an adversary.

There is a substantial body of work on learning in repeated Stackelberg games (both in general and in various special cases like security games, strategic classification, and dynamic pricing) in settings in which the Agent has complete information and the Principal needs to learn about the Agent's preferences (see e.g. \citep{blum2014learning,Balcan2015CommitmentWR,roth2016watch,dong2018strategic,chen2020learning,roth2020multidimensional}). In these works, the Agent is myopic and optimizes for their one-round payoff. \cite{haghtalab2022learning} consider a non-myopic agent who discounts the future, and give no-regret learning rules for the Principal that take advantage of the fact that for a future-discounting agent, mechanisms that are slow to incorporate learned information will induce near-myopic behavior. The regret bounds in \cite{haghtalab2022learning} tend to infinity as the Agent becomes more patient. 
\cite{collina2023efficient} derive optimal commitment algorithms for complete-information Stackelberg games when the follower is maximizing their total payoff in expectation. In contrast to these works, we (and \cite{camara2020mechanisms} before us) operate in a setting without distributions (or assumed distributions that Agents can be said to optimize over) and give policy regret bounds contingent on Agent's satisfying behavioral assumptions defined by regret bounds.  This is  similar in spirit to \cite{deng2019strategizing}, which considers playing a repeated game against an agent playing a no-swap regret algorithm and shows that the optimal strategy is to play the single-shot Stackelberg equilibrium at each round. \cite{haghtalab2023calibrated} show that the same is true if an agent is best-responding to a calibrated predictor for the Principal's actions --- and accomplish this also by using a form of ``stable'' policies as we do. 

There is a long tradition of using ``no-regret'' assumptions as relaxations of classical assumptions that players in a game either best respond to beliefs or play a Nash equilibrium --- for example, when proving price of anarchy bounds \citep{blum2008regret,roughgarden2015intrinsic,lykouris2016learning}, when doing econometric inference \citep{nekipelov2015econometrics}, or when designing optimal pricing rules \citep{braverman2018selling,cai2023selling}, as well as work focused on how to play games against no-regret learning agents \citep{deng2019strategizing,mansour2022strategizing,kolumbus2022and,brown2023learning}.

Finally, the use of calibrated forecasts in decision-making settings dates back to \cite{foster1999regret}, who showed that agents best-responding to calibrated forecasts of their payoffs have no internal (equivalently swap) regret. Similarly \cite{kakade2008deterministic} and \cite{foster2018smooth} connect a determinstic ``smooth'' version of calibration to Nash equilibrium. A recent literature on ``multicalibration'' \citep{hebert2018multicalibration} has investigated various refinements of calibration; this has developed into a large literature and we refer the reader to \cite{RothNotes} for an introductory overview. Work on ``omniprediction''  \citep{GopalanKRSW22,gopalan2023loss,GHK23,gopalan2023characterizing,GJRR23} uses multicalibration to provide guarantees for a variety of 1-dimensional downstream decision making problems. Decision calibration \citep{zhao2021calibrating} (in the batch setting) aims to calibrate predictions to the best-response correspondence of a downstream decision maker. The tools we use, developed by \cite{NRRX23} arise from this literature. 
\section{Model}
\label{sec:model}

Consider a repeated Stackelberg game between a female Principal and a male Agent with policy space $\cP$, action space $\cA$, and state space $\cY$. In rounds $t\in \{1,\ldots,T\}$, the Principal selects a policy $p_t\in \cP$ and (possibly) recommends an action $r_t\in \cA$ for the Agent. After observing the policy $p_t$ and the recommendation $r_t$, the Agent takes an action $a_t\in \cA$. At the end of round $t$, a state of nature $y_t$ chosen by nature is revealed to both the Principal and the Agent.
Utility functions depend on the action, the policy and the state of nature.
We denote the Agent's utility by $U(a_t,p_t,y_t)\in [-1,1]$ and the Principal's utility by $V(a_t,p_t,y_t) \in [-1,1]$.
For example, in the context of contract design, a policy corresponds to a contract, the action (to follow a  traditional two-action toy example) could be either ``working'' or ``shirking'', and the state of nature corresponds to the difficulty level of the job. 

When there is a known (to both the Principal and the Agent) common prior $\pi\in \Delta(\cY)$ and the state of nature $y_t$ is drawn from this prior, the Principal can maximize her utility by solving for an optimal policy by backwards induction, choosing the policy that will maximize her utility after the Agent best responds by breaking ties in favor of the Principal. Formally, for any prior distribution $\pi$, if the Principal selects a policy $p$, then the Agent will best respond to $(p, \pi)$ by choosing an action in $A^*(p,\pi):= \argmax_{a\in \cA} \EEs{y\sim \pi}{U(a,p,y)}$ to maximize the Agent's utility. When there are multiple best responding actions, the traditional assumption is that the Agent will break ties by maximizing the Principal's utility, i.e., 
\begin{equation}
    \brr(p,\pi) \in \argmax_{a\in A^*(p,\pi)} \EEs{y\sim \pi}{V(a,p,y)}\,. \label{eq:brr}
\end{equation} 

The Principal, assuming that the Agent will best respond, best responds to $\pi$ by selecting policy
\begin{equation}
    \brp(\pi) \in \argmax_{p\in \cP_0} \EEs{y\sim \pi}{V(\brr(p,\pi),p,y)}\,, \label{eq:brp}
\end{equation}
where $\cP_0\subseteq \cP$ is a set of given benchmark policies.
In Eq~\eqref{eq:brr} and~\eqref{eq:brp}, we break ties arbitrarily.
Therefore, given a prior $\pi$, the Principal will choose policy $p_t = \brp(\pi)$ and (may without loss of generality) recommend that the Agent take action $r_t= \brr(\brp(\pi),\pi)$. The Agent will follow the Principal's recommendation by taking action $r_t$.

In this work, we consider a more challenging prior-free scenario where there is no common prior and the states of the world can be generated adversarially. We also will \emph{not} assume that the Agent breaks ties in favor of the Principal.  
The Agent runs a learning algorithm $\cL$, which maps the state history $y_{1:t-1}$, the action history $a_{1:t-1}$, the recommendation history $r_{1:t-1}$, the policy history $p_{1:t-1}$, and the current policy $p_t$ and recommendation $r_t$ to a distribution over actions. Formally, the Agent's action distribution at round $t$ is given by a function:
\begin{equation*}
    \cL_t: \cY^{t-1}\times \cA^{t-1}\times \cA^{t}\times \cP^t\mapsto \Delta(\cA)\,.
\end{equation*}

The Principal runs a learning algorithm (henceforth, a mechanism $\sigma$) that maps the state history $y_{1:t-1}$, the recommendation history $r_{1:t-1}$, and the policy history $p_{1:t-1}$ to a distribution over policies and recommendations. Note that the Principal's algorithm does \emph{not} depend on the action history, which is by design (and in fact it is an important modelling choice that Agent's actions need not be directly observable to the Principal). The result is that the Principal's mechanism is \emph{nonresponsive} to Agent's actions, i.e., the Principal's policy at time $t$ does not depend on the Agent's action history. 
When mechanisms are nonresponsive, non-policy regret and policy regret coincide for the Agent and so lack of ``regret'' (to be defined shortly) is an unambiguously desirable property for the Agent to have. 
Formally, the Principal's policy distribution at round $t$ is given by a function: 
\begin{equation*}
    \sigma_t: \cY^{t-1}\times \cA^{t-1}\times \cP^{t-1}\mapsto \Delta(\cP\times \cA)\,.
\end{equation*}
In this work, we consider a specific family of mechanisms in which the Principal generates a forecast of the distribution over states in each round that will satisfy certain ``unbiasedness'' conditions, to be specified shortly. These forecasts will informally play the role of the prior distribution in the Principal's decision about which policy to offer. 

Specifically, assume that the Principal has access to a forecasting algorithm (implemented by either herself or a third party), which provides a forecast $\pi_t\in \Delta(\cY)$ of (the distribution over) the state in each round $t$. By viewing $\pi_t$ as the prior, the Principal selects policy $p_t = \psi(\pi_t)$, which is determined by $\pi_t$ and recommends that the Agent play the best response $r_t = \brr(p_t,\pi_t)$ --- as if $\pi_t$ were in fact a prior. The recommendation is the best action that the Agent could play were $\pi_t$ in fact a correct prior. The Agent is under no obligation to follow this recommendation, and may not---the recommendation is only as good as the Principal's forecast. However, in our mechanism, the forecasts will turn out to guarantee that \emph{if} the Agent follows the recommendation, then he will have strong regret guarantees with respect to his own utility function---and the behavioral assumptions we impose on the Agent will require that he satisfies these regret guarantees (whether or not he chooses to do so by following the recommendation, or satisfies these guarantees through some other means).

We only consider deterministic rules $\psi: \Delta(\cY)\mapsto \cP$, mapping forecasts $\pi_t$ to policies $p_t$ and our recommendations will always be $r_t = \brr(p_t,\pi_t)$.
The Principal-Agent interaction protocol is described as follows. 
\begin{protocol}[H]
    \caption{Principal-Agent Interaction at round $t$}
    \label{prot:interaction}
        \begin{algorithmic}[1]
            \State{The Principal produces or obtains a forecast $\pi_t$.}
            \State{The Principal chooses policy $p_t = \psi(\pi_t)$ and recommends that the Agent play action $r_t = \brr(p_t,\pi_t)$.}
            \State{The Principal discloses $(p_t,r_t)$ to the Agent.}
            \State{The Agent takes an action $a_t \sim \cL_t(y_{1:t-1}, a_{1:t-1},r_{1:t}, p_{1:t})$.}
            \State{The state $y_t$ is revealed to both the Principal and the Agent.}
        \end{algorithmic}
    \end{protocol}


Mechanisms designed within this framework (the only sort we consider in this paper) are specified by  a forecasting algorithm $\cF$ and a choice rule $\psi$ mapping forecasts to polices. 
Given a forecasting algorithm $\cF$, we want a choice rule $\psi$ that guarantees that the Principal has no ``regret'' to using, relative to having counter-factually offered the best fixed policy in hindsight, which we think of as using a constant ``mechanism'' from the set $\{\sigma^{p_0}|p_0\in  \cP_0\}$. The constant mechanism $\sigma^{p_0}$ ignores the history, and consistently chooses the policy $p_0\in \cP_0$ at every round, while recommending that the Agent take action $r_t^{p_0} = a^*(p_0,\pi_t)$---i.e. his best response to $p_0$ under the current realized forecast. Note that the sequence of \emph{forecasts} is the same under both the realized and counter-factual constant mechanism. Here we will define a strong notion of policy regret --- regret to the counterfactual world in which the Principal used a fixed policy, \emph{and the Agent responded to that fixed policy, producing a different sequence of actions}. 
Formally we define the Principal's policy regret as follows.

\begin{definition}[Principal's Regret]
    For a realized sequence of states of nature $y_{1:T}$, an Agent learning algorithm $\cL$, and a realized sequence of forecasts $\pi_{1:T}$, the Principal's policy regret from having used a rule $\psi$ is defined as:
\begin{equation*}
    \textrm{PR}(\psi,\pi_{1:T}, \cL, y_{1:T}) = \max_{p_0\in \cP_0} \EEs{a_{1:T},a^{p_0}_{1:T}}{\frac{1}{T}\sum_{t=1}^T\left(V(a_t^{p_0},p_0,y_t) -V(a_t,p_t,y_t)\right)} \,,
\end{equation*}
where $p_t = \psi(\pi_t)$ is the policy selected by the rule $\psi$, $a_{1:T}$ and $a^{p_0}_{1:T}$ are the sequences of actions generated by $\cL$ when the Principal selects policies according to the proposed rule $\psi$ and the constant policy $p_0$ respectively. The expectation is taken over the randomness of the learning algorithm $\cL$.
    

Observe that the forecasts are an argument to the Principal's regret, and these are random variables because the forecasting algorithm is permitted to be randomized. For a mechanism $\sigma^\dagger = (\cF,\psi)$, we compute the Principal's regret  by taking the expectation over the random forecasts generated by $\cF$
 \begin{equation*}
    \textrm{PR}(\sigma^\dagger,\cL, y_{1:T}) = \EEs{\pi_{1:T}}{\textrm{PR}(\psi,\pi_{1:T}, \cL, y_{1:T})}\,.
\end{equation*}
 
\end{definition}

Throughout this work, we consider finite action spaces and finite state spaces.
For notational simplicity, we represent actions $a\in \cA$ and states $y\in \cY$ in their one-hot encoding vector forms.

\section{Behavioral Assumptions}\label{sec:behavior}
In the common prior setting, it is clear how to model rational Agent behavior---the standard assumption is that the Agent chooses his action so as to maximize his payoff in expectation over the prior. This assumption, of course, no longer makes sense in a prior-free setting. However, we cannot simply drop all behavioral assumptions on the Agent when moving to the prior-free setting. Consider what happens if we allow the Agent's algorithm to be any mapping from a history of nature states, policies, and recommendations to an action in the current round. Then, the Agent's algorithm could be entirely agnostic to his own payoffs, playing actions with the sole purpose of minimizing the Principal's payoff under the Principal's deployed mechanism. The same algorithm for the Agent might, under some alternative mechanism for the Principal, choose  actions so as to \emph{maximize} the Principal's payoff. Such an algorithm will always lead to high policy regret for the Principal; to obtain diminishing policy regret, we need to make assumptions on the Agents' behavior that constrain them to be ``rational'' in some way. Similarly, we must preclude Agents that have perfect foreknowledge of the states of nature hard-coded into their learning algorithm when this information is not available to the Principal --- because he could then selectively use this information in a way that would preclude proving a bound on (counter-factual) policy regret. See \cite{camara2020mechanisms} and Section~\ref{sec:imposs} for extended discussions of these issues.

The upshot is that we cannot dispense with behavioral assumptions entirely. Instead, we establish more general assumptions which make sense in the prior-free setting. Our behavioral assumptions must hold in both the realized sequence of play and in several counterfactual scenarios, so that we can meaningfully measure policy regret. Taken together, the assumptions below are strictly weaker than the assumption that the Agent always best-responds to a common prior. The reader can therefore view our behavioral assumptions as a strict generalization of the definition of rational behavior in a common prior setting, which can be studied in the prior-free setting. The assumptions will also end up being strictly weaker than the assumption that the Agent follows the Principal's recommended action --- so they are easily satisfied if the Agent chooses to do this, but do not constrain the Agent to following the Principal's recommendations. We will now introduce our two key assumptions, along with intuition for how they generalize the common prior setting. 

The first assumption generalizes the `best-response' behavior of the Agent. While our Agent may not have access to a prior to best-respond to, we can still rule out some clearly suboptimal behavior. A standard prior-free rationality assumption is that the Agent should have no swap regret: i.e. for each of his actions, on the subsequence of rounds on which he played that action, he should be obtaining utility at least what he could have guaranteed by playing the best \emph{fixed} action on that subsequence. Swap regret is an efficiently obtainable guarantee, weaker than pointwise optimality under a common prior, and having lower swap regret is always desirable, since the Principal is non-responsive. Of course, in our setting, in which the Principal first commits to a policy, which defines the best response correspondence of the Agent, it makes little sense to speak of the ``best fixed action'' without first conditioning on the policy offered by the Principal. So we ask for a form of contextual swap regret that is a better fit to our setting: namely, that the Agent should have no swap regret not just overall, but on each subsequence that results from \emph{fixing} the policy and recommendation made by the Principal. Once again, this is a weaker assumption than that the Agent is best responding to a shared prior --- if the Agent is playing a pointwise optimal action, he will have no swap regret on every subsequence.  It also still always desirable (since the Principal is non-responsive), and efficiently obtainable in a prior-free setting: for example, by running a copy of a no-swap-regret algorithm like \cite{blum2007external} separately for each policy/recommendation pair $(p,r)$ offered by the Principal, or by best responding to appropriately calibrated, efficiently computable forecasts as in \cite{NRRX23}.



\begin{assumption}[No Contextual Swap Regret for The  Agent]\label{asp:no-internal-reg} 
We write $h:\cP\times \cA\times \cA\mapsto \cA$ to denote a modification rule that takes as input a policy and recommended action from the Principal, as well as a played action by the Agent, and as a function of these arguments ``swaps'' the Agent's action for an alternative action.  
Given the realized sequence of states $y_{1:T}$ and the realized sequence of policies and recommendations generated by either the deployed mechanism or the constant mechanisms, we define the Agent's swap regret to be:

    \begin{equation*}        \ir(y_{1:T},p_{1:T},r_{1:T}):=\E_{a_{1:T}}\left[\max_{h:\cP\times \cA\times \cA\mapsto \cA}\frac{1}{T}\sum_{t=1}^T (U(h(p_t,r_t,a_t),p_t, y_t) - U(a_t,p_t,y_t)) \right]\,,
    \end{equation*}
    and for all $p_0\in \cP_0$,
    \begin{equation*}        
    \ir(y_{1:T},(p_0,\ldots,p_0),r^{p_0}_{1:T}):=\E_{a_{1:T}^{p_0}}\left[\max_{h:\cP\times \cA\times \cA\mapsto \cA}\frac{1}{T}\sum_{t=1}^T (U(h(p_0,r^{p_0}_t,a^{p_0}_t),p_0, y_t) - U(a^{p_0}_t,p_0,y_t)) \right]\,.
    \end{equation*}
    We assume that there exists an $\eint$ such that for all fixed policies $p_0 \in \cP_0$ we have both:
    $$\ir(y_{1:T},p_{1:T},r_{1:T}) \leq \eint \ \ \ \ \ir(y_{1:T},(p_0,\ldots,p_0),r^{p_0}_{1:T}) \leq \eint\,.$$
    
\end{assumption}

The second assumption generalizes the notion of a shared prior. One important feature of the shared prior setting is that the realized state of nature is independent of the actions chosen by both the Principal and the Agent. In an adversarial setting, we can no longer appeal to statistical independence, as there is no distribution. But we need to preclude the possibility that the Agent somehow can ``predict the future'' in ways that the Principal can't. To do this, we make a ``no secret information'' assumption that informally requires that the Agent's actions appear to be (almost) statistically independent of the states of nature in the empirical transcript in terms of the utility functions of the Principal and Agent, conditionally on the policies and recommendations chosen by the Principal. Once again, this generalizes the shared prior assumption, in which we have actual statistical independence---and in which the Principal's ``recommendation'' is always the same as the Agent's action. Even in the adversarial setting, if for example, the Agent follows the Principal's recommendations, then this assumption will always be satisfied exactly --- but it can also be satisfied in many other ways. 
For any distribution $\mu$ over actions, let $U(\mu,p,y):= \EEs{a
\sim \mu}{U(a,p,y)}$ and $V(\mu,p,y):= \EEs{a
\sim \mu}{V(a,p,y)}$ denote the expected utilities when the action is sampled from $\mu$.


\begin{assumption}[No Secret Information]\label{asp:no-correlation} 
Consider any fixed sequence of forecasts $\pi_{1:T}$.
Given the sequence of policies $p_{1:T}$ and recommendations $r_{1:T}$ generated by the deployed mechanism, for any $(p,r)\in \cP\times \cA$, for any sequence of Agent's actions $a_{1:T}$ generated by $\cL$, let $\hat \mu_{p,r} = \frac{1}{n_{p,r}}\sum_{t:(p_t,r_t)= (p,r)} a_t$, where $n_{p,r} = |\{t :(p_t,r_t)= (p,r)\}|$, denote the empirical distribution of the Agent's actions during the subsequence of rounds in which $(p_t,r_t) = (p,r)$.
Then we assume that for all $(p,r)\in \cP\times \cA$,
\begin{align*}
\frac{1}{n_{p,r}}\EEs{a_{1:T}}{\abs{\sum_{t : (p_t,r_t) = (p,r)} (U(a_t, p, y_t) -U(\hat \mu_{p,r}, p, y_t))}}\leq \cO\left(\frac{1}{\sqrt{n_{p,r}}}\right)\,,\\
\frac{1}{n_{p,r}}\EEs{a_{1:T}}{\abs{\sum_{t: (p_t,r_t) = (p,r)} (V(a_t, p, y_t) -V(\hat \mu_{p,r}, p, y_t))}}\leq \cO\left(\frac{1}{\sqrt{n_{p,r}}}\right)\,.
\end{align*}

Similarly, given the sequence of policies $(p_0,\ldots,p_0)$ and recommendations $r^{p_0}_{1:T}$ generated by constant mechanism $\sigma^{p_0}$, for any $r\in \cA$, let $\hat \mu^{p_0}_{r} = \frac{1}{n^{p_0}_{r}}\sum_{t:r^{p_0}_t=r} a_t^{p_0}$, where $n^{p_0}_{r} = |\{t : r_t^{p_0}=r\}|$, denote the empirical distribution of the Agent's actions during the period's in which the recommendation $r_t^{p_0} = r$. 
Then we assume that, for all $p_0\in \cP_0$, for all $r\in \cA$,
\begin{align*}
    \frac{1}{n^{p_0}_{r}}\EEs{a^{p_0}_{1:T}}{\abs{\sum_{t:r^{p_0}_t=r} (U(a_t^{p_0}, p_0, y_t) - U(\hat \mu^{p_0}_{r}, p_0, y_t))}}\leq \cO\left(\frac{1}{\sqrt{n^{p_0}_{r}}}\right)\,,\\
     \frac{1}{n^{p_0}_{r}}\EEs{a^{p_0}_{1:T}}{\abs{\sum_{t:r^{p_0}_t=r} (V(a_t^{p_0}, p_0, y_t) - V(\hat \mu^{p_0}_{r}, p_0, y_t))}}\leq \cO\left(\frac{1}{\sqrt{n^{p_0}_{r}}}\right)\,.
\end{align*}
\end{assumption}

While the need for Assumption~\ref{asp:no-internal-reg} is clear (from the example provided earlier of an Agent who does not act to maximize his own payoffs, but instead behaves adversarially), the need for Assumption~\ref{asp:no-correlation} is less immediately clear. However it is indeed the case that Assumption~\ref{asp:no-internal-reg} is insufficient on its own.

\begin{restatable}[Necessity of Assumption~\ref{asp:no-correlation}]{proposition}{necessity}\label{lem:lb-easy}
There exists a simple linear contract setting where, for any Principal mechanism $\sigma$, one of the following must hold:
\begin{itemize}
    \item No learning algorithm $\cL^{*}$ can satisfy Assumption~\ref{asp:no-internal-reg} with $\eint = o(1)$ for all possible sequence of states $y_{1:T}\in \cY^T$.
    \item There exists a learning algorithm $\cL^{*}$ satisfying Assumption~\ref{asp:no-internal-reg} with $\eint = o(1)$ for all possible sequence of states $y_{1:T}\in \cY^T$ and a sequence of states $\bar y_{1:T} \in \cY^T$ for which $\sigma$ achieves non-vanishing regret, i.e., $\text{PR}(\sigma,\cL^{*},\bar y_{1:T}) = \Omega(1)$.
\end{itemize}
\end{restatable}

We will prove in Section~\ref{sec:linear} that in this same setting, if $\mathcal{L}$ satisfies Assumption~\ref{asp:no-internal-reg} and~\ref{asp:no-correlation}, there does exist a Principal mechanism which guarantees vanishing policy regret against $\mathcal{L}$. Therefore, Assumption~\ref{asp:no-correlation} plays an important role in our result. 
We will further discuss the necessity of the assumption in Section~\ref{sec:imposs}, where we also show that this impossibility result remains true even when Assumption~\ref{asp:no-internal-reg} is paired with an additional assumption which is in the same spirit of, but strictly weaker than, Assumption~\ref{asp:no-correlation}.

\section{Games with Stable Policy Oracles}\label{sec:stable}

In this section, we present a general no-policy-regret mechanism which applies in all settings where the Agent has access to a \emph{stable policy oracle}. A stable policy oracle is informally a way of producing or adjusting a policy to ensure that the Agent has only a single approximate best response given a particular fixed prior---or else that the Principal is almost indifferent between all of the Agent's approximate best responses.  What we will show is that the existence of such an oracle obviates the need for the kinds of very strong \emph{alignment} assumptions made in \cite{camara2020mechanisms}. In Section \ref{sec:oracles} we show that we in fact can implement such ``oracles'' in two very important cases: Principal Agent problems with \emph{linear} contracts, and binary state Bayesian Persuasion games, which allows us to obtain diminishing policy regret in these settings with minimal assumptions. In Section \ref{sec:general}, we extend our analysis to the general case (where Agents might unavoidably have multiple approximate best responses that the Principal is not indifferent between) --- there we will have to make the same kind of alignment assumption that is made in \cite{camara2020mechanisms}. 


Recall that we aim to resolve \emph{two} shortcomings of \cite{camara2020mechanisms}: the exponential computational and statistical complexity of producing calibrated forecasts, as well as the necessity to make strong alignment assumptions. To resolve the first issue,  rather than having the Principal produce calibrated forecasts, we have the Principal produce forecasts that satisfy a substantially weaker condition: unbiasedness subject to polynomially many ``events'', that will be eventually determined by the Principal's choice of policy and recommendation. Recent work of \cite{NRRX23} gives an algorithm for producing $d$-dimensional forecasts that satisfy this unbiasedness condition for polynomially in $d$ many events in time that is polynomial in $d$. Hence, this condition can be obtained with running time and bias bounds that scale only polynomially (rather than exponentially) in $|\cY|$. 

To resolve the second issue, rather than using the forecast $\pi_t$ directly as a prior and choosing the policy that would exactly optimize the Principal's payoff, we choose our policy using a stable policy oracle, defined below, which finds a policy that eliminates near ties: this will remove the necessity of an alignment assumption.  


First we define our notion of conditional bias.

\begin{definition}[Conditional Bias of Forecasts]\label{def:cal}
Let $\cE$ be a collection of ``events'', each defined by a function $E:\Delta(\cY)\rightarrow \{0,1\}$. 
For any sequence of states $y_{1:T}$, any sequence of forecasts $\pi_{1:T}$, and a collection of events $\cE$, we say $\pi_{1:T}$ has bias $\alpha$ conditional on $\cE$ if for all $E\in \cE$:
\begin{equation*}
    \frac{1}{T} \norm{\sum_{t=1}^T E(\pi_t) (\pi_t -y_t)}_1 \leq \alpha(E)\,.
\end{equation*}
\end{definition}

\cite{NRRX23} show how to efficiently make predictions obtaining low conditional bias against an adversarially chosen state sequence, for any polynomially sized collection of events:

\begin{theorem}[\cite{NRRX23}]\label{thm:forecast-bias}
    For any collection of events $\cE$ that can each be evaluated in polynomial time, there is a forecasting algorithm with per-round running time polynomial in $|\cY|$ and $|\cE|$ that produces forecasts $\pi_{1:T}$ such that for any (adversarially) chosen sequence of outcomes $y_{1:T}$, the expected bias conditional on $\cE$ is bounded by:
    $$\EEs{\pi_{1:T}}{\alpha(E)} \leq O\left(\frac{|\cY|\ln(|\cY||\cE|T)}{T} + \frac{|\cY|\sqrt{\ln(|\cY||\cE|T)|\{t : E(\pi_t) = 1|\}}}{T}\right) \leq O\left(\frac{|\cY|\sqrt{\ln(|\cY||\cE|T)}}{\sqrt{T}}\right)\,.$$
\end{theorem}

Next, we formalize our notion of a ``stable policy'' and a ``stable policy oracle''. Informally, what we need to deal with is  the possibility that the Agent has a range of approximate best responses with very different payoffs for the Principal. If this is the case, then the Agent could behave very differently given seemingly unimportant changes to the Principal's mechanism, leading to high policy regret. In many settings it is possible resolve this issue by adjusting the per-round policies a small amount to ensure a unique approximate best response---or else approximate indifference for the Principal between all of the Agent's approximate best responses.  

For any given prior distribution $\pi$, we say a policy $p$ is stable if choosing any action $a$ that deviates from the optimistic best response $\brr(p,\pi)$ results in either significantly lower Agent utility or a comparable level of utility for the Principal. Informally, this will mean that the Principal's payoff can be reliably predicted given the policy, assuming only that the Agent plays an approximate best response: any approximate best response will yield approximately the same payoff for the Principal. We emphasize that we will not \emph{assume} that policies are stable, but \emph{enforce it}. More specifically, for any prior distribution $\pi$, let $V(a,p,\pi) =\EEs{y\sim \pi}{V(a,p,y)}$ and $U(a,p,\pi) = \EEs{y\sim \pi}{U(a,p,y)}$ denote the expected utilities for the Principal and the Agent when the state $y$ is drawn from $\pi$. We define stable policies as follows.
\begin{definition}[Stable Policy]
    For any $\beta, \gamma>0$ and $\pi\in \Delta(\cY)$, a policy $p$ is $(\beta,\gamma)$-stable under $\pi$ if for all $a \neq \brr(p,\pi)$ in $\cA$, we have either
    \begin{equation*}
        U(a,p,\pi)\leq U(\brr(p,\pi),p,\pi) -\beta\,, 
    \end{equation*}
    or
    \begin{equation*}
        V(a,p,\pi)\geq V(\brr(p,\pi),p,\pi) -\gamma\,. 
    \end{equation*}
\end{definition}

Classically, in the common prior setting, both the Principal and the Agent best respond to (exactly) maximize their expected utilities. As discussed, in our setting, we have relaxed this best response assumption to a low-contextual-swap-regret assumption (Assumption~\ref{asp:no-internal-reg}), which is in fact a relaxation of an \emph{approximate} best response assumption --- i.e. it is satisfied in the commmon prior setting even if Agents do not exactly best respond, but merely approximately best respond. How shall we deal with this?

The Principal's utility would be maximized if the Agent were to choose amongst his approximate best responses so as to optimize for the Principal.  Specifically, let $\cB(p,\pi,\epsilon):= \{a\in \cA| U(a,p,\pi)\geq U(\brr(p,\pi),p,\pi) -\epsilon\}$ denote the set of all $\epsilon$-best responses for the Agent and let $\brr(p,\pi,\epsilon)$ denote the utility-maximizing action for the Principal, amongst the Agent's $\epsilon$-best responses to $p$, i.e.,
$$\brr(p,\pi,\epsilon) = \argmax_{a\in \cB(p,\pi,\epsilon)} V(a,p,\pi)\,.$$
Given any $\pi$, we say that a policy $p$ is an optimal stable policy under $\pi$ if $p$ is stable and implementing $p$ will lead to utility for the Principal that is comparable with her best achievable utility---i.e. the utility that the Principal could have obtained were the Agent guaranteed to  choose amongst his $\epsilon$-approximate best responses in the way that has highest payoff for the Principal.

\begin{definition}[Optimal Stable Policy Oracle]\label{def:stabilized}
    For a prior distribution $\pi$, we say that a policy $p$ is a $(c,\epsilon,\beta,\gamma)$-optimal stable policy under $\pi$ if 
    \begin{itemize}
        \item $p$ is $(\beta,\gamma)$-stable under $\pi$;
        \item and $V(a^*(p,\pi),p,\pi)\geq V(a^*(p_0,\pi,\epsilon),p_0,\pi)-c$ for all $p_0\in \cP_0$.
    \end{itemize} 
    An optimal stable policy oracle $\cO_{c,\epsilon,\beta,\gamma}: \Delta(\cY) \mapsto \cP_\cO$, given as input any prior $\pi$, outputs a $(c,\epsilon,\beta,\gamma)$-optimal stable policy in $\cP_\cO$ under $\pi$, where $\cP_\cO\subseteq \cP$ is the set of all possible output policies by the oracle.
\end{definition}

Intuitively, when $\beta > \epsilon$, then if the Agent can be assumed to play an $\epsilon$-best response to $\pi$ this is sufficient to guarantee that when the Principal deploys an optimal stable policy, she will obtain utility comparable to the utility she could have obtained assuming that the Agent were to best respond exactly while tiebreaking in the Principal's favor (i.e.  $V(\brr(p,\pi),p,\pi)$), and that, $V(\brr(p,\pi),p,\pi)$ is larger than the the utility achieved by any benchmark policy even if the Agent could have been assumed to optimistically respond. With such an oracle we can construct the mechanism described in Algorithm~\ref{alg:general-stable}, that guarantees the Principal no policy regret. Of course, we do \emph{not} assume that the Agent $\epsilon$-best responds to the forecast $\pi_t$ at round $t$ --- but as we will show, Assumptions \ref{asp:no-internal-reg} and \ref{asp:no-correlation} will be enough to make the analysis go through.

\begin{algorithm}[H]
    \caption{Principal's choice at round $t$}
    \label{alg:general-stable}
        \begin{algorithmic}[1]
            \State{\textbf{Input}: 
            Forecast $\pi_t\in \Delta(\cY)$
            }
            \State{Call the optimal stable policy oracle $\cO_{c,\epsilon,\beta,\gamma}$ to get a policy $p_t = \cO_{c,\epsilon,\beta,\gamma}(\pi_t)$ 
            }
        \end{algorithmic}
    \end{algorithm}

Let $\pit = \argmax_{p_0\in \cP_0} V(a^*(p_0,\pip_t,\epsilon),p_0,\pip_t)$ denote the policy that the Principal would pick if the Agent optimistically best responded to $(\pit,\pi_t)$ and $\rit = a^*(\pit,\pip_t,\epsilon)$ denote the corresponding optimistic $\epsilon$-best responding action.

\begin{restatable}{theorem}{thmstable}\label{thm:stable}

Define the following collections of events:
$$\cE_1 = \{\ind{(p_t,r_t) = (p,r)}\}_{p \in \cP_\cO, r \in \cA}\,, \ \ \  \ \ \cE_2 = \{\ind{(\pit,\rit) = (p,a)}\}_{p \in \cP_0, a \in \cA}\,,$$
$$\cE_3 = \{\ind{a^*(p_0,\pi_t) = a}\}_{p_0 \in \cP_0,a\in \cA}\,.$$
Let $\cE = \cE_1 \cup \cE_2 \cup \cE_3$, the union of these events. 
Assume that the Agent's learning algorithm $\cL$ satisfies the behavioral assumptions~\ref{asp:no-internal-reg} and \ref{asp:no-correlation}.
Given access to an optimal stable policy oracle $\cO_{c,\epsilon,\beta,\gamma}$, by running the forecasting algorithm from \cite{NRRX23} for events $\cE$ and the choice rule in Algorithm~\ref{alg:general-stable}, the Principal can achieve policy regret
\begin{align*}
&\text{PR}(\sigma^\dagger, \cL,y_{1:T}) \leq\tilde \cO\left(c +\gamma +\sqrt{\frac{\abs{\cP_0}\abs{\cA}}{T}} + \frac{\eint + \abs{\cY}\sqrt{\abs{\cP_\cO}\abs{\cA}/T}}{\beta} + \frac{\eint + \abs{\cY}\sqrt{\abs{\cA}/T}}{\epsilon}\right)\,,
\end{align*}
where $\tilde \cO$ ignores logarithmic factors in $T, \abs{\cY}, \abs{\cP_\cO},\abs{\cP_0}, \abs{\cA}$.
\end{restatable}
Note that we consider a fixed benchmark policy set, a fixed action space and a fixed state space. Hence we have that $\abs{\cP_0}$, $\abs{\cA}$ and $\abs{\cY}$ are all independent of $T$.
If we can construct an optimal stable policy oracle with $c,\gamma, \frac{\eint}{\beta},\frac{\sqrt{\abs{\cP_\cO}/T}}{\beta},\frac{\eint}{\epsilon},\frac{1}{\epsilon\sqrt{T}}  = o(1)$, then we can achieve vanishing regret $\text{PR}(\sigma^\dagger,\cL,y_{1:T}) = o(1)$. 
If the Agent is running a standard no-swap-regret algorithm, e.g. \citep{blum2007external}, the Agent can obtain swap regret $\eint = \cO(\sqrt{\abs{\cP_\cO}/T})$. We note that while it appears that the regret bound is decreasing in $\beta$ and $\epsilon$, when we actually construct optimal stable policy oracles in Section \ref{sec:oracles}, $c$ will grow with $\beta$ and $\epsilon$, and so there will be a tradeoff to manage.
The proof the theorem is deferred to Section~\ref{sec:stable-proof}.

\section{Constructing Stable Policy Oracles}
\label{sec:oracles}
In this section, we instantiate the general algorithm we derived in Section \ref{sec:stable} by constructing efficient stable policy oracles for two important special cases of the general Principal-Agent setting: the \emph{linear contracting} problem and the \emph{Bayesian Persuasion} problem in which there is an unknown binary state of nature. Linear contracting in particular has been focal in the contract theory literature due to the robustness and practical ubiquity of linear contracts \cite{carroll2015robustness,dutting19} --- and much of the recent computational and learning theoretic work on contract theory has focused exclusively or primarily on linear contracts. Binary state Bayesian Persuasion is a canonical case in Bayesian Persuasion, encompassing various intriguing scenarios, such as the FDA approval example.
In the following, we will introduce these two problems and construct efficient stable policy oracles for them.

\subsection{Linear Contracts}
\label{sec:linear}
 In the contract setting, there is a finite outcome space $\mathbb{O} = \{o_1,\ldots,o_m\}$ (e.g., \{success, failure\}).
A contract $p: \mathbb{O} \mapsto [0,1]$ is a mapping from outcomes to payments and the Principal commits to pay the Agent a specified amount $p(o)$ if the outcome is $o$.
The Principal provides a contract to the Agent, and the Agent then decides to take an action (e.g., working or shirking). 
The Agent's action and the state of nature (e.g., hard job or easy job) together determine the outcome through a mapping $o: \cA\times \cY\mapsto \mathbb{O}$. 
Different outcomes will lead to different outcome values.
The Agent incurs different costs by taking different actions.
Then the utility of the Principal is the difference between the the outcome value and the payment to the Agent. The utility of the Agent is the difference between the payment and the cost of taking the action.
More specifically, 
let $v: \mathbb{O}\mapsto [0,1]$ denote the value function of outcomes and $c:\cA\mapsto [0,1]$ denote the cost function for the Agent.
When the Principal offers contract $p$, the Agent takes action $a$, and the outcome is $o$, then the Principal's utility is $v(o) - p(o)$ and Agent's utility is $p(o) - c(a)$.

Our focus will be on linear contracts, a particularly simple and widespread type of contract which provides the Agent with a constant fraction of the outcome value. Linear contracts are focal in the contract theory literature in part because of their robustness properties \citep{carroll2015robustness,dutting19}.

\begin{definition}[Linear contract]
For a linear contract parameterized by $p \in [0,1]$, the Principal pays the Agent a $p$-fraction of the value, i.e., $p \cdot v(o)$ when the outcome is $o$. Hence, we use this fraction to represent the linear contract and write the policy space as $\cP = [0,1]$, the set of all parameters that can specify a linear contract. 
\end{definition}
For any linear contract $p\in \cP$, action $a\in \cA$ and state of nature $y\in \cY$, the Principal's utility is 
\begin{align*}
    V(a,p,y) = v(o(a,y)) - p\cdot v(o(a,y)) = (1-p)v(o(a,y))\,,
\end{align*}
and the Agent's utility is 
\begin{align*}
    U(a,p,y) = p\cdot v(o(a,y)) - c(a)\,.
\end{align*}
We consider a finite action space and assume that the costs are different for each action. Hence the minimum gap between the costs is positive, and we denote it by:
\begin{align*}
    \Delta_c = \min_{a_1, a_2\in \cA: a_1\neq a_2}\abs{c(a_1)-c(a_2)}>0\,.
\end{align*}
For any action $a\in \cA$ and prior $\pi$, let $$f(\pi,a) := \EEs{y \sim \pi}{v(o(a,y))} $$ denote the expected outcome value when the Agent takes action $a$ and the state of nature is drawn from the prior distribution $\pi$. Then the Principal's utility under $\pi$ can be written as 
\begin{align}
    V(a,p,\pi) = \EEs{y\sim \pi}{(1-p) \cdot v(o(a,y))} = (1-p) f(\pi,a)\,,\label{eq:expectedV}
\end{align}
and the Agent's utility can be written as 
\begin{align}
    U(a,p,\pi) =\EEs{y\sim \pi}{p \cdot v(o(a,y))}- c(a)= p f(\pi,a) - c(a)\,.\label{eq:expectedU}
\end{align}

Then we can construct an optimal stable policy oracle as follows.
Given any prior $\pi$, we initially identify the policy $p^\text{optimistic}$ that maximizes the Principal's utility assuming that the Agent optimistically approximately best responds---i.e. chooses the action amongst all of his \emph{approximate} best responses that maximizes the Principal's utility. However, $p^\text{optimistic}$ will generally be unstable, and thus the Agent may not actually  optimistically respond if we were to implement $p^\text{optimistic}$. 
The subsequent step involves stabilizing $p^\text{optimistic}$ by incrementally adjusting the contract until it becomes stable. 
It turns out that a small increase in $p^\text{optimistic}$ allows us to obtain a stable policy. Since this policy is close to $p^\text{optimistic}$, the Principal's utility remains comparable to the performance of any benchmark policy when the Agent optimistically approximately best responds---even though the stabilization means we no longer need to \emph{assume} that the Agent will optimistically best respond. Finally, recall that the regret guarantee in Theorem~\ref{thm:stable} depends on the cardinality of the output policy space. Consequently, we will have to discretize the policy space and provide a discretized stable policy.
Let $\cP_\delta = \{0, \delta, 2\delta, \ldots, \floor{\frac{1}{\delta}}\delta\}$ denote a $\delta$-cover of the linear contract space for
some $\delta =o(1)$. We construct the following optimal stable policy oracle with output space $\cP_\cO = \cP_\delta$ so that $\abs{\cP_\cO} = \floor{\frac{1}{\delta}}+1$.

\begin{algorithm}[H]
    \caption{Optimal Stable Policy Oracle for Linear Contracts}
    \label{alg:linear-oracle}
        \begin{algorithmic}[1]
            \State{\textbf{Parameters}: stability parameter $\beta$, discretization parameter $\delta$}
            \State{\textbf{Input}: prior distribution $\pi$}
            \State{Compute $p^\text{optimistic} = \argmax_{p\in \cP_0} \max_{a\in \cB(p,\pi,\frac{\Delta_{c}\beta}{2})} V(a,p,\pip)$}
            \State{\textbf{Output}: $$p(\pi) = \min\left(\left\{p\in \cP_\delta|p\geq p^\text{optimistic}, p \text{ is } \left(\frac{\Delta_c \beta}{2},0\right)\text{-stable under }\pi \right\}\cup \{1\}\right)$$}
        \end{algorithmic}
    \end{algorithm}

\begin{theorem}[Optimal Stable Policy Oracle for Linear Contracts]\label{thm:linear}
    Algorithm~\ref{alg:linear-oracle} is a $(|\cA|(\beta+\delta),\frac{\Delta_{c}\beta}{2},\frac{\Delta_{c}\beta}{2},0)$-optimal stable policy oracle with $\abs{\cP_\cO} = \cO(\frac{1}{\delta})$. By combining with Theorem~\ref{thm:stable} and setting $\beta = T^{-\frac{1}{4}}$ and $\delta = \sqrt{\beta}$, we can achieve Principal's regret: 
    $$\text{PR}(\sigma^\dagger, \cL,y_{1:T}) = \tilde\cO\left(T^{-\frac{1}{8}}\right)\,,$$
    when the Agent obtains swap regret $\eint = \cO(\sqrt{\abs{\cP_\cO}/T})$.
\end{theorem}

\begin{proof}
According to the definition of optimal stable policy oracle (Definition~\ref{def:stabilized}), the proof of the theorem follows directly from Lemma~\ref{lem:stable_contract} and Lemma~\ref{lem:payoff_close}.
\begin{lemma}\label{lmm:linear-stable}
   For any prior $\pi$, the policy $p(\pi)$ returned by Algorithm~\ref{alg:linear-oracle}, is a $(\frac{\beta \Delta_c}{2},0)$-stable policy under $\pi$ and satisfies that $p(\pi) \leq p^{\textit{optimistic}} + \abs{\cA}(\beta + \delta)$ . \label{lem:stable_contract}
\end{lemma}

\begin{lemma}\label{lmm:linear-opt}
    For any prior $\pi$, the policy $p(\pi)$ returned by Algorithm~\ref{alg:linear-oracle} satisfies that
    $$V(a^*(p(\pi),\pi),p(\pi),\pi)\geq V(a^*(p_0,\pi,\frac{\beta \Delta_c}{2}),p_0,\pi) - |\cA|(\beta + \delta)$$
    for all $p_0\in \cP_0$. \label{lem:payoff_close}
\end{lemma}

Lemma~\ref{lem:stable_contract} shows that for any $\pi$, the returned  linear contract $p(\pi)$ is $(\frac{\Delta_c \beta}{2},0)$-stable and is not much larger than $p^\text{optimistic}$.
This implies that the Principal will not pay a much larger fraction of her value under $p(\pi)$ than she would under $p^\text{optimistic}$. In Lemma~\ref{lem:payoff_close}, we prove that the Principal's utility under $p(\pi)$ is comparable to her utility under any benchmark contract.

\begin{proof}[Proof of Lemma~\ref{lmm:linear-stable}]
The intuition for this stability result is that, for any policy returned, either the Agent has a unique best response that gets him a payoff $\frac{\beta \Delta_{c}}{2}$ higher than all other actions, or the Principal is completely indifferent between what actions the Agent selects. We first show that there must be a policy with such a unique best response in the interval $[p^\text{optimistic}, p^\text{optimistic} + \abs{\cA}(\beta + \delta)]$, as long as this interval lies fully within the linear contract policy space of $[0,1]$, i.e., $p^\text{optimistic} + \abs{\cA}(\beta + \delta)\leq 1$. To do this, we take advantage of the fact that for a fixed $\pi$, there are a bounded number of policies which induce ties between actions (Lemma~\ref{lem:ties}), and for all policies far enough away from these policies, the Agent actions are well-separated (Lemma~\ref{lem:gap}). When $p^\text{optimistic}$ is larger than $1 - |\cA|(\beta + \delta)$, we no longer have this guarantee--however, if the Principal does not return a $(\frac{\beta \Delta_c}{2},0)$-stable policy in this case, she will return $p(\pi)= 1$, which is still close to $p^\text{optimistic}$, and furthermore gets the Principal a payoff of $0$ regardless of what action the Agent takes, leading her to be indifferent to the Agent's action.

\begin{restatable}{lemma}{lemties}\label{lem:ties}
For any $\pi$, there are at most $|\cA|-1$ linear contracts resulting in more than one best response for the Agent, i.e.: 
$$\abs{\{p\in \cP|\cB(p,\pi,0)| >1\}}\leq |\cA|-1\,.$$ 
\end{restatable}

\begin{restatable}{lemma}{lemgap}\label{lem:gap}

For any prior $\pi$ and any $\bar p\in [0,1]$, if $a^*$ is an Agent's best response to both $(\bar p - \beta, \pi)$, and  $(\bar p + \beta, \pi)$, then $U(a^*,\bar p,\pi) \geq U(a,\bar p,\pi) + \Delta_c \cdot \beta$, for all actions $a \neq a^*$.  
\end{restatable}
Now we start formally proving Lemma~\ref{lmm:linear-stable}. There are two cases:
\begin{itemize}
\item $p^\text{optimistic} \leq 1 - |\cA|(\beta + \delta)$. Then, let us consider the policies in the range $[p^\text{optimistic}, p^\text{optimistic} + \abs{\cA}(\beta + \delta)]$ for which the Agent has more than one optimal response. Call this set $s$. By Lemma~\ref{lem:ties}, we have $|s| \leq |\cA|-1$. Note that, by the definition of $s$, for any given $i\in [|s|]$, all policies $p \in (s_{i}, s_{i+1})$ (where $s_i$ is the $i$-th smallest element in $s$) must lead to a unique best response action for the Agent, and must lead to the same best response as each other by the continuity of the Agent's utility with respect to the Principal policy.

Now, let's augment $s$ with the endpoints of the interval by letting $s' = \{p^\text{optimistic}\} \cup s \cup \{p^\text{optimistic} + \abs{\cA}(\beta\ + \delta)\}$.
For any $i\in [|s'|]$, let $s'_i$ denote the $i$-th smallest element in $s'$.
We will lower bound the largest gap between any two neighboring policies in $s'$.
\begin{align*}
 \argmax_{i \in [|s'|-1]}(s'_{i+1}-s'_{i}) & \geq \frac{s'_{|s'|} - s'_{1}}{|s'|-1} = \frac{|\cA|(\beta + \delta)}{|s'|-1} \geq  \frac{|\cA|(\beta + \delta)}{|s|+1}\geq \beta + \delta \,,
\end{align*}
where the last inequality applies Lemma~\ref{lem:ties}.

Hence, there exists an $i\in [|s'|-1]$ such that $s'_{i+1}-s'_{i} \geq \beta + \delta$. Now, consider any policy $p \in [s'_{i} + \frac{\beta}{2}, s'_{i+1} - \frac{\beta}{2}]$. By Lemma~\ref{lem:gap}, we have $U(a^{*}(p,\pi),p,\pi) \geq U(a,p,\pi) + \frac{\Delta_c \beta}{2}$ for all $a \neq a^{*}(p,\pi)$. Therefore, every policy in this range is $(\frac{\Delta_c \beta}{2},0)$-stable under $\pi$. As this range is of size at least $\delta$, there must be at least one policy $p \in \cP_{\delta}$ in the range $[p^\text{optimistic}, p^\text{optimistic} + \abs{\cA}(\beta + \delta)]$ that is $(\frac{\Delta_c \beta}{2},0)$-stable under $\pi$. By the definition of the algorithm, the returned $p(\pi)$ is $(\frac{\Delta_c \beta}{2},0)$-stable under $\pi$ and is in the range $[p^\text{optimistic}, p^\text{optimistic} + \abs{\cA}\beta + \delta]$.

\item $p^\text{optimistic} \geq 1 - |\cA|(\beta + \delta)$. Then the returned policy must be in the range $[p^\text{optimistic}, p^\text{optimistic} + \abs{\cA}(\beta + \delta)]$. If some $p(\pi) < 1$ is returned, by the definition of the algorithm, it will be $(\frac{\beta \Delta_{c}}{2},0)$-stable. Otherwise, the algorithm returns $p(\pi) = 1$, and we have that 
\begin{align*}
V(a^{*}(p(\pi),\pi),p(\pi),\pi)  
& = (1-p(\pi)) \cdot f(a^{*}(p(\pi),\pi),\pi) = 0  \leq V(a,p(\pi),\pi)\,,
\end{align*}
for any $a \in \A$. Thus, in this case we have that $ V(a,p(\pi),\pi)\geq V(\brr(p(\pi),\pi),p(\pi),\pi) - 0$, and thus the policy is also $(\frac{\beta \Delta_{c}}{2},0)$-stable. Furthermore, in this case the returned $p(\pi)$ is also in the range $[p^\text{optimistic}, p^\text{optimistic} + \abs{\cA}(\beta + \delta)]$.
\end{itemize}
This completes the proof of Lemma~\ref{lmm:linear-stable}.
\end{proof}

Now we move on to prove Lemma~\ref{lmm:linear-opt}.
For this part, we must upper bound the difference between the Principal's utility under the policy $p=p(\pi)$ returned by Algorithm~\ref{alg:linear-oracle} and her utility under the best benchmark policy $p_{0}$. 
To do this, 
we compare the utility of the Principal under $p^\text{optimistic}$ to her utility under $p$, taking advantage of the fact that $p$ is not much larger than $p^\text{optimistic}$. We crucially make use of the monotone relationship between $p$ and $f(\pi, a^{*}(p,\pi,\epsilon))$ for linear contracts (Lemma~\ref{lem:monotonicity}).

\begin{restatable}{lemma}{lemmonotonicity}
    \label{lem:monotonicity}
For any two linear contracts $p_{1} \geq p_{2}$, $$\max_{a \in \cB(p_{1},\pi,\epsilon)}f(\pi,a) \geq \max_{a \in \cB(p_{2},\pi,\epsilon)}f(\pi,a)$$ for all $\pi$ and all $\epsilon \geq 0$.
\end{restatable}

\begin{proof}[Proof of Lemma~\ref{lmm:linear-opt}]
We consider two cases: $p(\pi)<1$ and $p(\pi)=1$.
\begin{itemize}
    \item $p(\pi)< 1$. Since $p(\pi)$ is $(\frac{\Delta_{c}\beta}{2},0)$-stable according to Lemma~\ref{lmm:linear-stable}, then for all $a \neq a^{*}(p,\pi)$, either $ U(a,p,\pi) \leq  U(a^{*}(p,\pi),p,\pi) - \frac{\Delta_c \beta}{2}$  or $V(a,p(\pi),\pi) = V(a^{*}(p(\pi),\pi),p(\pi),\pi)$. 
    For all $a$ with $V(a,p(\pi),\pi) = V(a^{*}(p(\pi),\pi),p(\pi),\pi)$, we have $f(\pi,a) = \frac{V(a,p(\pi),\pi)}{1-p(\pi)} = f(\pi,a^*(p(\pi),\pi))$. Therefore, we have 
    \begin{equation}\label{eq:maxball}
        \max_{a \in \mathcal{B}(p(\pi),\pi,\frac{\Delta_{c}\beta}{2})}f(\pi,a) = f(\pi,a^{*}(p(\pi),\pi))\,.
    \end{equation}
\begin{align*}
 &\max_{p_0 \in \cP_{0}}V(a^*(p_0,\pi,\frac{\Delta_{c}\beta}{2}), p_0,\pi) \\
 = &  V(a^*(p^\text{optimistic},\pi,\frac{\Delta_{c}\beta}{2}),p^\text{optimistic},\pi) \tag{ Definition of $p^\text{optimistic}$} \\
= & \max_{a \in \mathcal{B}(p^\text{optimistic},\pi,\frac{\Delta_{c}\beta}{2})}V(a,p^\text{optimistic},\pi)  \\ 
= & (1 - p^\text{optimistic})\max_{a \in \mathcal{B}(p^\text{optimistic},\pi,\frac{\Delta_{c}\beta}{2})}f(\pi,a)\tag{Applying Eq~\eqref{eq:expectedV}}\\ 
\leq & (1 - p^\text{optimistic})\max_{a \in \mathcal{B}(p(\pi),\pi,\frac{\Delta_{c}\beta}{2})}f(\pi,a) \tag{Applying Lemma~\ref{lem:monotonicity}, as $p(\pi) \geq p^\text{optimistic}$}\\ 
= & (1 - p^\text{optimistic})f(\pi,a^{*}(p(\pi),\pi)) \tag{Applying Eq~\eqref{eq:maxball}}\\
\leq & (1 - p(\pi) + |\cA|(\beta + \delta))f(\pi,a^{*}(p(\pi),\pi))) \tag{Applying the gap condition in Lemma~\ref{lem:stable_contract}}\\
= & V(a^*(p,\pi),p,\pi) +|\cA|(\beta + \delta) \cdot f(a^{*}(p,\pi),p,\pi)  \\
\leq & V(a^*(p,\pi),p,\pi) + |\cA|(\beta + \delta)\,.
\end{align*}
\item $p(\pi) = 1$. In this case, we have
$p^\text{optimistic}\geq 1-|\cA|(\beta+\delta)$.
Then we have
\begin{align*}
    &\max_{p_0 \in \cP_{0}}V(a^*(p_0,\pi,\frac{\Delta_{c}\beta}{2}), p_0,\pi)\\
    =&V(a^*(p^\text{optimistic},\pi,\frac{\Delta_{c}\beta}{2}),p^\text{optimistic},\pi) \\
    \leq & 1-p^\text{optimistic}\\
    \leq &|\cA|(\beta+\delta) \leq V(a^*(p,\pi),p,\pi) + |\cA|(\beta+\delta)
\end{align*}
\end{itemize}

This completes the proof of Lemma 2.
\end{proof}

By Lemmas~\ref{lem:stable_contract} and~\ref{lem:payoff_close}, we get that, for any prior $\pi$, the policy $p(\pi)$ returned by Algorithm~\ref{alg:linear-oracle} is a $(\frac{\beta \Delta_{c}}{2},0)$-stable policy under $\pi$, and furthermore that $V(a^*(p(\pi),\pi),p(\pi),\pi)\geq V(a^*(p_0,\pi,\Delta_{c}\frac{\beta}{2}),p_0,\pi)-\delta - |\cA|\beta$ for all $p_0\in \cP_0$. Putting these together proves that Algorithm~\ref{alg:linear-oracle} is a $(\delta + \beta|\cA|,\frac{\Delta_{c}\beta}{2},\frac{\Delta_{c}\beta}{2},0)$-optimal stable policy oracle.
\end{proof}

\subsection{Bayesian Persuasion}\label{sec:bayes}
Bayesian Persuasion is another important special case of the general Principal Agent problem that is quite different from the linear contracting case. 
In Bayesian Persuasion~\citep{kamenica2011bayesian}, Sender (the Principal) wishes to persuade Receiver (the Agent), to choose a particular action: but by controlling the information structure used to communicate with Receiver, rather than by making monetary payments.
For example, a traditional example is a prosecutor (Sender) who tries to convince a judge (Receiver) that a defendant is guilty.
\subsubsection{Fundamentals of Bayesian Persuasion}
A policy in Bayesian Persuasion is a signal scheme, which consists of a signal space $\Sigma$ and a family of distributions $\{\sig(\cdot|y)\in \Delta (\Sigma)\}_{y\in \cY}$ mapping ``states of nature'' $\cY$ to ``signals'' $\Sigma$.
Sender selects and sends a signal scheme to Receiver.
After observing the signal scheme and a signal realization $\sigma\sim \sig(\cdot|y)$ as a function of the underlying state of nature $y$, Receiver selects her strategy $s$ from a strategy space $\cS$.
In other words, after observing the signal scheme, Receiver selects an action $a: \Sigma\mapsto \cS$, which maps signals to strategies.
Both Sender's utility $v(s,y)\in [0,1]$ and Receiver's utility $u(s,y)\in [0,1]$ are functions of Receiver's strategy $s\in \cS$ and the state of nature $y\in \cY$.
For any policy $p$ and any action $a$, the Principal's utility is
\begin{equation*}
    V(a,p,y)=\EEs{\sigma\sim p(\cdot|y)}{v(a(\sigma), y)}\,,
\end{equation*}
and the Agent's utility is
\begin{equation*}
    U(a,p,y)=\EEs{\sigma\sim p(\cdot|y)}{u(a(\sigma),y)}\,.
\end{equation*}

In the common prior setting, there exists a common prior distribution $\pi$ over states of nature $\cY$.
To maximize the expected utility, the Agent will form his posterior distribution conditional on the signal $\pi_\sigma= \pi(y|\sigma)$ using Bayes's rule and best respond by selecting strategy $\argmax_{s\in \cS} \EEs{y\sim \pi_{\sigma}}{u(s,y)}$.
Consider the traditional example of a  prosecutor and a judge. The state space is $\cY$ = \{Innocent, Guilty\} and the strategy space is $\cS$ = \{Convict, Acquit\}.
The judge has 0-1 utility and prefers to convict if the defendant is guilty and acquit if the defendant is innocent.
Regardless of the state, the prosecutor’s utility is 1 following a conviction and 0 following an acquittal.
Consider the case that $\pi(\text{Guilty}) = 0.3$. 
If there is no communication, the judge will  always acquits because guilt is less likely than innocence under his prior. 
However, the prosecutor can construct the following signal scheme to improve her utility.
\begin{align*}
   &p(\text{i} | \text{Innocent}) = \frac{4}{7}, \quad p(\text{g} | \text{Innocent}) = \frac{3}{7}\,,\\
   &p(\text{i} | \text{Guilty}) = 0, \quad p(\text{g} | \text{Guilty}) = 1\,.
\end{align*}
The posterior distribution of observing signal $g$ is $\pi_{g}(\text{Guilty}) = \pi_{g}(\text{Innocent}) = 0.5$ and the judge will convict when observing signal $g$. This leads the judge to convict with probability $0.6$.

A signal scheme is said to be ``straightforward'' if the signal space $\Sigma= \cS$ and Receiver's best responding strategy equals the signal realization. In other words, a straightforward signal scheme simply tells the receiver what action to take, and it is in the reciever's interest to comply. \cite{kamenica2011bayesian} shows that the optimal value can be achieved by straightforward signal schemes. 
Hence, we restrict to straightforward signal schemes in the following and let $\cP$ be the space of all straightforward signal schemes.

A common special case of Bayesian Persuasion is that both the states of nature and the strategies are real-valued, 
and Sender's preferences over Receiver's strategies do not depend on the nature state $y$. 
Hence, Sender's utility can be written as a function of Receiver's strategy, i.e.,
\begin{equation*}
    v(s, y) = v(s)\,.
\end{equation*}
We consider a simpler but very common case where the number of states of nature is $2$, i.e., $\abs{\cY} = 2$.  In the example of prosecutor, the states are \{Innocent, Guilty\}. In the context of drug trials, a drug company (the Principal) seeks approval from FDA (the Agent) for a new drug, and the states are \{Effective, Ineffective\}. We remark that in this special case, our improvement over \cite{camara2020mechanisms} is in the removal of the Alignment assumption ---  since the state space is binary, our general efficiency improvements in terms of the cardinality of the state space are not relevant.
Without loss of generality, we assume that $\cY = \{0,1\}$ and $\cS\subset [0,1]$. We consider finite discrete strategy space $\cS$.
For any $\mu\in [0,1]$, let $u(s,\mu) = \EEs{y\sim\Ber(\mu)}{u(s,y)}$ denote the expected Agent's utility of choosing $s$ when $y$ is drawn from $\Ber(\mu)$.  We will assume (without loss of generality) that every strategy is a best response for the Agent for \emph{some} prior distribution (otherwise we can remove such a strategy from $\cS$):
\begin{assumption}\label{asp-bayes:interval}
We assume that for all $s\in \cS$, there exists a $\mu\in [0,1]$ such that $u(s,\mu)>u(s',\mu)$ for all $s'\neq s$ in $\cS$. 
\end{assumption}


Since we only focus on Bernoulli distributions, when we refer to $\mu$ as a belief/prior, we are using this as shorthand for the distribution $\Ber(\mu)$.
For any $\mu\in [0,1]$, let
$S^*(\mu) = \argmax_{s\in \cS}u(s,\mu)$
denote the set of optimal strategies under prior $\mu$ and let 
$$s^*(\mu) = \argmax_{s\in S^*(\mu)}v(s)$$ 
denote the optimal strategy breaking ties by maximizing the Principal's utility.

\begin{figure}[!t]
    \centering
    \begin{subfigure}[b]{0.5\textwidth}
    \includegraphics[width = \textwidth]{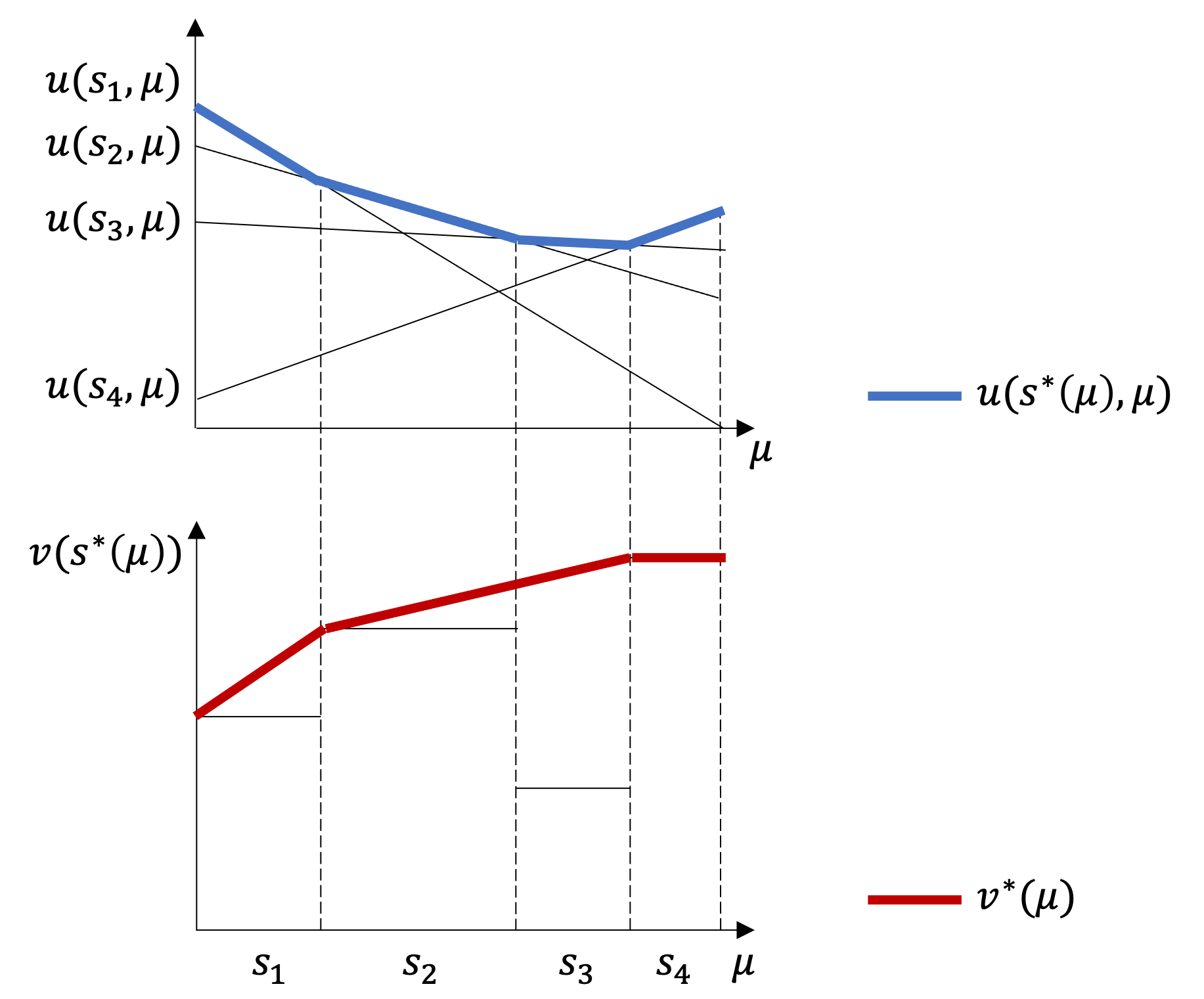}
    \caption{For any $s\in \cS$, $u(s,\mu)$ is a linear function of $\mu$. $u(s^*(\mu),\mu)$ is the maximum over all these linear functions. $v(s^*(\mu))$ is a piecewise constant function. $v^*(\mu)$ is defined in Eq~\eqref{eq:opt-bayes}.}
    \label{fig:bayes}
    \end{subfigure}
    \hfill
    \begin{subfigure}[b]{0.45\textwidth}
    \includegraphics[width = \textwidth]{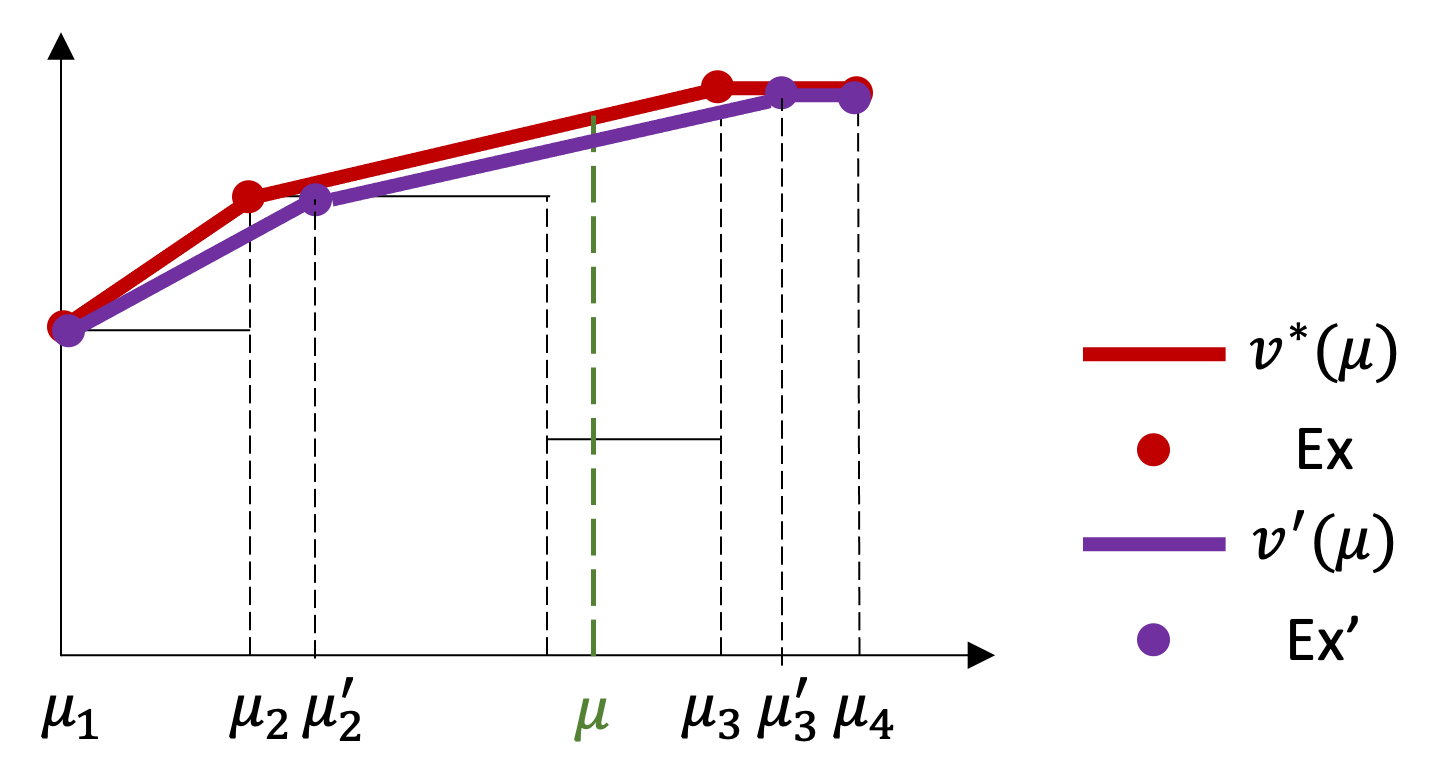}
    \caption{Illustration of $\text{Ex}, \text{Ex}', v^*(\mu), v'(\mu)$. In this figure, the optimal achievable value $v^*(\mu)$ is achieved by the convex combination of $\mu_2$ and $\mu_3$. The value achieved by our scheme $v'(\mu)$ is attained by convex combination of $\mu_2'$ and $\mu_3'$.  $v'(\mu)$ is very close to $v^*(\mu)$.}
    \label{fig:bayes2}
    \end{subfigure}
    \caption{Illustration of utilities in Bayesian Persuasion.}
\end{figure}


As depicted in Fig~\ref{fig:bayes}, for any $s\in \cS$, $u(s,\mu)$ is linear in $\mu$ with the absolute value of the slope $\abs{\partial u(s,\cdot)} \leq 1$ since the utilities are in $[0,1]$. It is easy to check that for any $\mu<\mu'\in [0,1]$, if $s$ is an optimal strategy for both $\Ber(\mu)$ and $\Ber(\mu')$, then for any $\mu''\in [\mu,\mu']$, $s$ is also an optimal strategy for $\Ber(\mu'')$.
Hence, $[0,1]$ is divided into $n$ closed intervals $(S_1,\ldots, S_{n})$ for some $n\leq |\cS|$, such that all $\mu\in S_i$ have a single shared optimal strategy, denoted by $s_i$.
\begin{restatable}{lemma}{implicationgap}\label{lmm:implication-gap-asp}
    Under Assumption~\ref{asp-bayes:interval}, we have the following observations:
    \begin{itemize}
        \item Each strategy in $\cS$ corresponds to one interval in $(S_1,\ldots, S_{n})$. In other words, we have $\cS = \{s_1, s_2,\ldots, s_n\}$ and $n=\abs{\cS}$.
        \item There exists a positive constant $C>0$ such that the length of every interval in $\{S_1,\ldots, S_{n}\}$ is lower bounded by $C$. For every interval $S_i$, for any $\mu$ inside $S_i$ (not on the edge), $s_i$ is the unique optimal strategy under prior $\mu$.
        \item There exists a positive constant $c_1>0$ such that for any two different strategies $s,s'$, the difference between the utility slopes, $\abs{\partial u(s,\cdot)-\partial u(s',\cdot)}$, is bounded below by $c_1$. 
    \end{itemize}    
\end{restatable}



Then the Agent's utility $u(s^*(\mu),\mu)$, given that he selects the optimal strategy $s^*(\mu)$, as a function of $\mu$, is taking a maximum over the set of linear functions $\{u(s,\mu)|s\in \cS\}$, as depicted in blue in Fig~\ref{fig:bayes}.
The Principal's utility $v(s^*(\mu))$ given that the Agent selects the optimal strategy $s^*(\mu)$, is a piecewise constant function since for all $\mu\in S_i$ (except for the boundary of $S_i$), $v(s^*(\mu)) = v(s_i)$.

Given any prior $\pi = \Ber(\mu)$, it is easy to see that a signal scheme induces a distribution over posteriors, $\pi(y|s)$ for all $s\in \cS$. 
The reverse is true as well: any distribution over posteriors that is consistent with our prior corresponds to a signal scheme. Given a  distribution of posteriors $\{(\tau_i,\mu_i)|i\in [n]\}$ with $\tau_i\geq 0$, $\sum_{i=1}^n \tau_i = 1$, we call the distribution Bayes-plausible if the expected posterior equals the prior, i.e., $\sum_{i}\tau_i \mu_i = \mu$.
Given a Bayes-plausible distribution of posteriors, we can recover the corresponding signal scheme  $p(s_i|y) = \frac{\tau_i\pi(y|s_i)}{\pi(y)}$ by Bayes' rule, where $s_i = s^*(\mu_i)$.
More explicitly, we have
\begin{align}
    p(s_i|y=1) = \frac{\tau_i \cdot \mu_i}{\mu}\,, \qquad p(s_i|y=0) = \frac{\tau_i \cdot (1-\mu_i)}{1-\mu}\,.\label{eq:signal}
\end{align}
Therefore, when given a prior $\mu$, selecting a signal scheme is equivalent  to selecting a Bayes-plausible distribution of posteriors. In the following, we choose a Bayes-plausible distribution of posteriors to represent a signal scheme. 

In Bayesian Persuasion, given any policy $p$ and prior $\pi=\Ber(\mu)$, the optimistic best response $a^*(p,\mu)$ for the Agent is selecting the optimal strategy $s^*(\pi(y|s))$ under the posterior $\pi(y|s)$ when observing signal $s$. Given prior $\pi=\Ber(\mu)$, the optimal achievable Principal's utility is defined as the maximum utility given that the Agent always best responds optimistically, i.e., $v^*(\mu) := \argmax_{p\in \cP}V(a^*(p,\mu), p, \mu)$.
\begin{lemma}[\cite{kamenica2011bayesian}]
    The optimal achievable Principal's utility is the concave closure of the convex hull of $(\mu,v(s^*(\mu)))$:
\begin{equation}\label{eq:opt-bayes}
    v^*(\mu) = \sup\{z|(\mu,z)\in \conv(v)\}\,,
\end{equation}
where $\conv(v)$ is the convex hull of $\{(\mu, v(s^*(\mu)))|\mu\in [0,1]\}$. 
\end{lemma}
The optimal achievable value $v^*(\mu)$ given the prior $\mu$ is depicted in red in Fig~\ref{fig:bayes}. Let 
$\textrm{Ex}= \{(\mu_{1}, v(s^*(\mu_{1}))),\ldots,(\mu_{{K}}, v(s^*(\mu_{{K}})))\}$ 
denote all extreme points of $\conv(v)$ on the concave closure, where $\mu_{1}=0$ and $\mu_{{K}} = 1$ for notation convenience.
By \cite{kamenica2011bayesian}, there exists two points in $\textrm{Ex}$ 
such that $(\mu, v^*(\mu))$ is represented as the convex combination of them. This defines a Bayes-plausible distribution of posteriors, where $\mu_{j}$ is the posterior given signal $s^*(\mu_{j})$ and the weight on $(\mu_{j}, v(s^*(\mu_{j})))$ is the probability mass assigned to $\mu_{j}$. The optimal scheme is the one which induces this distribution of posteriors. Note that all these $\mu_{j}$'s lie on the boundaries of the intervals $\{S_1,\ldots,S_n\}$.



\subsubsection{Optimal Stable Policy Oracle Construction}
Now we are ready to describe how to construct a stable policy oracle in Bayesian Persuasion based on the above optimal scheme.
The reader may notice that the above optimal scheme is not stable since each possible posterior $\mu_{j}$ lies on the edge of intervals and there will be two optimal strategies under $\mu_{j}$, which might lead to different Principal utilities. Hence, we need first to stabilize the optimal scheme. Besides, recall that the Principal's regret in Theorem~\ref{thm:stable} depends on the cardinality  $|\cP_\cO|$ of the output policy space, and so it is not enough to be able to construct near optimal stable policies --- we need to be able to construct near optimal stable policies that are always members of a small discrete set. Thus, as a second part of our construction we need to discretize the output space. We will introduce the stabilization step in the following and defer the discretization step to Appendix~\ref{app:bayes}.

\paragraph{Stabilization of the optimal scheme}
To stabilize the scheme, we need to make sure that the possible posteriors will lie inside the intervals such that each corresponds to one unique optimal Agent strategy.
For each $j\in \{2,\ldots,K-1\}$, if $s^*(\mu_j)=s_{i_j}$, 
our method will move $\mu_{j}$ into $S_{i_j}$ by $\beta$ for some $\beta>0$. Specifically, let $\mu_{j}' = \mu_{j}-\beta$ if the interval $S_{i_j}$ is below $\mu_{j}$; and $\mu_{j}' = \mu_{j}+\beta$ if the interval $S_{i_j}$ is above $\mu_{j}$. There is no need to move $\mu_1 =0$ and $\mu_K=1$ as they already correspond to a unique optimal strategy. Hence, we let $\mu_1' = \mu_1$ and $\mu_K' =\mu_K$.
As mentioned previously, the length of each interval in $\{S_1,\ldots, S_{n}\}$ is at least $C$. We set $\beta < \frac{C}{4}$ to be small enough so that $\mu_{j}'\in S_{i_j}$ and thus we have $s_{i_j}= s^*(\mu_{j}')$.
Now let $\textrm{Ex}' =\{(\mu_{1}', v(s_{i_1})),\ldots, (\mu_{K}', v(s_{i_K}))\}$ denote the modified set of extreme points and let
$$v'(\mu) = \sup\{z|(\mu,z)\in \conv(\textrm{Ex}')\}\,.$$
We illustrate $\textrm{Ex}'$ and $v'(\mu)$ in Fig~\ref{fig:bayes2}. Similar to $v^*(\mu)$, we can achieve $v'(\mu)$ by finding two points in $\textrm{Ex}'$ to represent $(\mu,v'(\mu))$ by a convex combination of them. This convex combination leads to a distribution of posteriors and thus a signal scheme. We denote this signal scheme by $p'(\mu)$.

\begin{restatable}{lemma}{lmmbayesalt}\label{lmm:bayes-alt}
    There exists a constant $c_2>0$ such that for any $\mu, \epsilon, x \in [0,1]$, $p'(\mu)$ is a $(\frac{3\beta}{C} + c_2\sqrt{\epsilon},\epsilon,x \cdot c_{1}\beta, x)$-optimal stable policy under $\mu$.  
\end{restatable}
Recall that the cardinality $|\cP_\cO|$ of the output policy space of the oracle matters (in Theorem~\ref{thm:stable}) but the output space of $p'$ could be huge. Hence we need to discretize the output space $\{p'(\mu)|\mu\in [0,1]\}$. We defer the details of discretization to Appendix~\ref{app:bayes}. The upshot of the discretization step is that together with our stabilization step, we can obtain the following theorem:
\begin{restatable}[Stable Policy Oracle for Bayesian Persuasion]{theorem}{thmbayesdiscrete}\label{thm:bayes-discrete}
     There exist positive constants $C,c_1,c_2$ such that for any $\beta\in [0,\frac{C}{4}), \epsilon,x\in [0,1]$ and any $\delta \leq \frac{\beta^2}{16}$, there exists a policy oracle $p_\delta(\cdot)$ which is
     $(\frac{3\beta}{C} + c_2 \sqrt{\epsilon} +2\sqrt{\delta},\epsilon, x \cdot c_{1}\beta/2, \max(x,\sqrt{\delta}))$-optimal stable with $|\cP_\cO|=\cO(\frac{n^2}{\delta^2})$. By combining with   Theorem~\ref{thm:stable} and setting $\epsilon=T^{-\frac{1}{5}}$, $x=\beta=\sqrt{\epsilon}$, and $\delta = \frac{\beta^2}{16}$, we can achieve Principal's regret: 
     $$\text{PR}(\sigma^\dagger, \cL,y_{1:T}) = \tilde\cO\left(T^{-\frac{1}{10}}\right)\,,$$
     when the Agent obtains swap regret $\eint = \cO(\sqrt{\abs{\cP_\cO}/T})$.
\end{restatable}

\section{The General Case}\label{sec:general}
In Section \ref{sec:stable} we solved the special case in which we have a stable policy oracle available to us, and in Section \ref{sec:oracles} we showed how to construct stable policy oracles for two important settings: linear contracting, and binary state Bayesian persuasion. 
In this section, we consider the general case, in which we cannot assume the existence of an optimal stable policy oracle. 
In Section~\ref{sec:imposs} we give an example of a setting in which there is no optimal stable policy  (see Lemma~\ref{lem:impossstable}) --- and so indeed, if we want to handle the general case, we need do without such oracles. 
In this case, in addition to the behavioral assumptions in Section~\ref{sec:behavior}, we propose an additional alignment assumption, following~\cite{camara2020mechanisms}. To build intuition for the Alignment assumption, recall that the Principal provides recommendations $r_t$ to the Agent which are the Agent's best response \emph{under the prior corresponding to the Principal's forecast}. We can view the Principal's recommendation as a reflection of what she expects the Agent to do.  The Agent is under no obligation to follow these recommendations however, and will instead play some action $a_t$. In hindsight, we can consider the optimal policy for the Agent mapping the Principal's chosen policies and recommendations to actions for the Agent. We can view this as the benchmark that the Principal expects the Agent to do well with respect to. Alternately, we could consider a richer set of ``swap'' policies that map the Principal's chosen policies and recommendations \emph{and} the Agent's chosen actions to new actions.  The Agent will do well according to this set of swap benchmark policies because of our low swap regret assumption.  This counterfactual ``swap'' set of policies is only richer than the Principal's expectation for the Agent (as it takes as input more information), and so leads to utility for the Agent that is only greater: We call this difference the ``Gap''. The Alignment assumption says that the difference in Principal utility when the Agent plays actions $a_1,\ldots,a_T$ rather than recommendations $r_1,\ldots,r_T$ is upper bounded as a function of the Gap. Or in other words, the only reason that the Principal's utility can substantially suffer given what the Agent plays, compared to what the Principal's expectation was, is if the Gap was large. Said another way, the Principal's utility may well suffer compared to her expectation because the Agent deviates in ways that are beneficial to himself --- but the Agent will not ``frivolously'' deviate in ways that are harmful to the Principal without being helpful to the Agent. In this sense we can view the Alignment assumption as a moral analogue of the traditional assumption that the Agent breaks ties in favor of the Principal.

There is a subtle distinction between our assumption and the one employed in \cite{camara2020mechanisms}: they apply this alignment assumption to the utilities of the stage game for any prior $\pi$ and any $\epsilon$-best response action, whereas we make a similar assumption concerning the sequence of states $y_{1:T}$ for a specific learning algorithm $\cL$ employed by the Agent. Thus it can be that our alignment is satisfied even if the alignment assumption in \cite{camara2020mechanisms} is not.

\begin{assumption}[Alignment]\label{asp:alignment}
For mechanism $\sigma$, let $p^\sigma_{1:T}$ and $r^\sigma_{1:T}$ denote the sequences of realized policies and recommendations and let  $a^\sigma_{1:T}$ denote a realized sequence of actions selected by the Agent's learning algorithm $\cL$. We define the gap of the Agent's utilities to be the difference between the optimal achievable utility when the Agent can adopt any modification rule taking (policy, recommendation, action) as input and the optimal achievable utility when the Agent can adopt any modification rule taking (policy, recommendation) as input. More formally, $\ugap(y_{1:T},p^\sigma_{1:T},r^\sigma_{1:T},a^\sigma_{1:T})$ is defined as
\[\frac{1}{T}\max_{h:\cP_0\times\cA\times \cA\mapsto \cA} \min_{h':\cP_0\times\cA \mapsto \cA}\sum_{t=1}^T (U(h(p^\sigma_t,r^\sigma_t, a^\sigma_t),p^\sigma_t, y_t) -U(h'(p^\sigma_t,r^\sigma_t),p^\sigma_t, y_t))\,.\]
Then we assume that the sequence of states of nature $y_{1:T}$ satisfies that there exists an $M_1=\cO(1)$ and $M_2 =o(1)$ for which, under the proposed mechanism,
\begin{align*}
    &\frac{1}{T}\sum_{t=1}^T (V(r_t,p_t, y_t) -V(a_t,p_t, y_t)) \leq M_1\cdot \ugap(y_{1:T},p_{1:T},r_{1:T},a_{1:T}) + M_2\,, 
\end{align*}
and under any constant mechanism $\sigma^{p_0}$,
\begin{align*}
   \frac{1}{T}\sum_{t=1}^T (V(a_t^{p_0},p_0,y_t)-V(r^{p_0}_t,p_0,y_t))\leq M_1\cdot \ugap(y_{1:T},(p_0,\ldots,p_0),r^{p_0}_{1:T},a^{p_0}_{1:T}) + M_2\,.
    \end{align*}
\end{assumption}


Again, as discussed in Section~\ref{sec:behavior}, behavioral assumptions are still necessary. We maintain the no contextual swap regret assumption and a less restrictive version of the no secret information assumption. 
    

We consider a weaker ``no secret information'' assumption than Assumption~\ref{asp:no-correlation} that corresponds to assuming that the Agent's ``cross-swap-regret'' with respect to the Principal's communications (policy and recommendation) is not too negative. Intuitively, cross swap regret compares the Agent's utility to a benchmark that lets the Agent choose an action using an arbitrary mapping from the Principal's policies and recommendations to actions. Having very negative cross swap regret means that the Agent is performing substantially better than is possible using the information contained in the Principal's communications. We assume that this is not the case.  
\begin{assumption}[No Negative Cross-Swap-Regret]\label{asp:no-sec-info}
Fix any realized sequence of states $y_{1:T}$.
The Agent's corresponding negative \emph{cross-swap-regret} given the sequence of policy-recommendation pairs  $(p_{1:T}^\sigma,r^\sigma_{1:T})$ is defined to be: 
$$\nr(y_{1:T},p_{1:T}^\sigma,r^\sigma_{1:T}):=\frac{1}{T}\EEs{a_{1:T}^\sigma}{\sum_{t=1}^T U(a_t^\sigma,p_t^\sigma,y_t)- \max_{h:\cP_0 \times \cA\mapsto \cA} \sum_{t=1}^T U(h(p_t^\sigma, r_t^\sigma),p_t^\sigma,y_t)}\,.$$
We assume that the Agent's negative cross swap regret is bounded by $\eneg$ for both the realized sequence of policies and recommendations generated by the Principal's mechanism, as well as counterfactually for any constant mechanism:

    \begin{equation*}
        \nr(y_{1:T},p_{1:T},r_{1:T})\leq\eneg\,,
    \end{equation*}
    and for all $p_0\in \cP_0$,
    \begin{equation*}
        \nr(y_{1:T},(p_0,\ldots,p_0),r^{p_0}_{1:T}) \leq\eneg\,.
    \end{equation*}
\end{assumption}

The no negative-cross-swap-regret assumption can be viewed as a ``no-secret-information'' assumption. But it seems to have a different character than the no-secret-information assumption we made in previous sections (Assumption~\ref{asp:no-correlation}). Recall that Assumption~\ref{asp:no-correlation} informally asked that the Agent's actions should appear to be statistically independent of the state of nature, conditional on the policy and recommendation offered by the Principal. We note, however, that Assumption \ref{asp:no-sec-info} is strictly weaker than Assumption~\ref{asp:no-correlation}:

\begin{restatable}{lemma}{lmmweakerasp}\label{lmm:weakerasp}
    Assumption~\ref{asp:no-sec-info} is weaker than Assumption~\ref{asp:no-correlation}. More specifically, Assumption~\ref{asp:no-correlation} implies Assumption~\ref{asp:no-sec-info} with $\eneg= \cO(\sqrt{\abs{\cP'}\abs{\cA}/T})$, where $\cP'$ is the set of all possible output policies by the proposed mechanism.
\end{restatable}

\begin{remark}
We also note a more intuitive and direct way to model the idea of ``no secret information'': to assume that the Agent cannot consistently outperform the Principal's recommendation, i.e., 
    \begin{equation}
        \frac{1}{T}\sum_{t=1}^T U(a_t,p_t,y_t) - \frac{1}{T}\sum_{t=1}^T U(r_t,p_t,y_t) \leq\eneg\,.\label{eq:no-better-than-guess}
    \end{equation}
    This is also a stronger assumption than Assumption~\ref{asp:no-sec-info}. 
    If the Agent can't consistently outperform the Principal's  recommendation (Eq~\eqref{eq:no-better-than-guess}), then Assumption~\ref{asp:no-sec-info} holds.
\end{remark}

Under this new set of assumptions, the Principal only needs to select the policy that would be optimal each round in the common prior setting, treating the forecast $\pi_t$ as the common prior (Algorithm~\ref{alg:general}). 
\begin{algorithm}[H]
    \caption{Principal's choice at round $t$}
    \label{alg:general}
        \begin{algorithmic}[1]
            \State{\textbf{Input}: Forecast $\pi_t\in \Delta(\cY)$}
            \State{Select policy $p_t = p^*(\pi_t)\in \argmax_{p\in \cP_0} \EEs{y\sim \pi_t}{V(\brr(p,\pi_t),p,y)}$
            }
        \end{algorithmic}
    \end{algorithm}
\begin{theorem}
\label{thm:general}
Recall the definition of the set of events
$$\cE_3 = \{\ind{a^*(p_0,\pi_t) = a}\}_{p_0 \in \cP_0,a\in \cA}$$
and define 
$$\cE_4 = \{p^*(\pi_t)= p, a^*(p,\pi_t) =a\}_{p\in \cP_0, a\in \cA}\,.$$
Let $\cE' = \cE_3\cup\cE_4$, the union of these events. Under Assumptions~\ref{asp:no-internal-reg} (No Contextual Swap Regret), \ref{asp:alignment} (Alignment), and \ref{asp:no-sec-info} (No Secret Information), by running the forecasting algorithm from \cite{NRRX23} for events $\cE'$ and the choice rule in Algorithm \ref{alg:general},  the Principal can achieve policy regret:
    \begin{equation*}
        \text{PR}(\sigma^\dagger, \cL,y_{1:T}) \leq \tilde \cO\left(\abs{\cY}\sqrt{\frac{\abs{\cP_0}\abs{\cA}}{T}}\right) +M_1(\eint + \eneg) + M_2.
    \end{equation*}
\end{theorem}
Recall that the forecasting algorithm of \cite{NRRX23} runs in time polynomial in $|\cY|$ and the number of events we ask for low bias on, which in this case is a set of size polynomial in the problem parameters: $|\cE'| = O(|\cP_0||\cA|)$. 
The proof of Theorem \ref{thm:general} decomposes into two lemmas. The first lemma bounds the loss of the Principal when the Agent behaves in a very simple manner: he simply follows the recommendation of the Principal at every round. In this case, we can bound the regret of the Principal  by the conditional bias of the Principal's predictions:

\begin{restatable}[Regret is Low if Agent Follows Recommendations]
    {lemma}{lmmfollowrecommendation}\label{lmm:regret-follow-recommendation}Recall the definition of events
$$\cE_3 = \{\ind{a^*(p_0,\pi_t) = a}\}_{p_0 \in \cP_0,a\in \cA}\,,\ \ \  \ \  \cE_4 = \{p^*(\pi_t)= p, a^*(p,\pi_t) =a\}_{p\in \cP_0, a\in \cA}\,.$$
Let $\cE' = \cE_3\cup\cE_4$, the union of these events. If the Principal runs the forecasting algorithm from \cite{NRRX23} for events $\cE'$ and the choice rule in Algorithm \ref{alg:general}, and the Agent follows the Principal's recommendations, then we have:

    \begin{align*}
        \EEs{\pi_{1:T}}{\max_{p_0\in \cP} \frac{1}{T}\sum_{t=1}^T \left(V(r_t^{p_0}, {p_0}, y_t)-V(r_t, p_t, y_t)\right)} \leq  \tilde \cO\left(\abs{\cY}\sqrt{\frac{\abs{\cP_0}\abs{\cA}}{T}}\right) \,,
    \end{align*}
    where $r_t^{p_0} = \brr({p_0},\pip_t)$ and $r_t = \brr(p_t,\pip_t)$ are recommendations under constant mechanism $\sigma^{p_0}$ and the proposed mechanism respectively.
\end{restatable}

The next lemma compares the Principal's cumulative utility under the Agent's actual behavior, compared to the utility he would have obtained had the Agent simply followed the Principal's recommendations. It states that under our behavioral assumptions on the Agent, these two quantities are similar, for both the mechanism run by the Principal and for any constant benchmark mechanism. Specifically, the utility obtained by the Principal under the run mechanism cannot be much smaller than the utility she would have obtained had the Agent followed her recommendations --- and for the constant benchmark mechanisms, the utility obtained by the Principal cannot be much \emph{larger} than the utility she would have obtained had the Agent followed her recommendations. Here ``much smaller'' and ``much larger'' are controlled by the parameters $\eint$ and $\eneg$ in the behavioral assumptions.

\begin{restatable}[Principal's Utility is Close to Agent Following Recommendations]{lemma}{lmmutilitydiff}\label{lmm:utility-difference}
    For any sequence of states of nature $y_{1:T}$ and sequnece of forecast $\pi_{1:T}$, under Assumptions \ref{asp:no-internal-reg}, \ref{asp:alignment} and \ref{asp:no-sec-info}, we have
     \begin{align*}
        \EEs{a_{1:T}}{\frac{1}{T}\sum_{t=1}^T  V(a_t,p_t, y_t)}\geq \frac{1}{T}\sum_{t=1}^T V(r_t,p_t, y_t) - M_1(\eint + \eneg) - M_2
    \end{align*}
     and for all $p_0\in \cP_0$,
    \begin{align*}
    \EEs{a^{p_0}_{1:T}}{\frac{1}{T}\sum_{t=1}^T V(a_t^{p_0},p_0,y_t)}\leq \frac{1}{T}\sum_{t=1}^T V(r^{p_0}_t,p_0,y_t) + M_1(\eint + \eneg) + M_2\,.
    \end{align*}
    
\end{restatable}

Together, these two lemmas combine to give the Theorem. 

\section{Impossiblity Results}
\label{sec:imposs}

Throughout this paper, we have given policy regret bounds for the Principal under a variety of kinds of assumptions: behavioral assumptions for the Agent, and either alignment assumptions on the interaction, or else assumed access to a way of constructing optimal stable policies. In this Section we interrogate the necessity of those assumptions.

\subsection{Stable Policies Do Not Always Exist}
We avoided alignment assumptions by showing how to construct optimal ``stable'' policies in two important special cases: linear contracting settings, and binary state Bayesian persuasion settings. Might we be able to avoid alignment assumptions in full generality this way? Unfortunately not. The lemma below implies that it is sometimes not possible to construct a $(c, \epsilon, \beta, \gamma)$-stable policy oracle such that Theorem~\ref{thm:stable} guarantees vanishing policy regret. The counterexample involves a simple two-policy, two-action contract setting in which the Principal can get high regret to either of their policies, depending on the tiebreaking rule of the Agent.

\begin{restatable}{proposition}{nostable}\label{lem:impossstable}
There exists a Principal/Agent problem 
in which for all priors $\pi$ and for all $c \leq \frac{1}{4}$, $\epsilon \geq 0$, $\gamma \leq \frac{1}{2}$ and $\beta > 0$, there is no $(c,\epsilon,\beta, \gamma)$-optimal stable policy under $\pi$.
\end{restatable}

An implication of this is that it is not possible to extend our ``stable policy oracle'' approach to capture the entire scope of the Principal/Agent problem we study in this paper.

\subsection{A No-Secret-Information Assumption is Necessary}

\necessity*

Recall that in Section \ref{sec:behavior} we introduced two behavioral assumptions: A no contextual-swap-regret assumption (Assumption \ref{asp:no-internal-reg}), as well as a ``no-secret-information'' assumption (Assumption \ref{asp:no-correlation}). Assumption \ref{asp:no-internal-reg} was straightforwardly motivated as the ``rationality'' assumption in our model, but it was less clear that Assumption \ref{asp:no-correlation}---which informally asked that the Agent's actions be un-correlated with the states, conditional on the Principal's actions---was necessary. In this Section we establish the necessity of Assumption \ref{asp:no-correlation}.

This proposition can be interpreted as follows: against any Principal mechanism, either there is an Agent learning algorithm that achieves vanishing Contextual Swap Regret and ensures the Principal high regret, or it is impossible for any Agent learning algorithm to achieve vanishing Contextual Swap Regret. This second case is a degenerate case and could only occur if the Principal mechanism is allowed to output $\Omega(T)$ different policies, leading to an unfairly fine-grained context for the Agent to compete against. In this case, no contextual-swap-regret assumption will rule out all learning algorithms.

One might ask whether Assumption~\ref{asp:no-correlation} is unnecessarily strong for this task; in other words, it might be possible to prove a positive result when the Agent is constrained by Assumption~\ref{asp:no-internal-reg} and a weakened version of Assumption~\ref{asp:no-correlation}. To address this, we also prove that if the Agent is allowed to play any algorithm satisfying Assumption~\ref{asp:no-internal-reg}
and Assumption~\ref{asp:no-sec-info} (introduced in Section~\ref{sec:general}), which is similar to but weaker than Assumption~\ref{asp:no-correlation}, he can ensure the Principal high regret. 

Intuitively, Assumption~\ref{asp:no-correlation} asks for the Agent's actions to not be statistically correlated with the state of nature, while Assumption~\ref{asp:no-sec-info} asks for the Agent to not perform much better than the best fixed mapping from (policy, recommendation) to actions. We show in Lemma~\ref{lmm:weakerasp} that Assumption~\ref{asp:no-sec-info} is weaker than Assumption~\ref{asp:no-correlation}. However, it still asks for something quite strong from the Agent: when combined, Assumptions~\ref{asp:no-internal-reg} and~\ref{asp:no-sec-info} bound the performance of the Agent from above and below. This might seem to suggest that the Agent cannot do much other than play a standard no-regret algorithm.

However, we show that even when satisfying Assumptions~\ref{asp:no-internal-reg} and~\ref{asp:no-sec-info}, an Agent can leverage extra information he has to ensure that the Principal attains high regret. In a simple linear contract setting, we construct an Agent algorithm $\mathcal{L}$ which either plays a simple no-regret algorithm, or uses knowledge of the states of nature to play a sequence that gets him the same utility and ensures the Principal larger utility. Depending on the Principal's actions and the states of nature, $\mathcal{L}$ selects which sub-algorithm to run. We show that for every Principal mechanism, there must be some state of nature sequence under which $\mathcal{L}$ picks the worst option for the Principal, leading to non-vanishing policy regret.

For this additional result to hold, we only need there to exist some Agent learning algorithm which not only gets vanishing Contextual Swap Regret, but also gets vanishing negative regret. Many well-known no-regret algorithms are known to have this guarantee \cite{gofer16}.

\begin{restatable}[Necessity of Assumption~\ref{asp:no-correlation}, Strengthened]{proposition}{strongnecessity}\label{lem:lb-hard}

There exists a simple linear contract setting where, for any Principal mechanism $\sigma$, one of the following must hold:
\begin{itemize}
    \item No learning algorithm $\cL^{*}$ can satisfy Assumption~\ref{asp:no-internal-reg} with $\eint = o(1)$ and Assumption~\ref{asp:no-sec-info} with $\eneg = o(1)$ for all possible sequence of states $y_{1:T}\in \cY^T$.
    \item There exists a learning algorithm $\cL^{*}$ satisfying Assumption~\ref{asp:no-internal-reg} with $\eint = o(1)$ and Assumption~\ref{asp:no-sec-info} with $\eneg = o(1)$ for all possible sequence of states $y_{1:T}\in \cY^T$ and a sequence of states $\bar y_{1:T} \in \cY^T$ for which any mechanism $\sigma$ achieves non-vanishing regret for the Principal, i.e., $\text{PR}(\sigma,\cL^{*},\bar y_{1:T}) = \Omega(1)$.
\end{itemize}
\end{restatable}

We show our impossibility result in a linear contract setting, the same setting we show positive results for in Section~\ref{sec:linear} when the Agent is further constrained by Assumption~\ref{asp:no-correlation}. Therefore, when keeping all else fixed, we prove that Assumption~\ref{asp:no-correlation} makes the difference between a tractable and intractable setting. Note that this does not imply that a Principal can never achieve vanishing regret without Assumption~\ref{asp:no-correlation}. Indeed in Section~\ref{sec:general} we show that Assumption~\ref{asp:no-sec-info} (which is weaker than Assumption~\ref{asp:no-correlation}) suffices if it is paired with an \emph{Alignment} assumption (Assumption \ref{asp:alignment}). However, Alignment assumptions are different in character to our behavioral assumptions: they constrain the sequence of states of nature, and simply rule out the kinds of examples we use in proving our lower bound statements. Thus we can also view this proposition as demonstrating the necessity of the Alignment condition in general.


\section{Discussion and Conclusion}

We have shown how to give strong \emph{policy regret} bounds for a Principal interacting with a long-lived, non-myopic Agent, in an adversarial, prior free setting. In place of common prior assumptions, we have relied on strictly weaker behavioral assumptions, in the style of \cite{camara2020mechanisms}. However, unlike \cite{camara2020mechanisms}, our mechanisms are efficient in the cardinality of the state space. Additionally, for several important special cases, including the linear contracting setting that has been focal in both the economic and computer science contract theory literature, we do not need any other assumptions (in particular avoiding the ``Alignment'' assumption of \cite{camara2020mechanisms})---which means that our setting is a strict relaxation of the common prior setting. 

In fact, our ability to avoid Alignment assumptions is not specific to linear contracting settings (or binary state Bayesian Persuasian settings) --- but is proven for any class of interactions for which we can derive algorithms implementing ``stable policy oracles''. We gave given two such examples in this paper, but surely more exist. Understanding which kinds of interactions admit stable policy oracles---and which do not---seems important to understand, towards being able to flexibly solve repeated Principal/Agent problems in an assumption minimal way.

\bibliographystyle{plainnat}
\bibliography{ref}

\begin{thebibliography}{59}
\providecommand{\natexlab}[1]{#1}
\providecommand{\url}[1]{\texttt{#1}}
\expandafter\ifx\csname urlstyle\endcsname\relax
  \providecommand{\doi}[1]{doi: #1}\else
  \providecommand{\doi}{doi: \begingroup \urlstyle{rm}\Url}\fi

\bibitem[Balcan et~al.(2015)Balcan, Blum, Haghtalab, and
  Procaccia]{Balcan2015CommitmentWR}
Maria-Florina Balcan, Avrim Blum, Nika Haghtalab, and Ariel~D. Procaccia.
\newblock Commitment without regrets: Online learning in stackelberg security
  games.
\newblock \emph{Proceedings of the Sixteenth ACM Conference on Economics and
  Computation}, 2015.
\newblock URL \url{https://api.semanticscholar.org/CorpusID:14830193}.

\bibitem[Balsubramani(2015)]{balsubramani2015sharp}
Akshay Balsubramani.
\newblock Sharp finite-time iterated-logarithm martingale concentration, 2015.

\bibitem[Bernasconi et~al.(2022)Bernasconi, Castiglioni, Marchesi, Gatti, and
  Trov{\`o}]{bernasconi2022sequential}
Martino Bernasconi, Matteo Castiglioni, Alberto Marchesi, Nicola Gatti, and
  Francesco Trov{\`o}.
\newblock Sequential information design: Learning to persuade in the dark.
\newblock \emph{Advances in Neural Information Processing Systems},
  35:\penalty0 15917--15928, 2022.

\bibitem[Bernasconi et~al.(2023)Bernasconi, Castiglioni, Celli, Marchesi,
  Trov{\`o}, and Gatti]{bernasconi2023optimal}
Martino Bernasconi, Matteo Castiglioni, Andrea Celli, Alberto Marchesi,
  Francesco Trov{\`o}, and Nicola Gatti.
\newblock Optimal rates and efficient algorithms for online bayesian
  persuasion.
\newblock In \emph{International Conference on Machine Learning}, pages
  2164--2183. PMLR, 2023.

\bibitem[Blum and Mansour(2007)]{blum2007external}
Avrim Blum and Yishay Mansour.
\newblock From external to internal regret.
\newblock \emph{Journal of Machine Learning Research}, 8\penalty0 (6), 2007.

\bibitem[Blum et~al.(2008)Blum, Hajiaghayi, Ligett, and Roth]{blum2008regret}
Avrim Blum, MohammadTaghi Hajiaghayi, Katrina Ligett, and Aaron Roth.
\newblock Regret minimization and the price of total anarchy.
\newblock In \emph{Proceedings of the fortieth annual ACM symposium on Theory
  of computing}, pages 373--382, 2008.

\bibitem[Blum et~al.(2014)Blum, Haghtalab, and Procaccia]{blum2014learning}
Avrim Blum, Nika Haghtalab, and Ariel~D Procaccia.
\newblock Learning optimal commitment to overcome insecurity.
\newblock \emph{Advances in Neural Information Processing Systems}, 27, 2014.

\bibitem[Bolton and Dewatripont(2004)]{bolton2004contract}
Patrick Bolton and Mathias Dewatripont.
\newblock \emph{Contract theory}.
\newblock MIT press, 2004.

\bibitem[Braverman et~al.(2018)Braverman, Mao, Schneider, and
  Weinberg]{braverman2018selling}
Mark Braverman, Jieming Mao, Jon Schneider, and Matt Weinberg.
\newblock Selling to a no-regret buyer.
\newblock In \emph{Proceedings of the 2018 ACM Conference on Economics and
  Computation}, pages 523--538, 2018.

\bibitem[Brown et~al.(2023)Brown, Schneider, and Vodrahalli]{brown2023learning}
William Brown, Jon Schneider, and Kiran Vodrahalli.
\newblock Is learning in games good for the learners?
\newblock \emph{arXiv preprint arXiv:2305.19496}, 2023.

\bibitem[Cai et~al.(2023)Cai, Weinberg, Wildenhain, and Zhang]{cai2023selling}
Linda Cai, S~Matthew Weinberg, Evan Wildenhain, and Shirley Zhang.
\newblock Selling to multiple no-regret buyers.
\newblock \emph{arXiv preprint arXiv:2307.04175}, 2023.

\bibitem[Camara et~al.(2020)Camara, Hartline, and
  Johnsen]{camara2020mechanisms}
Modibo~K Camara, Jason~D Hartline, and Aleck Johnsen.
\newblock Mechanisms for a no-regret agent: Beyond the common prior.
\newblock In \emph{2020 ieee 61st annual symposium on foundations of computer
  science (focs)}, pages 259--270. IEEE, 2020.

\bibitem[Carroll(2015)]{carroll2015robustness}
Gabriel Carroll.
\newblock Robustness and linear contracts.
\newblock \emph{American Economic Review}, 105\penalty0 (2):\penalty0 536--563,
  2015.

\bibitem[Carroll(2021)]{carroll2021contract}
Gabriel Carroll.
\newblock Contract theory.
\newblock 2021.

\bibitem[Castiglioni et~al.(2021)Castiglioni, Marchesi, and
  Gatti]{castiglioni2021bayesian}
Matteo Castiglioni, Alberto Marchesi, and Nicola Gatti.
\newblock Bayesian agency: Linear versus tractable contracts.
\newblock \emph{arXiv e-prints}, pages arXiv--2106, 2021.

\bibitem[Chassang(2013)]{chassang2013calibrated}
Sylvain Chassang.
\newblock Calibrated incentive contracts.
\newblock \emph{Econometrica}, 81\penalty0 (5):\penalty0 1935--1971, 2013.

\bibitem[Chen et~al.(2020)Chen, Liu, and Podimata]{chen2020learning}
Yiling Chen, Yang Liu, and Chara Podimata.
\newblock Learning strategy-aware linear classifiers.
\newblock \emph{Advances in Neural Information Processing Systems},
  33:\penalty0 15265--15276, 2020.

\bibitem[Cohen et~al.(2022)Cohen, Deligkas, and Koren]{cohen2022learning}
Alon Cohen, Argyrios Deligkas, and Moran Koren.
\newblock Learning approximately optimal contracts.
\newblock In \emph{International Symposium on Algorithmic Game Theory}, pages
  331--346. Springer, 2022.

\bibitem[Cohen and Mansour(2019)]{cohen2019optimal}
Lee Cohen and Yishay Mansour.
\newblock Optimal algorithm for bayesian incentive-compatible exploration.
\newblock In \emph{Proceedings of the 2019 ACM Conference on Economics and
  Computation}, pages 135--151, 2019.

\bibitem[Collina et~al.(2023)Collina, Arunachaleswaran, and
  Kearns]{collina2023efficient}
Natalie Collina, Eshwar~Ram Arunachaleswaran, and Michael Kearns.
\newblock Efficient stackelberg strategies for finitely repeated games.
\newblock In \emph{Proceedings of the 2023 International Conference on
  Autonomous Agents and Multiagent Systems}, pages 643--651, 2023.

\bibitem[Deng et~al.(2019)Deng, Schneider, and Sivan]{deng2019strategizing}
Yuan Deng, Jon Schneider, and Balusubramanian Sivan.
\newblock Strategizing against no-regret learners, 2019.

\bibitem[Dong et~al.(2018)Dong, Roth, Schutzman, Waggoner, and
  Wu]{dong2018strategic}
Jinshuo Dong, Aaron Roth, Zachary Schutzman, Bo~Waggoner, and Zhiwei~Steven Wu.
\newblock Strategic classification from revealed preferences.
\newblock In \emph{Proceedings of the 2018 ACM Conference on Economics and
  Computation}, pages 55--70, 2018.

\bibitem[Dughmi and Xu(2016)]{dughmi2016algorithmic}
Shaddin Dughmi and Haifeng Xu.
\newblock Algorithmic bayesian persuasion.
\newblock In \emph{Proceedings of the forty-eighth annual ACM symposium on
  Theory of Computing}, pages 412--425, 2016.

\bibitem[D\"{u}tting et~al.(2019)D\"{u}tting, Roughgarden, and
  Talgam-Cohen]{dutting19}
Paul D\"{u}tting, Tim Roughgarden, and Inbal Talgam-Cohen.
\newblock Simple versus optimal contracts.
\newblock In \emph{Proceedings of the 2019 ACM Conference on Economics and
  Computation}, EC '19, page 369–387, New York, NY, USA, 2019. Association
  for Computing Machinery.
\newblock ISBN 9781450367929.
\newblock \doi{10.1145/3328526.3329591}.
\newblock URL \url{https://doi.org/10.1145/3328526.3329591}.

\bibitem[D{\"u}tting et~al.(2022)D{\"u}tting, Ezra, Feldman, and
  Kesselheim]{dutting2022combinatorial}
Paul D{\"u}tting, Tomer Ezra, Michal Feldman, and Thomas Kesselheim.
\newblock Combinatorial contracts.
\newblock In \emph{2021 IEEE 62nd Annual Symposium on Foundations of Computer
  Science (FOCS)}, pages 815--826. IEEE, 2022.

\bibitem[Foster and Hart(2018)]{foster2018smooth}
Dean~P Foster and Sergiu Hart.
\newblock Smooth calibration, leaky forecasts, finite recall, and nash
  dynamics.
\newblock \emph{Games and Economic Behavior}, 109:\penalty0 271--293, 2018.

\bibitem[Foster and Vohra(1999)]{foster1999regret}
Dean~P Foster and Rakesh Vohra.
\newblock Regret in the on-line decision problem.
\newblock \emph{Games and Economic Behavior}, 29\penalty0 (1-2):\penalty0
  7--35, 1999.

\bibitem[Foster and Vohra(1998)]{foster1998asymptotic}
Dean~P Foster and Rakesh~V Vohra.
\newblock Asymptotic calibration.
\newblock \emph{Biometrika}, 85\penalty0 (2):\penalty0 379--390, 1998.

\bibitem[Gan et~al.(2022)Gan, Majumdar, Radanovic, and Singla]{gan2022bayesian}
Jiarui Gan, Rupak Majumdar, Goran Radanovic, and Adish Singla.
\newblock Bayesian persuasion in sequential decision-making.
\newblock In \emph{Proceedings of the AAAI Conference on Artificial
  Intelligence}, volume~36, pages 5025--5033, 2022.

\bibitem[Garg et~al.(2024)Garg, Jung, Reingold, and Roth]{GJRR23}
Sumegha Garg, Christopher Jung, Omer Reingold, and Aaron Roth.
\newblock Oracle efficient online multicalibration and omniprediction.
\newblock In \emph{ACM-SIAM Symposium on Discrete Algorithms}, 2024.

\bibitem[Globus{-}Harris et~al.(2023)Globus{-}Harris, Harrison, Kearns, Roth,
  and Sorrell]{GHK23}
Ira Globus{-}Harris, Declan Harrison, Michael Kearns, Aaron Roth, and Jessica
  Sorrell.
\newblock Multicalibration as boosting for regression.
\newblock In Andreas Krause, Emma Brunskill, Kyunghyun Cho, Barbara Engelhardt,
  Sivan Sabato, and Jonathan Scarlett, editors, \emph{International Conference
  on Machine Learning, {ICML} 2023, 23-29 July 2023, Honolulu, Hawaii, {USA}},
  volume 202 of \emph{Proceedings of Machine Learning Research}, pages
  11459--11492. {PMLR}, 2023.
\newblock URL \url{https://proceedings.mlr.press/v202/globus-harris23a.html}.

\bibitem[Gofer and Mansour(2016)]{gofer16}
Eyal Gofer and Yishay Mansour.
\newblock Lower bounds on individual sequence regret.
\newblock \emph{Mach. Learn.}, 103\penalty0 (1):\penalty0 1–26, apr 2016.
\newblock ISSN 0885-6125.
\newblock \doi{10.1007/s10994-015-5531-y}.
\newblock URL \url{https://doi.org/10.1007/s10994-015-5531-y}.

\bibitem[Gopalan et~al.(2022)Gopalan, Kalai, Reingold, Sharan, and
  Wieder]{GopalanKRSW22}
Parikshit Gopalan, Adam~Tauman Kalai, Omer Reingold, Vatsal Sharan, and Udi
  Wieder.
\newblock Omnipredictors.
\newblock In Mark Braverman, editor, \emph{13th Innovations in Theoretical
  Computer Science Conference, {ITCS} 2022, January 31 - February 3, 2022,
  Berkeley, CA, {USA}}, volume 215 of \emph{LIPIcs}, pages 79:1--79:21. Schloss
  Dagstuhl - Leibniz-Zentrum f{\"{u}}r Informatik, 2022.
\newblock \doi{10.4230/LIPIcs.ITCS.2022.79}.
\newblock URL \url{https://doi.org/10.4230/LIPIcs.ITCS.2022.79}.

\bibitem[Gopalan et~al.(2023{\natexlab{a}})Gopalan, Hu, Kim, Reingold, and
  Wieder]{gopalan2023loss}
Parikshit Gopalan, Lunjia Hu, Michael~P Kim, Omer Reingold, and Udi Wieder.
\newblock Loss minimization through the lens of outcome indistinguishability.
\newblock In \emph{14th Innovations in Theoretical Computer Science Conference
  (ITCS 2023)}. Schloss Dagstuhl-Leibniz-Zentrum f{\"u}r Informatik,
  2023{\natexlab{a}}.

\bibitem[Gopalan et~al.(2023{\natexlab{b}})Gopalan, Kim, and
  Reingold]{gopalan2023characterizing}
Parikshit Gopalan, Michael~P Kim, and Omer Reingold.
\newblock Characterizing notions of omniprediction via multicalibration.
\newblock \emph{arXiv preprint arXiv:2302.06726}, 2023{\natexlab{b}}.

\bibitem[Grossman and Hart(1992)]{grossman1992analysis}
Sanford~J Grossman and Oliver~D Hart.
\newblock An analysis of the principal-agent problem.
\newblock In \emph{Foundations of Insurance Economics: Readings in Economics
  and Finance}, pages 302--340. Springer, 1992.

\bibitem[Haghtalab et~al.(2022)Haghtalab, Lykouris, Nietert, and
  Wei]{haghtalab2022learning}
Nika Haghtalab, Thodoris Lykouris, Sloan Nietert, and Alexander Wei.
\newblock Learning in stackelberg games with non-myopic agents.
\newblock In \emph{Proceedings of the 23rd ACM Conference on Economics and
  Computation}, pages 917--918, 2022.

\bibitem[Haghtalab et~al.(2023)Haghtalab, Podimata, and
  Yang]{haghtalab2023calibrated}
Nika Haghtalab, Chara Podimata, and Kunhe Yang.
\newblock Calibrated stackelberg games: Learning optimal commitments against
  calibrated agents, 2023.

\bibitem[H{\'e}bert-Johnson et~al.(2018)H{\'e}bert-Johnson, Kim, Reingold, and
  Rothblum]{hebert2018multicalibration}
Ursula H{\'e}bert-Johnson, Michael Kim, Omer Reingold, and Guy Rothblum.
\newblock Multicalibration: Calibration for the (computationally-identifiable)
  masses.
\newblock In \emph{International Conference on Machine Learning}, pages
  1939--1948. PMLR, 2018.

\bibitem[Ho et~al.(2014)Ho, Slivkins, and Vaughan]{ho2014adaptive}
Chien-Ju Ho, Aleksandrs Slivkins, and Jennifer~Wortman Vaughan.
\newblock Adaptive contract design for crowdsourcing markets: Bandit algorithms
  for repeated principal-agent problems.
\newblock In \emph{Proceedings of the fifteenth ACM conference on Economics and
  computation}, pages 359--376, 2014.

\bibitem[Holmstr{\"o}m(1979)]{holmstrom1979moral}
Bengt Holmstr{\"o}m.
\newblock Moral hazard and observability.
\newblock \emph{The Bell journal of economics}, pages 74--91, 1979.

\bibitem[Holmstrom and Milgrom(1987)]{holmstrom1987aggregation}
Bengt Holmstrom and Paul Milgrom.
\newblock Aggregation and linearity in the provision of intertemporal
  incentives.
\newblock \emph{Econometrica: Journal of the Econometric Society}, pages
  303--328, 1987.

\bibitem[Kakade and Foster(2008)]{kakade2008deterministic}
Sham~M Kakade and Dean~P Foster.
\newblock Deterministic calibration and nash equilibrium.
\newblock \emph{Journal of Computer and System Sciences}, 74\penalty0
  (1):\penalty0 115--130, 2008.

\bibitem[Kamenica and Gentzkow(2011)]{kamenica2011bayesian}
Emir Kamenica and Matthew Gentzkow.
\newblock Bayesian persuasion.
\newblock \emph{American Economic Review}, 101\penalty0 (6):\penalty0
  2590--2615, 2011.

\bibitem[Kolumbus and Nisan(2022)]{kolumbus2022and}
Yoav Kolumbus and Noam Nisan.
\newblock How and why to manipulate your own agent: On the incentives of users
  of learning agents.
\newblock \emph{Advances in Neural Information Processing Systems},
  35:\penalty0 28080--28094, 2022.

\bibitem[Lykouris et~al.(2016)Lykouris, Syrgkanis, and
  Tardos]{lykouris2016learning}
Thodoris Lykouris, Vasilis Syrgkanis, and {\'E}va Tardos.
\newblock Learning and efficiency in games with dynamic population.
\newblock In \emph{Proceedings of the twenty-seventh annual ACM-SIAM symposium
  on Discrete algorithms}, pages 120--129. SIAM, 2016.

\bibitem[Mansour et~al.(2022{\natexlab{a}})Mansour, Mohri, Schneider, and
  Sivan]{mansour2022strategizing}
Yishay Mansour, Mehryar Mohri, Jon Schneider, and Balasubramanian Sivan.
\newblock Strategizing against learners in bayesian games.
\newblock In \emph{Conference on Learning Theory}, pages 5221--5252. PMLR,
  2022{\natexlab{a}}.

\bibitem[Mansour et~al.(2022{\natexlab{b}})Mansour, Slivkins, Syrgkanis, and
  Wu]{mansour2022bayesian}
Yishay Mansour, Aleksandrs Slivkins, Vasilis Syrgkanis, and Zhiwei~Steven Wu.
\newblock Bayesian exploration: Incentivizing exploration in bayesian games.
\newblock \emph{Operations Research}, 70\penalty0 (2):\penalty0 1105--1127,
  2022{\natexlab{b}}.

\bibitem[Nekipelov et~al.(2015)Nekipelov, Syrgkanis, and
  Tardos]{nekipelov2015econometrics}
Denis Nekipelov, Vasilis Syrgkanis, and Eva Tardos.
\newblock Econometrics for learning agents.
\newblock In \emph{Proceedings of the sixteenth acm conference on economics and
  computation}, pages 1--18, 2015.

\bibitem[Noarov et~al.(2023)Noarov, Ramalingam, Roth, and Xie]{NRRX23}
Georgy Noarov, Ramya Ramalingam, Aaron Roth, and Stephan Xie.
\newblock High-dimensional prediction for sequential decision making.
\newblock \emph{arXiv preprint arXiv:2310.17651}, 2023.

\bibitem[Roth(2023)]{RothNotes}
Aaron Roth.
\newblock Uncertain: Modern topics in uncertainty estimation, September 2023.

\bibitem[Roth et~al.(2016)Roth, Ullman, and Wu]{roth2016watch}
Aaron Roth, Jonathan Ullman, and Zhiwei~Steven Wu.
\newblock Watch and learn: Optimizing from revealed preferences feedback.
\newblock In \emph{Proceedings of the forty-eighth annual ACM symposium on
  Theory of Computing}, pages 949--962, 2016.

\bibitem[Roth et~al.(2020)Roth, Slivkins, Ullman, and
  Wu]{roth2020multidimensional}
Aaron Roth, Aleksandrs Slivkins, Jonathan Ullman, and Zhiwei~Steven Wu.
\newblock Multidimensional dynamic pricing for welfare maximization.
\newblock \emph{ACM Transactions on Economics and Computation (TEAC)},
  8\penalty0 (1):\penalty0 1--35, 2020.

\bibitem[Roughgarden(2015)]{roughgarden2015intrinsic}
Tim Roughgarden.
\newblock Intrinsic robustness of the price of anarchy.
\newblock \emph{Journal of the ACM (JACM)}, 62\penalty0 (5):\penalty0 1--42,
  2015.

\bibitem[Sellke and Slivkins(2021)]{sellke2021price}
Mark Sellke and Aleksandrs Slivkins.
\newblock The price of incentivizing exploration: A characterization via
  thompson sampling and sample complexity.
\newblock In \emph{Proceedings of the 22nd ACM Conference on Economics and
  Computation}, pages 795--796, 2021.

\bibitem[Wu et~al.(2022)Wu, Zhang, Feng, Wang, Yang, Jordan, and
  Xu]{wu2022sequential}
Jibang Wu, Zixuan Zhang, Zhe Feng, Zhaoran Wang, Zhuoran Yang, Michael~I
  Jordan, and Haifeng Xu.
\newblock Sequential information design: Markov persuasion process and its
  efficient reinforcement learning.
\newblock \emph{arXiv preprint arXiv:2202.10678}, 2022.

\bibitem[Zhao et~al.(2021)Zhao, Kim, Sahoo, Ma, and Ermon]{zhao2021calibrating}
Shengjia Zhao, Michael Kim, Roshni Sahoo, Tengyu Ma, and Stefano Ermon.
\newblock Calibrating predictions to decisions: A novel approach to multi-class
  calibration.
\newblock \emph{Advances in Neural Information Processing Systems},
  34:\penalty0 22313--22324, 2021.

\bibitem[Zhu et~al.(2022)Zhu, Bates, Yang, Wang, Jiao, and
  Jordan]{zhu2022sample}
Banghua Zhu, Stephen Bates, Zhuoran Yang, Yixin Wang, Jiantao Jiao, and
  Michael~I Jordan.
\newblock The sample complexity of online contract design.
\newblock \emph{arXiv preprint arXiv:2211.05732}, 2022.

\bibitem[Zu et~al.(2021)Zu, Iyer, and Xu]{zu2021learning}
You Zu, Krishnamurthy Iyer, and Haifeng Xu.
\newblock Learning to persuade on the fly: Robustness against ignorance.
\newblock In \emph{Proceedings of the 22nd ACM Conference on Economics and
  Computation}, pages 927--928, 2021.

\end{thebibliography}

\appendix
\section{Table of Notation}
\begin{table}[H]
    \begin{tabular}{p{4cm}p{10.5cm}}
    \toprule
        \centering\textbf{Symbol} & \textbf{Description}\\
    \midrule
        \centering$\cP$& Policy space.\\
        \centering$\cA$& Action space.\\
        \centering$\cP_0 \subset \cP$ & Benchmark policy set.\\
        \centering$p \in \cP$& Principal's policy.\\
        \centering $a \in \cA$ & Agent's action.\\
        \centering $\mu \in \Delta (\cA)$ & Distribution over Agent's actions.\\
        \centering$r \in \cA$ & Principal's recommended action for the Agent.\\
         \centering$y \in \cY$ & State of nature.\\
        \centering$\hat{y}$ & Empirical distribution over states of nature over a particular subsequence. \\
        \centering$\pi_t\in \Delta(\cY)$ & forecast at time $t$.\\
        \centering $V(a,p,y)$ & Principal's utility.\\
        \centering $U(a,p,y)$ & Agent's utility.\\
        \centering$\brp(\pi)$ & Principal best response assuming a shared prior $\pi$. \\
        \centering$\brr(p,\pi)$ & Agent best response assuming a prior $\pi$, breaking ties in favor of the Principal's utility. \\
        \centering $\cB(p,\pi,\epsilon)$ & the set of all $\epsilon$-best responses for the Agent.\\
        \centering $\brr(p,\pi,\epsilon)$ & the utility-maximizing action for the Principal amongst the Agent's $\epsilon$-best responses to $p$.\\
        \centering$\alpha$ & conditional bias parameter. \\
        \centering$\eint$ & swap regret upper bound. \\
        \centering$\eneg$ & negative regret upper bound. \\
    \bottomrule
    \end{tabular}
    \caption{Summary of game-theoretic notation used in this article.}
    \label{tab:notation}
\end{table}
\section{Proofs from Section~\ref{sec:stable}}\label{sec:stable-proof}

\thmstable*

Let $E_{1,p,r}$ denote the event of $\ind{(p_t,r_t) = (p,r)}$ for all $(p,r)$, $E_{2,p,a}$ denote the event of $\ind{(\pit,\rit) = (p,a)}$ for all $(p,a)$ and $E_{3,p_0,a}$ denote the event of $\ind{a^*(p_0,\pi_t) = a}$ for all $a$.
Let $\cE_{3,p_0} = \{\ind{a^*(p_0,\pi_t) = a}\}_{a\in \cA}$.
Let $\alpha(\cE_1) = \sum_{E\in \cE_1} \alpha(E)$, $\alpha(\cE_2) = \sum_{E\in \cE_2} \alpha(E)$ and $\alpha(\cE_{3,p_0}) = \sum_{E\in \cE_{3,p_0}} \alpha(E)$.
We introduce the following generalized version of Theorem~\ref{thm:stable}.
\begin{restatable}{theorem}{thmstableapp} \label{thm:stable-appendix}
Assume that the Agent's learning algorithm $\cL$ satisfies the behavioral assumptions~\ref{asp:no-internal-reg} and \ref{asp:no-correlation} and 
that the forecasts $\pip_{1:T}$  have conditional bias $\alpha$ conditional on the events $\cE$. 
Given access to an optimal stable policy oracle $\cO_{c,\epsilon,\beta,\gamma}$, by running Algorithm~\ref{alg:general-stable}, which uses $\cO_{c,\epsilon,\beta,\gamma}$ as the choice rule, the Principal can achieve policy regret
\begin{align*}
    &\text{PR}(\cO_{c,\epsilon,\beta,\gamma}, \pi_{1:T}, \cL,y_{1:T}) \\
    = & c+ 3\alpha(\cE_1) + 2\alpha(\cE_2) + \max_{p_0\in \cP_0}\alpha(\cE_{3,p_0}) + \gamma + \frac{\eint +\cO(\sqrt{\abs{\cP_\cO}\abs{\cA}/T})+2\alpha(\cE_1)  }{\beta} \\
    & + \frac{\eint + \cO(\sqrt{\abs{\cA}/T})+2\max_{p_0\in \cP_0}\alpha(\cE_{3,p_0}) }{\epsilon}\,.
\end{align*}
    
\end{restatable}

\begin{proof}[Proof of Theorem~\ref{thm:stable}]
    By Theorem~\ref{thm:forecast-bias}, we have
\begin{align*}
    \EEs{\pi_{1:T}}{\alpha(E)}\leq O\left(\frac{|\cY|\ln(|\cY||\cE|T)}{T} + \frac{|\cY|\sqrt{\ln(|\cY||\cE|T)|\{t : E(\pi_t) = 1|\}}}{T}\right) \,.
\end{align*}
Hence, we have:
\begin{align*}
    &\EEs{\pi_{1:T}}{\alpha(\cE_1)} \leq \cO(\frac{|\cY|\ln(|\cY|(|\cP_{\cO}| +|\cP_0|)\abs{\cA}T)}{T} +\abs{\cY} \sqrt{\frac{\ln(|\cY|(|\cP_{\cO}| +|\cP_0|)\abs{\cA}T)\abs{\cP_\cO}\abs{\cA}}{T}})\,,
    \\
    & \EEs{\pi_{1:T}}{\alpha(\cE_2)} \leq  \cO(\frac{|\cY|\ln(|\cY|(|\cP_{\cO}| +|\cP_0|)\abs{\cA}T)}{T}+\abs{\cY}\sqrt{\frac{\ln(|\cY|(|\cP_{\cO}| +|\cP_0|)\abs{\cA}T)\abs{\cP_0}\abs{\cA}}{T}})\,,\\
   & \EEs{\pi_{1:T}}{\alpha(\cE_{3,p_0})} \leq  \cO(\frac{|\cY|\ln(|\cY|(|\cP_{\cO}| +|\cP_0|)\abs{\cA}T)}{T}+\abs{\cY}\sqrt{\frac{\ln(|\cY|(|\cP_{\cO}| +|\cP_0|)\abs{\cA}T)\abs{\cA}}{T}})\,.
\end{align*}
By taking expectation over $\pi_{1:T}$ and plugging these values into Theorem~\ref{thm:stable-appendix}, we have
\begin{align*}
    &\text{PR}(\sigma^\dagger, \cL,y_{1:T}) \leq\tilde \cO\left(c +\gamma +\sqrt{\abs{\cP_0}\abs{\cA}/T} + \frac{\eint + \abs{\cY}\sqrt{\abs{\cP_\cO}\abs{\cA}/T}}{\beta} + \frac{\eint + \abs{\cY}\sqrt{\abs{\cA}/T}}{\epsilon}\right)\,.
\end{align*}
Hence we are done with proof of Theorem~\ref{thm:stable}.
\end{proof}

\subsection{Proof of Theorem~\ref{thm:stable-appendix} }
\thmstableapp*
\begin{proof}
    For any sequence of states $y_{1:T}$ and sequence of forecasts $\pi_{1:T}$, and any constant policy $p_0\in \cP_0$, for any realized sequence of actions $a_{1:T}$ and $a^{p_0}_{1:T}$,
    we can decompose the (realized) regret compared with constant mechanism $\sigma^{p_0}$ as 
\begin{align*}
    &\frac{1}{T}\sum_{t=1}^T\left(V(a_t^{p_0},{p_0},y_t)- V(a_t,p_t ,y_t)\right)\\
    =& \frac{1}{T}\left(\underbrace{\sum_{t=1}^T(V(\rit,\pit,y_t)- V(r_t,p_t,y_t) )}_{(a)} + \underbrace{\sum_{t=1}^T(V(r_t,p_t,y_t)-V(a_t,p_t ,y_t))}_{(b)}\right.\\
    &\left.+\underbrace{\sum_{t=1}^T(V(a_t^{p_0},{p_0},y_t)- V(\rit,\pit,y_t))}_{(c)} \right)
\end{align*}
\begin{enumerate}
    \item We bound term (a) using the fact that $p_t$ is a $(c,\epsilon,\beta,\gamma)$-optimal stable policy under $\pi_t$.
    According to the definition of stable policy oracle (Definition~\ref{def:stabilized}), we have $V(r_t, p_t,\pip_t)\geq V(\rit,\pit,\pip_t) - c$. Then since $\pip_{1:T}$ has $\alpha$ bias conditional on $(p_t, r_t)$ and $(\pit,\rit)$, we have
    \begin{align*}
        \frac{1}{T}\sum_{t=1}^T V(r_t,p_t, y_t) \geq & \frac{1}{T}\sum_{t=1}^T V(r_t,p_t, \pip_t) - \alpha(\cE_1)\tag{$\cE_1$-bias}\\
        \geq &\frac{1}{T}\sum_{t=1}^TV(\rit,\pit,\pip_t)  -c-\alpha(\cE_1)\tag{stabilization}\\
        \geq &\frac{1}{T}\sum_{t=1}^TV(\rit,\pit,y_t)  -c-\alpha(\cE_2)-\alpha(\cE_1)\,.\tag{$\cE_2$-bias}
    \end{align*}
    Therefore, we have
    \begin{equation*}
        \text{Term (a)} \leq (c+\alpha(\cE_2)+\alpha(\cE_1))T\,.
    \end{equation*}

    \item We bound term (c) using the fact that $V(\rit,\pit,\pi_t)$ is the optimal optimistic achievable utility of the Principal.

For constant mechanism $\sigma^{p_0}$, let $t\in (r)$ denote $t: r^{p_0}_t = r$. Let $n^{p_0}_r = \sum_{t\in (r)}1 $ denote the number of rounds in which $r$ is recommended. 
Let 
\begin{align*}
    b^{p_0}_r = \frac{1}{n^{p_0}_r}\max \left(\abs{\sum_{t:r_t=r} U(a_t^{p_0}, p, y_t) - U(\hat \mu^{p_0}_{r}, p, y_t)}, \abs{\sum_{t:r_t=r} V(a_t^{p_0}, p, y_t) - V(\hat \mu^{p_0}_{r}, p, y_t)} \right)\,.
\end{align*}
By Assumption~\ref{asp:no-correlation}, we have $\EEs{\cL}{b^{p_0}_r} = \cO(\frac{1}{\sqrt{n^{p_0}_r}})$. 
Let $\hat \mu^{p_0}_r = \frac{1}{n^{p_0}_r}\sum_{t\in (r)} a^{p_0}_t$ denote the empirical distribution of Agent's action in the subsequence where $r$ is the recommendation.
Let 
\begin{equation*}
    \ir^{p_0}_{r} = \max_{h:\cA\mapsto \cA}\sum_{t\in (r)} \left(U(h(a^{p_0}_t),p_0, y_t) - U(a^{p_0}_t,{p_0},y_t)\right) 
\end{equation*}
denote the swap regret in this subsequence and let $\ir^{p_0} = \sum_{r\in \cA}\ir^{p_0}_{r}$ denote the swap regret for $a^{p_0}_{1:T}$.
Then we have
    \begin{align}
        &\sum_{t\in (r)}U(\hat \mu^{p_0}_r,p_0,\pip_t) \nonumber\\
        \geq &\sum_{t\in (r)}U(\hat \mu^{p_0}_r,p_0,y_t) - \alpha(E_{3,p_0,r})  T \tag{$\cE_3$-bias} \nonumber\\
        \geq &\sum_{t\in (r)}U(a_t^{p_0},p_0,y_t) -n^{p_0}_r b^{p_0}_r - \alpha(E_{3,p_0,r}) T \tag{no secret info} \nonumber\\
        \geq &\sum_{t\in (r)}U( r,p_0,y_t)-\ir^{p_0}_{r}-n^{p_0}_r b^{p_0}_r-\alpha(E_{3,p_0,r})  T\tag{definition of $\ir^{p_0}_{r}$}\nonumber\\
        \geq &\sum_{t\in (r)}U(r,{p_0},\pip_t)-\ir^{p_0}_{r}-n^{p_0}_r b^{p_0}_r -2\alpha(E_{3,p_0,r}) T\label{eq:u-gap}\,, 
    \end{align}
    where the last inequality again uses our bound on $\cE_3$-bias. 
    For a random action $a\sim \hat \mu^{p_0}_r$, let $F_t$ denote the event that $U(a,{p_0},\pip_t)< U(r,{p_0},\pip_t)-\epsilon$. We have
    \begin{align*}
    &\sum_{t\in (r)}U(\hat \mu^{p_0}_r,p_0,\pip_t)\\
        =&\sum_{t\in (r)}\left(\Pr_{a\sim \hat \mu^{p_0}_r}(F_t) \EEc{U(a,{p_0},\pip_t)}{F_t} +\Pr_{a\sim \hat \mu^{p_0}_r}(\neg F_t) \EEc{U(a,{p_0},\pip_t)}{\neg F_t}\right)\\
        \leq & \sum_{t\in (r)}\left(\Pr_{a\sim \hat \mu^{p_0}_r}(F_t) (U(r,{p_0},\pip_t)-\epsilon) +\Pr_{a\sim \hat \mu^{p_0}_r}(\neg F_t) U(r,{p_0},\pip_t)\right)\,.
    \end{align*}
    By combining with Eq~\eqref{eq:u-gap}, we have
    \begin{align}
        \sum_{t\in (r)}\Pr_{a\sim \hat \mu^{p_0}_r}(F_t)\leq \frac{\ir^{p_0}_{r}+n^{p_0}_r b^{p_0}_r +2\alpha(E_{3,p_0,r}) T}{\epsilon}\,.\label{eq:sum-prob}
    \end{align}
    We also have:
    \begin{align}
    V(\hat \mu^{p_0}_r,p_0,\pip_t) \leq& 
        \Pr_{a\sim \hat \mu^{p_0}_r}(\neg F_t)\max_{\tilde r\in \cB( p_0,\pip_t,\epsilon)}V(\tilde r,p_0,\pip_t) + \Pr_{a\sim \hat \mu^{p_0}_r}(F_t)\nonumber\\
        \leq& \max_{\tilde r\in \cB( p_0,\pip_t,\epsilon)}V(\tilde r,p_0,\pip_t) + \Pr_{a\sim \hat \mu^{p_0}_r}(F_t)\,.\label{eq:v-max}
    \end{align}
    By combining Eqs~\eqref{eq:sum-prob} and \eqref{eq:v-max}, we have
    \begin{align}
        \sum_{t\in (r)}\max_{\tilde r\in \cB( p_0,\pip_t,\epsilon)}V(\tilde r,p_0,\pip_t) \geq \sum_{t\in (r)}V(\hat \mu^{p_0}_r,p_0,\pip_t) - \frac{\ir^{p_0}_{r}+n^{p_0}_r b^{p_0}_r +2\alpha(E_{3,p_0,r}) T}{\epsilon}\,.\label{eq:rounds-not-in-eps-ball}
    \end{align}

    Then we have
    \begin{align*}
        &\sum_{t=1}^T V(\rit,\pit,y_t) \\
        \geq & \sum_{t=1}^T V(\rit,\pit,\pip_t)- \alpha(\cE_2) T \tag{$\cE_2$-bias}\\
        = & \sum_{t=1}^T \max_{\tilde p\in \cP}\max_{\tilde r\in \cB(\tilde p,\pip_t,\epsilon)}V(\tilde r,\tilde p,\pip_t)- \alpha(\cE_2) T\tag{definition of $(\pit,\rit)$}\\
        \geq & \sum_{t=1}^T \max_{\tilde r\in \cB( p_0,\pip_t,\epsilon)}V(\tilde r,p_0,\pip_t)- \alpha(\cE_2) T\\
        = & \sum_{r\in \cA}\sum_{t\in (r)}\max_{\tilde r\in \cB( p_0,\pip_t,\epsilon)}V(\tilde r,p_0,\pip_t)- \alpha(\cE_2) T\\
        \geq & \sum_{r\in \cA}\left(\sum_{t\in (r)}V(\hat \mu^{p_0}_r,p_0,\pip_t)- \frac{\ir^{p_0}_{r}+n^{p_0}_r b^{p_0}_r +2\alpha(E_{3,p_0,r}) T}{\epsilon}\right)- \alpha(\cE_2) T\tag{applying Eq~\eqref{eq:rounds-not-in-eps-ball}}\\
        \geq & \sum_{r\in \cA}\sum_{t\in (r)}V(\hat \mu^{p_0}_r,p_0,y_t)- \frac{\ir^{p_0} + \sum_{r\in \cA}n^{p_0}_r b^{p_0}_r +2\alpha(\cE_{3,p_0}) T}{\epsilon}-\alpha(\cE_{3,p_0}) T- \alpha(\cE_2) T\tag{$\cE_3$-bias}\\
        \geq &  \sum_{t} V(a_t^{p_0},p_0,y_t) - \sum_{r\in \cA}n^{p_0}_r b^{p_0}_r- \frac{\ir^{p_0} + \sum_{r\in \cA}n^{p_0}_r b^{p_0}_r +2\alpha(\cE_{3,p_0}) T}{\epsilon} -  (\alpha(\cE_{3,p_0}) + \alpha(\cE_2)) T\,.\tag{no secret info}
    \end{align*}
    Hence, we have
    \begin{align*}
        \text{Term (c)}=& \sum_{t=1}^T V(a_t^{p_0},{p_0},y_t)- V(\rit,\pit,y_t) \\
        \leq& \sum_{r\in \cA}n^{p_0}_r b^{p_0}_r+ \frac{\ir^{p_0} + \sum_{r\in \cA}n^{p_0}_r b^{p_0}_r +2\alpha(\cE_{3,p_0}) T}{\epsilon} +  (\alpha(\cE_{3,p_0}) + \alpha(\cE_2)) T\,.
    \end{align*}
    By taking expectation over the randomness of the Agent's learning algorithm $\cL$, we have
    \begin{equation*}
        \EEs{\cL}{\text{Term (c)}}\leq \left( \cO(\sqrt{\abs{\cA}/T})+ \frac{\eint + \cO(\sqrt{\abs{\cA}/T})+2\alpha(\cE_{3,p_0}) }{\epsilon} +  \alpha(\cE_{3,p_0}) + \alpha(\cE_2)\right)T\,.
    \end{equation*}
    \item We bound term (b) by proving that the number of rounds in which the Agent does not follow the recommendation $r_t$ is small using the fact that that $p_t$ is $(\beta,\gamma)$-stable under $\pi_t$.

    For proposed mechanism, let $t\in (p,r)$ denote $t:(p_t,r_t) = (p,r)$. Let $n_{p,r} = \sum_{t=1}^T\ind{t\in (p,r)}$ denote the number of rounds in which $(p_t,r_t) = (p,r)$.
    \begin{align*}
    b_{p,r} = \frac{1}{n_{p,r}}\max \left(\abs{\sum_{t\in (p,r)} U(a_t, p, y_t) -U(\hat \mu_{p,r}, p, y_t)}, \abs{\sum_{t\in (p,r)} V(a_t, p, y_t) -V(\hat \mu_{p,r}, p, y_t)} \right)\,.
\end{align*}
By Assumption~\ref{asp:no-correlation}, we have $\EEs{\cL}{b_{p,r}} = \cO(\frac{1}{\sqrt{n_{p,r}}})$. 
    Let $\hat \mu_{p,r} = \frac{1}{n_{p,r}}\sum_{t\in (p,r)} a_t
    $ denote the empirical distribution of the actions on this subsequence.
    Let $\hat y_{p,r} = \frac{1}{n_{p,r}}\sum_{t\in (p,r)} y_t$ denote the empirical distribution of states in these rounds and $\pip_{p,r} = \frac{1}{n_{p,r}}\sum_{t\in (p,r)} \pip_t$ denote the empirical distribution of the forecasts.
    Let
    \begin{equation*}
        \ir_{p,r} = \max_{h:\cA\mapsto \cA}\sum_{t\in (p,r)} \left(U(h(a_t),p_t, y_t) - U(a_t,p_t,y_t)\right)
    \end{equation*}
    denote the swap regret for the Agent over the subsequence in which $(p_t,r_t) = (p,r)$
    and let $\ir = \sum_{(p,r)\in \cP_\cO\times \cA} \ir_{p,r}$ denote the total swap regret (for the action sequence $a_{1:T}$).

In the rounds in which $(p_t,r_t) = (p,r)$, similar to Eq~\eqref{eq:u-gap}, we have
    \begin{align*}
        &\sum_{t\in (p,r)}U(\hat \mu_{p,r},p,\pip_t) \\
        \geq &\sum_{t\in (p,r)}U(\hat \mu_{p,r},p,y_t) - \alpha(E_{1,p,r})  T \tag{$\cE_1$-bias} \nonumber\\
        \geq &\sum_{t\in (p,r)}U(a_t,p,y_t) -n_{p,r}b_{p,r} - \alpha(E_{1,p,r})T \tag{no secret info} \nonumber\\
        \geq &\sum_{t\in (p,r)}U( r,p,y_t)-\ir_{p,r}-n_{p,r}b_{p,r}  -\alpha(E_{1,p,r})  T\tag{definition of $\ir^{p_0}_{r}$}\nonumber\\
        \geq &\sum_{t\in (p,r)}U( r,p,\pi_t)-\ir_{p,r}-n_{p,r}b_{p,r}  -2\alpha(E_{1,p,r})  T\tag{$\cE_1$-bias}\,. 
    \end{align*}
    Since $p$ is $(\beta,\gamma)$-stable under $\pi_t$ for all $t\in (p,r)$, we have $U(a,p,\pip_t)\leq U(r,p,\pip_t)-\beta$ or $V(a,p,\pip_t)\geq V(r,p,\pip_t)-\gamma$ for all $a\neq r$ in $\cA$.
    Let $\rho_{p,r,t} = \Pr_{a\sim \hat \mu_{p,r}}(U(a,p,\pip_t)\leq U(r,p,\pip_t)-\beta)$ denote the probability of $U(a,p,\pip_t)\leq U(r,p,\pip_t)-\beta$ for $a\sim \hat \mu_{p,r}$.
    By combining with $U(a,p,\pip_t)\leq U(r,p,\pip_t)$ for all $a\in \cA$, we have 
    \begin{equation}
        \sum_{t\in (p,r)}\rho_{p,r,t} \leq \frac{\ir_{p,r}+n_{p,r}b_{p,r}  +2\alpha(E_{1,p,r})  T}{\beta}\,.\label{eq:rho}
    \end{equation}
    Therefore, we have
    \begin{align*}
        \text{Term (b)} =& \sum_{t=1}^T (V(r_t,p_t,y_t)-V(a_t,p_t ,y_t))\\
        \leq&\sum_{(p,r)\in \cP_\cO\times \cA} \sum_{t\in (p,r)} (V(r,p,y_t)-V(\hat \mu_{p,r},p ,y_t)) + \sum_{(p,r)\in \cP_\cO\times \cA} n_{p,r}b_{p,r}\tag{no secret info}\\
        \leq &\sum_{(p,r)\in \cP_\cO\times \cA} \sum_{t\in (p,r)} V(r,p,\pi_t)-V(\hat \mu_{p,r},p ,\pi_t) + 2\alpha(\cE_1) T +\sum_{(p,r)\in \cP_\cO\times \cA} n_{p,r}b_{p,r}\tag{$\cE_1$-bias}\\
        \leq& \gamma T + \sum_{(p,r)\in \cP_\cO\times \cA}\sum_{t\in (p,r)}\rho_{p,r,t} + 2\alpha(\cE_1) T+\sum_{(p,r)\in \cP_\cO\times \cA} n_{p,r}b_{p,r}\tag{stability of $p$}\\
        \leq & \gamma T + \frac{\ir +\sum_{(p,r)\in \cP_\cO\times \cA}n_{p,r}b_{p,r}  +2\alpha(\cE_1)  T}{\beta}+ 2\alpha(\cE_1) T+\sum_{(p,r)\in \cP_\cO\times \cA} n_{p,r}b_{p,r}\,.\tag{Apply Eq~\eqref{eq:rho}}
    \end{align*}  
    Hence, by taking the expectation over the randomness of the Agent's algorithm $\cL$, we have
    \begin{align*}
        \EEs{\cL}{\text{Term (b)}} \leq \left(\gamma + \frac{\eint +\cO(\sqrt{\abs{\cP_\cO}\abs{\cA}/T})+2\alpha(\cE_1)  }{\beta}+ 2\alpha(\cE_1)+\cO(\sqrt{\abs{\cP_\cO}\abs{\cA}/T})\right)T\,.
    \end{align*}
\end{enumerate}
Now we have the Principal's regret upper bounded by 
\begin{align*}
    &\textrm{PR}(\cO_{c,\epsilon,\beta,\gamma},\pi_{1:T}, \cL, y_{1:T})\\
    \leq & c+\alpha(\cE_2)+\alpha(\cE_1) 
    \\& \gamma + \frac{\eint +\cO(\sqrt{\abs{\cP_\cO}\abs{\cA}/T})+2\alpha(\cE_1)  }{\beta}+ 2\alpha(\cE_1)+\cO(\sqrt{\abs{\cP_\cO}\abs{\cA}/T})
    \\ & \cO(\sqrt{\abs{\cA}/T})+ \frac{\eint + \cO(\sqrt{\abs{\cA}/T})+2\max_{p_0\in \cP_0}\alpha(\cE_{3,p_0}) }{\epsilon} +  \max_{p_0\in \cP_0}\alpha(\cE_{3,p_0}) + \alpha(\cE_2)
    \\
    = & c+ 3\alpha(\cE_1) + 2\alpha(\cE_2) + \max_{p_0\in \cP_0}\alpha(\cE_{3,p_0}) + \gamma + \frac{\eint +\cO(\sqrt{\abs{\cP_\cO}\abs{\cA}/T})+2\alpha(\cE_1)  }{\beta} \\
    & + \frac{\eint + \cO(\sqrt{\abs{\cA}/T})+2\max_{p_0\in \cP_0}\alpha(\cE_{3,p_0}) }{\epsilon} \,.
\end{align*}
Since $\abs{\cP_0}$ and $\abs{\cA}$ are $\Theta(1)$, we have
\begin{align*}
    &\textrm{PR}(\cO_{c,\epsilon,\beta,\gamma},\pi_{1:T}, \cL, y_{1:T}) \\
    = &\cO(c+ \abs{\cP_\cO}\alpha + \gamma +\frac{\eint +\sqrt{\abs{\cP_\cO}/T}+\abs{\cP_\cO}\alpha }{\beta} + \frac{\eint + \sqrt{1/T}+\alpha }{\epsilon} +\sqrt{\abs{\cP_\cO}/T})\,.
\end{align*}
\end{proof}

\section{Proofs from Section \ref{sec:linear}}

\begin{lemma}
For any $\pi$, and for any Agent actions $a_{1}$ and $a_{2}$ s.t. $a_{1} \neq a_{2}$, there is a unique linear contract $p$ such that $$U(a_{1},p,\pi) = U(a_{2},p,\pi)$$
\label{lem:tie}
\end{lemma}

\begin{proof}
In order for two Agent actions to give the same payoff, we need a $p$ such that

\begin{align*}
& pf(\pi, a_{1}) - c(a_{1}) = pf(\pi, a_{2}) - c(a_{2})  \\
& p = \frac{c(a_{1}) - c(a_{2})}{f(\pi, a_{1}) - f(\pi, a_{2})}
\end{align*}

If $f(\pi, a_{1}) - f(\pi, a_{2}) \neq 0$ this expression is well defined and has a unique solution, and therefore there can be at most one $p$ for which this is true. If $f(\pi, a_{1}) = f(\pi, a_{2})$, then 
\begin{align*}
& pf(\pi, a_{1}) - c(a_{1}) = pf(\pi, a_{1}) - c(a_{2})  \\
& \Leftrightarrow c_{a_{1}} = c_{a_{2}}
\end{align*}

This is a contradiction, as we assume all costs are separated by $\Delta_{c} \geq 0$. Therefore this expression must be well defined and have a unique solution.
\end{proof}

\lemties*

\begin{proof}
To show this, we will first show that for any Agent action $a^{*}$, there are at most $2$ policies for which $a^{*}$ a non-unique best response. To see this, let's consider the smallest linear contract $p_{1}$ such that $a$ is a best response. Let us also consider the largest linear contract $p_{2}$ such that $a$ is a best response. We will show that for all $p$ such that $p_{1} < p < p_{2}$, $a$ is a unique best response. 

As $a^{*}$ is a best response to $p_{1}$, we have

\begin{align*}
& U(a^{*},p_{1},\pi) = \max_{a \in \cA}U(a,p_{1},\pi) \\
& \Leftrightarrow p_{1} f(\pi, a^{*}) - c(a^{*}) = \max_{a \in \cA}(p_{1} f(\pi, a) - c(a) )
\end{align*}

Similarly,

\begin{align*}
& p_{2} f(\pi, a^{*}) - c(a^{*}) = \max_{a \in \cA}(p_{2} f(\pi, a) - c(a) )
\end{align*}

Combining these, we get that, for any $x \in [0,1]$: 

\begin{align*}
& (x p_{1} + (1-x) p_{2})\cdot f(\pi, a^{*}) - c(a^{*}) = x \cdot \max_{a \in \cA}(p_{1} f(\pi, a) - c(a) ) + (1-x) \cdot \max_{a \in \cA}(p_{2} f(\pi, a) - c(a) )
\end{align*}

Now, consider any action $\bar a \neq a^{*}$, evaluated on the linear contract defined by  $(x p_{1} + (1-x) p_{2})$. Assume for contradiction that $\bar a$ is optimal on this contract. Then we have that 

\begin{align*}
& (x p_{1} + (1-x) p_{2})\cdot f(\pi, \bar a) - c(\bar a) = x(p_{1}f(\pi, \bar a) -  c(\bar a)) + (1-x)(p_{2}f(\pi, \bar a) -  c(\bar a))\\
& \geq x \cdot \max_{a \in \cA}(p_{1} f(\pi, a) - c(a) ) + (1-x) \cdot \max_{a \in \cA}(p_{2} f(\pi, a) - c(a) ) \\
\end{align*}

Therefore, it must be the case that $\bar a$ is optimal at $p_{1}$ and $p_{2}$. So $a^{*}$ and $\bar a$ have the same Agent utility at $2$ different contracts. But this is a contradiction of Lemma~\ref{lem:tie}. Therefore any action can be non-uniquely optimal at at most $2$ contracts, the smallest contract at which it is optimal and the largest contract at which it is optimal. At $p = 0$, the optimal action must be the cheapest action, which by our assumption is unique. Therefore there is at least one action that is uniquely optimal at its smallest optimal contract and can only be non-uniquely optimal at $1$ contract. So the total number of contracts with multiple optimal actions is at most
$$\frac{2(|\cA| - 1) + 1}{2}$$
The largest integer value this could be is $|\cA| - 1$, completing our proof.

\end{proof}

\lemgap*

\begin{proof}
Consider any action $a \neq a^{*}$, and the linear contract $\hat p$ such that $U(a, \hat p, \pi) = U(a^{*}, \hat p, \pi)$. Then,

\begin{align*}
& \hat p f(\pi, a) - c(a) = \hat p f(\pi, a^{*}) - c(a^{*}) \\
& \Rightarrow  \hat p (f(\pi, a) - f(\pi, a^{*})) =  c(a)- c(a^{*}) \\
& \Rightarrow  |\hat p (f(\pi, a) - f(\pi, a^{*}))| =  |c(a)- c(a^{*})| \geq \Delta_c \\
& \Rightarrow  |f(\pi, a) - f(\pi, a^{*})| \geq  \Delta_c \\
\end{align*}

At linear contract $\bar p$, the payoff of $a$ is 

\begin{align*}
& U(a, \bar p, \pi) = \bar p f(\pi, a) - c(a) \\
& = (\bar p - \hat p) f(\pi, a) + \hat p f(\pi, a) - c(a) \\
& = (\bar p - \hat p) f(\pi, a) + \hat p f(\pi, a^{*}) - c(a^{*}) = (\bar p - \hat p) f(\pi, a) + U(a^{*},\hat p, \pi) \tag{By the definition of $\hat p$} \\
\end{align*}

Furthermore, we know that 

\begin{align*}
& U(a^{*},\bar p, \pi) = \bar p f(\pi, a^*) - c(a^*) \\
& = (\bar p - \hat p) f(\pi, a^*) + \hat p f(\pi, a^{*}) - c(a^{*}) = (\bar p - \hat p) f(\pi, a^*) + U(a^*, \hat p, \pi)
\end{align*}

Combining these, we get that 

\begin{align*}
& U(a^{*},\bar p, \pi) - U(a, \bar p, \pi) = (\bar p - \hat p) (f(\pi, a^*) - f(\pi, a)) \\
&= |(\bar p - \hat p)| \cdot |(f(\pi, a^*) - f(\pi, a))| \tag{As $a^*$ is optimal at $\bar p$, and thus this difference cannot be negative}\\
& \geq \beta \cdot \Delta_c \\
\end{align*}

\end{proof}

\lemmonotonicity*

\begin{proof}
Let $a_{1} = \max_{a \in \cB(p_{1},\pi,0)}f(\pi,a)$, let $a_{1,\epsilon} = \max_{a \in \cB(p_{1},\pi,\epsilon)}f(\pi,a)$ and let \\$a_{2,\epsilon} = \max_{a \in \cB(p_{2},\pi,\epsilon)}f(\pi,a)$. Note that $a_{1}$ is the Agent's exact best response action under $p_{1}$ which is best for the Principal, while $a_{1,\epsilon}$ and $a_{2,\epsilon}$ are the Agent's $\epsilon$-approximate best response actions which are best for the Principal, under their respective policies. This, we can restate our lemma as proving that for any two linear contracts $p_{1}$, $p_{2}$ s.t. $p_{1} \geq p_{2}$, $f(\pi, a_{1,\epsilon}) \geq f(\pi, a_{2,\epsilon})$.


Assume for contradiction that this is not the case, and $f(\pi, a_{1,\epsilon}) < f(\pi, a_{2,\epsilon})$. Then it must be that $a_{2,\epsilon} \notin \cB(p_{1},\pi,\epsilon)$, as otherwise we would have that 

\begin{align*}
& f(\pi, a_{2,\epsilon}) > f(\pi, a_{1,\epsilon})  \\
& \geq f(\pi, a_{2,\epsilon}) \tag{By the fact that $a_{2,\epsilon} \in \cB(p_{1},\pi,\epsilon)$ and $a_{1,\epsilon}$ is optimal over all $\cB(p_{1},\pi,\epsilon)$} \\
\end{align*}
This is a contradiction. 

As $a_{2,\epsilon} \notin \cB(p_{1},\pi,\epsilon)$, $a_{2,\epsilon}$ is not an $\epsilon$-approximate best response to $p_{1}$. So we have that





 

\begin{align*}
&p_{1}f(\pi,a_{1}) - c(a_{1}) > p_{1}f(\pi,a_{2,\epsilon}) - c(a_{2,\epsilon}) + \epsilon \\
& \Leftrightarrow  c(a_{2,\epsilon}) - c(a_{1}) - \epsilon > p_{1}f(\pi,a_{2,\epsilon}) - f(\pi,a_{1}))  
\end{align*}

Furthermore, as $a_{2,\epsilon}$ is an $\epsilon$-approximate best response under $p_{2}$, we have that 

\begin{align*}
&p_{2}f(\pi,a_{1}) - c(a_{1}) \leq p_{2}f(\pi,a_{2,\epsilon}) - c(a_{2,\epsilon}) + \epsilon \\
& \Leftrightarrow c(a_{2,\epsilon}) - c(a_{1}) - \epsilon \leq p_{2}(f(\pi,a_{2,\epsilon}) - f(\pi,a_{1})) \\
\end{align*}

Finally, we note that 

\begin{align*}
& f(\pi,a_{2,\epsilon}) - f(\pi,a_{1}) \geq f(\pi,a_{2,\epsilon}) - f(\pi,a_{1,\epsilon}) \tag{As $a_{1,\epsilon}$ is maximizing over a larger set} \\
& > 0 \tag{By our assumption}
\end{align*}

Putting these together, we get that 

\begin{align*}
& p_{2}(f(\pi,a_{2,\epsilon}) - f(\pi,a_{1})) > p_{1}(f(\pi,a_{2,\epsilon}) - f(\pi,a_{1}))  \\
& \Rightarrow p_{2} > p_{1} \\
\end{align*}

We have derived a contradiction, completing our proof.
\end{proof}

\section{More Details and Proofs from Section~\ref{sec:bayes}}\label{app:bayes}
\subsection{Discretization details}\label{app:bayes-discrete}
Recall the explicit representation of the signal scheme in Eq~\eqref{eq:signal}. Note that each signal scheme selected under our construction of $p'$ selects two strategies in $\cS$ and each distribution $p'(\cdot|y)$ is supported only on these two strategies.
Now we want to discretize $p'(\cdot|y)$.
For some discretization precision $\delta\ll \beta$ with $\frac{1}{\delta}\in \NN_+$, let $\varphi_{i,j,k_0,k_1}$ for $i,j\in [n], k_0,k_1\in \{0,1,\ldots, \frac{1}{\delta}\}$  represent the signal scheme with 
\[\varphi(s_i|y=1) = k_0\delta \,, \qquad \varphi(s_i|y=0) = k_1\delta\,.\]
Then we let $\cP_\delta = \{\varphi_{i,j,k_0,k_1}| i,j\in [n], k_0,k_1\in  \{0,1,\ldots, \frac{1}{\delta}\}\}$  denote the set of all such signal schemes. We have $\abs{\cP_\delta} = \cO(\frac{n^2}{\delta^2})$. We will return the signal scheme $p_\delta(\mu) \in P_\delta$ closest to $p'(\mu)$. 
Recall that our definition of $p'(\mu)$ induces a convex combination of two points in $\text{Ex}'$, saying $\mu = \tau \cdot \mu_{k}' + (1-\tau) \cdot \mu'_{l}$. Then the explicit form of $p'(\mu)$ is
    \begin{align*}
         p(s_{i_{k}}|y=1) = \frac{\tau \cdot \mu_{k}'}{\mu}\,, \quad &p(s_{i_{l}}|y=1) = \frac{(1-\tau) \cdot \mu_{l}'}{\mu}\\
          p(s_{i_{k}}|y=0) = \frac{\tau \cdot (1-\mu_{k}')}{1-\mu}\,, \quad &p(s_{i_{l}}|y=0) = \frac{(1-\tau) \cdot (1-\mu_{l}')}{1-\mu}\,.
    \end{align*}
    By rounding these two probabilities, we obtain a discretized signal scheme $p_{\delta}(\mu)$ with
    \begin{align*}
         p_\delta(s_{i_{k}}|y=1) = \delta \cdot \argmin_{k\in \{0,\ldots,1/\delta\}} \abs{k\delta -p(s_{i_{k}}|y=1) }\,, \\
         p_\delta(s_{i_{k}}|y=0) =\delta \cdot \argmin_{k\in\{0,\ldots,1/\delta\}} \abs{k\delta -p(s_{i_{k}}|y=0) }\,.
    \end{align*}


\subsection{Proofs}
For any signal $p$ and any prior distribution $\pi = \Ber(\mu)$, let $\{(\tau_i, \Ber(\mu_i))\}_{i\in [n]}$ denote the induced distribution of posteriors where $\tau_i = \sum_{y\in \cY}p(s_i|y)\pi(y)$ is the probability of the signal being $s_i$ and $\Ber(\mu_i)$ is the posterior distribution $\pi(y|s_i)$ of $y$ given the signal $s_i$. 
Then the expected Principal's utility is
\begin{equation*}
    V(a,p,\mu):=\EEs{y\sim \Ber(\mu)}{V(a,p,y)}=\EEs{y\sim \Ber(\mu)}{\EEs{s\sim p(\cdot|y)}{v(a(s),y)}} = \sum_{i\in [n]}\tau_i v(a(s_i))\,,
\end{equation*}
and the expected Agent's utility is 
\begin{equation*}
    U(a,p,\mu):=\EEs{y\sim \Ber(\mu)}{U(a,p,y)}=\EEs{y\sim \Ber(\mu)}{\EEs{s\sim p(\cdot|y)}{u(a(s),y)}} = \sum_{i\in [n]}\tau_i u(a(s_i),\mu_i)\,.
\end{equation*}
Hence, the best response $a^*(p,\mu)$ is defined by letting $a^*(p,\mu)(s_i) = s^*(\mu_i)$ and an action $a$ is an $\epsilon$-best response if $\sum_{i\in[n]} \tau_i u(a(s_i), \mu_i) \geq \sum_{i\in [n]} \tau_i u(s^*(\mu_i), \mu_i)- \epsilon$. Then we first introduce the following lemma to prove our results in Bayesian Persuasion.

\begin{lemma}\label{lmm:bayes-stable}
    For any $x\in [0,1]$, for any $\mu\in [0,1]$, a signal scheme $p$, which induces distribution of posteriors as $(\tau, w_i), ((1-\tau),w_j)$ with $ w_i\in S_i$ and $w_j\in S_j$, is $(x\cdot \eta, x)$-stable under $\mu$ for any $\eta$ with $[w_i-\eta,w_i +\eta]\subset S_i$ and $[w_j-\eta,w_j +\eta]\subset S_j$.
\end{lemma}
    
\begin{proof}
By Assumption \ref{asp-bayes:interval}, each interval has a length of at least $C$. Then for any $i\in n$ and any $\eta < \frac{C}{2}$, let $S_i^\eta$ denote the interval $[\min(S_i)+\eta,\max(S_i)-\eta]$ by removing $\eta$ top values and $\eta$ bottom values from the interval $S_i$. Then for all $\mu\in S_i^\eta$, we have
\begin{align*}
    u(s_j, \mu) \leq u(s_i, \mu) - c_1 \eta\,,
\end{align*}
for all $j\neq i$. 
This directly follows from Assumption \ref{asp-bayes:interval}. As mentioned before, by Assumption \ref{asp-bayes:interval}, there is some minimum difference $c_{1}$ between the utility slopes $\partial u(s,\cdot)$ of any two strategies. Hence for any $\mu$ which is $\eta$-far away from an interval edge, we can see that the Agent utility of every strategy $s_j$ other than the optimal strategy $s_i$ at $\mu$ is at least $c_{1}\cdot \eta$ lower. 
Hence, taking any strategy other than $s_i$ after seeing signal $s_i$ would achieve a utility at least $c_1 \eta$ lower under $w_i$.


If the action $a$ taken by the Agent plays a non-optimal strategy to both $s_i$ and $s_j$, it leads to an expected loss for the Agent of $\geq \tau \cdot c_{1} \eta + (1-\tau) \cdot c_{1} \eta = c_{1}\eta$. More formally, $U(a, p, \mu)\leq U(a^*(p,\mu), p,\mu) - c_1 \eta$. Thus, for any $x\in [0,1]$, we have
\begin{align*}
    U(a, p, \mu)\leq U(a^*(p,\mu),p,\mu) - x\cdot c_1 \eta\,.
\end{align*}

Now consider action $a$ playing one optimal response and one non-optimal response. W.l.o.g., assume that $a(s_i) \neq s_i$ and $a(s_j) = s_j$.
Then we have
\begin{align*}
    U(a,p,\mu) &\leq U(a^*(p,\mu), p,\mu) -\tau c_1 \eta\,,\\
    V(a,p,\mu) &\geq V(a^*(p,\mu), p,\mu) - \tau\,.
\end{align*}
Hence, for any $x\in [0,1]$, if $\tau\leq x$, we have
\begin{align*}
    V(a,p,\mu) \geq V(a^*(p,\mu), p,\mu) - x\,.
\end{align*}
If $\tau>x$, we have
\begin{align*}
    U(a,p,\mu) \leq U(a^*(p,\mu), p,\mu) -x\cdot c_1 \eta\,.
\end{align*}
By combining the two cases, we have proved the lemma.
\end{proof}

\subsubsection{Proof of Lemma~\ref{lmm:bayes-alt}}
\lmmbayesalt*
Before the proof, we first introduce the following lemma.
\begin{lemma}\label{lmm:stabilized}
For any $\mu\in [0,1]$,  we have
    \begin{equation*}
        V(a^*(p'(\mu),\mu),p'(\mu),\mu) \geq V(a^*(p,\mu,\epsilon),p,\mu) -\frac{3\beta}{C} - c_2\sqrt{\epsilon}\,,
    \end{equation*}
    for all $p\in \cP$.
\end{lemma}



\begin{proof}
The proof is decomposed to two parts.
\begin{itemize}
    \item $V(a^*(p'(\mu),\mu),p',\mu)\geq v^*(\mu)- \frac{3\beta}{C}$. (Lemma~\ref{lmm:tech3})
    \item There exists a constant $c_2$ such that $V(a^*(p,\mu,\epsilon),p,\mu) \leq v^*(\mu)+ c_2\sqrt{\epsilon}$ for all $p\in \cP$. (Lemma~\ref{lmm:tech4})
\end{itemize}
By combining these two parts, we prove Lemma~\ref{lmm:stabilized}.

\begin{lemma}\label{lmm:tech3}
For any $\mu\in [0,1]$, we have $V(a^*(p',\mu),p',\mu)\geq v^*(\mu)- \frac{3\beta}{C}$ where $p'=p'(\mu)$.
\end{lemma}
\begin{proof}[Proof of Lemma~\ref{lmm:tech3}]
Recall that the method of finding the optimal achievable Principal's utility by \cite{kamenica2011bayesian}, we have $(\mu,v^*(\mu)) =\tau (\mu_{i_j}, v(s_{i_j})) + (1-\tau)(\mu_{i_{j+1}}, v(s_{i_{j+1}}))$.
Now considering our signal scheme $p'$, there are two cases.

\paragraph{Case 1} The prior $\mu$ lies in $[\mu_{i_j}',\mu_{i_{j+1}}']$ with $\mu = \tau' \mu_{i_j}'+ (1-\tau')\mu_{i_{j+1}}'$.
Recalling our definition of $p'$ (where we find the optimal convex combination of points in $\textrm{Ex}'$), we must have 
$$V(a^*(p',\mu),p',\mu) \geq \tau' v(s_{i_j})+ (1-\tau')v(s_{i_{j+1}})\,.$$
Since $\mu = \tau' \mu_{i_j}'+ (1-\tau')\mu_{i_{j+1}}'$ and $\mu = \tau \mu_{i_j}+ (1-\tau)\mu_{i_{j+1}}$, we have 
\begin{equation*}
    \tau (\mu_{i_{j+1}}-\mu_{i_j}) - \tau' (\mu_{i_{j+1}}'-\mu_{i_j}') = \mu_{i_{j+1}}- \mu_{i_{j+1}}'\,.
\end{equation*}
According to the definition of $\mu'$s, we have
\begin{align*}
     \mu_{i_{j+1}}-\mu_{i_j} + 2\beta \leq \mu_{i_{j+1}}'-\mu_{i_j}' \leq \mu_{i_{j+1}}-\mu_{i_j} + 2\beta \,.
\end{align*}
Therefore, we have
\begin{align*}
    \abs{\tau-\tau'} \leq \frac{\abs{\mu_{i_{j+1}}- \mu_{i_{j+1}}'} + \tau' \cdot 2\beta}{\mu_{i_{j+1}}-\mu_{i_j}}\,.
\end{align*}
According to Assumption \ref{asp-bayes:interval} and definition of $\mu'$s, we have
\begin{align*}
    &\mu_{i_{j+1}}-\mu_{i_j} \geq C\,,\\
    &\\
    &\abs{\mu_{i_{j+1}}- \mu_{i_{j+1}}'}\leq\beta\,.
\end{align*}
Hence, we have $\abs{\tau-\tau'}\leq \frac{3\beta}{C}$. Thus, we have
\begin{align*}
    V(a^*(p',\mu),p',\mu) \geq &\tau' v(s_{i_j})+ (1-\tau')v(s_{i_{j+1}}) \geq  \tau v(s_{i_j})+ (1-\tau)v(s_{i_{j+1}}) -\frac{3\beta}{C} \\
    = &v^*(\mu)-\frac{3\beta}{C}\,. 
\end{align*}

\paragraph{Case 2 }
The prior $\mu$ does not lie in $[\mu_{i_j}',\mu_{i_{j+1}}']$.
Since $\mu \in [\mu_{i_j}, \mu_{i_{j+1}}]$, we have $\mu$ lies in either $[\mu_{i_j},\mu_{i_j}')$ or $(\mu_{i_{j+1}}', \mu_{i_{j+1}}]$. 
W.l.o.g., suppose that $\mu$ lies in $[\mu_{i_j},\mu_{i_j}')$. Then we have $\abs{\mu - \mu_{i_j}} \leq \beta$ and $\abs{\mu - \mu_{i_j}'} \leq \beta$.
Hence we have $\tau \geq 1 - \frac{\beta}{C}$ and 
\begin{align*}
    v^*(\mu)\leq v(s_{i_j}) +\frac{\beta}{C}\,.
\end{align*}
Since $\beta < \frac{C}{4}$, we could find a $\tau'\in [0,1]$ s.t. $\mu = (1-\tau')\mu_{i_{j-1}}'+\tau' \mu_{i_j}'$. Similarly, we have $\tau' \geq 1 - \frac{\beta}{C}$ and thus
\begin{align*}
    V(a^*(p',\mu),p',\mu) \geq (1-\tau')v(s_{i_{j-1}})+\tau' v(s_{i_{j}}) \geq v(s_{i_j}) - \frac{\beta}{C}\,.
\end{align*}
Hence, we have $V(a^*(p',\mu),p',\mu) \geq v^*(\mu) -\frac{2\beta}{C}$.
\end{proof}
\begin{lemma}\label{lmm:tech4}
    There exists a constant $c_2$ such that $V(a^*(p,\mu,\epsilon),p,\mu) \leq v^*(\mu)+ c_2\sqrt{\epsilon}$ for all $p\in \cP$.
\end{lemma}
For any $p\in \cP$, let $\{(\tau_i,\Ber(w_i))\}_{i\in[n]}$ denote the distribution of posteriors induced by policy $p$ and prior $\mu$.
Let $a$ be any $\epsilon$-best response to $(p,\mu)$. Then to prove the lemma, we need to show that there exists a constant $c_2$ such that $V(a,p,\mu) \leq v^*(\mu)+ c_2\sqrt{\epsilon}$ for all $p$.
We introduce lemmas~\ref{lmm:tech1} and \ref{lmm:tech2} to prove Lemma~\ref{lmm:tech4}.

\begin{lemma}\label{lmm:tech1}
        For any $\alpha\in [0, c_1 \cdot C)$, if a strategy $s$ is an $\alpha$-approximate optimal strategy to $\mu$, i.e., $u(s,\mu) =  u(s^*(\mu),\mu)-\alpha$, then there exists $\mu'\in [\mu-\frac{\alpha}{c_1},\mu+\frac{\alpha}{c_1}]$ s.t. $s \in s^*(\mu')$.
    \end{lemma}
    \begin{proof}
    If $\alpha = 0$, then let $\mu' = \mu$. Now we consider the case of $\alpha>0$.
    We first show that if $s_i$ is a best response to $\mu$ and $s_{i+1}$ is not for some $i\in [n]$, then $s_{i+2}$ cannot be an $\alpha$-approximate optimal strategy to $\mu$. This is because $u(s_i,\mu) -u(s_{i+2},\mu)\geq c_1\cdot C$.
    Therefore, if a strategy $s$ is an $\alpha$-approximate optimal strategy to $\mu$, then $s$ can only be $s_{i-1}$ or $s_{i+1}$.    
    W.l.o.g., suppose that $s=s_{i+1}$. 
    Let $\mu' = S_i\cap S_{i+1}$ be the boundary value s.t. both $s_i$ and $s_{i+1}$ are best response to $\mu'$. Then we have $\alpha \geq c_1 \abs{\mu'-\mu}$.         
    \end{proof}
    \begin{lemma}\label{lmm:tech2}
        If an action $a$ is an $\epsilon$-best response to $(p,\mu)$, i.e., $\sum_{i\in[n]} \tau_i u(a(s_i), w_i) \geq \sum_{i\in [n]} \tau_i u(s_i, w_i)- \epsilon$ with $s_i \in s^*(w_i)$, then we can find a set of $\{w_i'|i\in [n]\}$ such that $a(s_i)\in s^*(w_i')$ and $\sum_{i\in [n]}\tau_i \abs{w_i - w_i'}\leq \frac{\epsilon}{c_1}(1+\frac{1}{C})$.
    \end{lemma}
    \begin{proof}
        For each $i\in [n]$, if $u(a(s_i),w_i)\geq u(s_i,w_i) - c_1\cdot C$, then we can find $w_i'$ in the way introduced in Lemma~\ref{lmm:tech1}. Let $A=\{i|u(a(s_i),w_i)\geq u(s_i,w_i) - c_1\cdot C\}$ denote the corresponding subset of $i$'s.
        According to Lemma~\ref{lmm:tech1}, for all $i\in A$, we have
        \begin{align*}
            \abs{w_i - w_i'} \leq \frac{u(s_i,w_i)-u(a(s_i),w_i)}{c_1}\,.
        \end{align*}
        For $i\notin A$, we just arbitrarily pick an $w_i'$ s.t. $a(s_i)$ is an optimal strategy under $w_i'$, i.e, $a(s_i) \in s^*(w_i')$. 
        Then we have $u(s_i,w_i)-u(a(s_i),w_i)> c_1\cdot C$ for all $i\notin A$, and thus
        \begin{equation*}
            \epsilon\geq \sum_{i\notin A} \tau_i \left(u(s_i,w_i)-u(a(s_i),w_i)\right)\geq c_1\cdot C \sum_{i\notin A} \tau_i \,.
        \end{equation*}
        Therefore, we have $\sum_{i\notin A} \tau_i\leq \frac{\epsilon}{c_1 C}$.
        Then we have
        \begin{equation*}
            \sum_{i\in [n]}\tau_i \abs{w_i - w_i'} \leq \frac{1}{c_1}\sum_{i\in A} \tau_i ( u(s_i,w_i)-u(a(s_i),w_i)) + \sum_{i\notin A} \tau_i\leq \frac{\epsilon}{c_1}(1+\frac{1}{C}).
        \end{equation*}
    \end{proof}
    \begin{proof}[Proof of Lemma~\ref{lmm:tech4}]
    Recall that $\{(\tau_i,\Ber(w_i))\}_{i\in[n]}$ is the distribution of posteriors induced by signal scheme $p$ and prior $\mu$ and $a$ is an $\epsilon$-best response to $(p,\mu)$. Now we want to construct another signal scheme $\varphi$ such that $V(a, p, \mu) \leq V(a^*(\varphi, \mu),\varphi,\mu) +c_2 \sqrt{\epsilon}$. 
    Since $v^*(\mu)\geq V(a^*(\varphi, \mu),\varphi,\mu)$ due to that $v^*(\mu)$ is the optimal achievable value when the Agent best respond, we prove Lemma~\ref{lmm:tech4}.

    Our goal is to apply the construction in Lemma~\ref{lmm:tech2}, and construct a distribution of posteriors with support $\{w_i'|i\in [n]\}$. 
    Since $\sum_{i\in [n]} \tau_i w_i' \neq \mu$, we need to find an alternative set of weights $\tau_i'$'s such that $\sum_{i\in [n]} \tau_i' w_i' = \mu$. 
    According to the construction in Lemma~\ref{lmm:tech1}, for those $i$ with $w_i'\neq w_i$, $w'_i$ must lie in $[C,1-C]$ since $w_i'$ always lie on the boundary of two intervals. Let $B = \{i|w_i'\neq w_i\}$.
    Let $q = \sum_{i\in [n]}\tau_i (w_i'-w_i)$. We have
\begin{equation*}
    \sum_{i\in B} \tau_i w_i' = \mu' + q \,,
\end{equation*}
with $\mu'=\mu-\sum_{i\notin B} \tau_i w_i$. According to Lemma~\ref{lmm:tech2}, we have $q\leq \frac{\epsilon}{c_1}(1+\frac{1}{C})$.  Let $\tau_B = \sum_{i\in B}\tau_i$ denote the probability mass of $i\in B$. Then there are three cases.

\begin{itemize}
    \item $\mu'< \frac{\epsilon}{c_1C}(1+\frac{1}{C})$. In this case, we move all probability mass of $\tau_B$ to $w_{n+1}' = \frac{\mu'}{\tau_B}$, which must lie in $[0,1]$ as $\mu' = \sum_{t\in B} \tau_i w_i$. That is to say, let $\tau_i'=0$ for all $i\in B$, $\tau_i'=\tau_i$ for all $i\notin B$ and $\tau_{n+1}'=\tau_B$ for $w_{n+1}' = \mu'$. 
    Then we have $\sum_{i=1}^{n+1}\tau_i' w_i' = \mu' + \sum_{i\notin B} \tau_i w_i = \mu$.
    Thus, $\{(\tau_i',w_i')|i=1,\ldots,n+1\}$ is a Bayesian-plausible distribution of posteriors with $s^*(w_i') = a(s_i)$ for all $i\in [n]$.
    \item $q>0$. we let $\tau_i' = \frac{\mu'\tau_i}{\mu'+q}$ for $i\in B$, $\tau_i' = \tau_i$ for $i\notin B$, and the remaining probability mass $\tau_{n+1}' = 1-\sum_{i\in [n]}\tau_i'$ on $w_{n+1}'=0$. Then we have  $\sum_{i=0}^n \tau_i' w_i' = \mu$ and thus, $\{(\tau_i',w_i')|i=1,\ldots,n+1\}$ is a Bayesian-plausible distribution of posteriors with $s^*(w_i') = a(s_i)$ for all $i\in [n]$.
    \item $q<0$ and $\mu'\geq \frac{\epsilon}{c_1C}(1+\frac{1}{C})$.
Then let $\tau_i' = \frac{(\tau_B-\mu') \tau_i}{\tau_B-\mu'-q}$ for $i\in B$, $\tau_i' = \tau_i$ for $i\notin B$ and the remaining probability mass $\tau_{n+1}' = 1-\sum_{i}\tau_i'$ on $w_{n+1}'=1$.
Note that $\tau_B\geq \frac{1}{1-C}(\mu'+q) \geq \mu'$ where the first inequality holds due to $w_i'\leq 1-C$ for all $i\in B$ and the second inequality holds due to $\mu'\geq \frac{\epsilon}{c_1C}(1+\frac{1}{C})\geq \frac{|q|}{C}$. Thus we have $\tau_i'\geq 0$ and $\tau_i$'s define a legal distribution.
Then we have
\begin{align*}
    \sum_{i\in [n+1]}\tau_i' w_i' &= \frac{(\tau_B-\mu')}{\tau_B-\mu'-q}\sum_{i\in B}\tau_i w_i' + \sum_{i\notin B}\tau_i w_i + (1-\frac{(\tau_B-\mu')}{\tau_B-\mu'-q})\tau_B\\
    & = \frac{(\tau_B-\mu')}{\tau_B-\mu'-q}(\mu'+q) + \sum_{i\notin B}\tau_i w_i + (1-\frac{(\tau_B-\mu')}{\tau_B-\mu'-q})\tau_B\\
    & = \mu'+ \sum_{i\notin B}\tau_i w_i =\mu\,.
\end{align*}
Hence, $\{(\tau_i',w_i')|i=1,\ldots,n+1\}$ is a Bayesian-plausible distribution of posteriors with $s^*(w_i') = a(s_i)$ for all $i\in [n]$. 
\end{itemize}
Since $w_i'\in [C,1-C]$ for all $i\in B$, we have $\mu'+q \in [\tau_B C, \tau_B (1-C)]$.
Then in the first case, we have
\begin{align*}
     V(a, p, \mu) = &\sum_{i\in [n]}\tau_i v(a(s_i)) =\sum_{i\notin B}\tau_i v(a(s_i)) + \tau_B v(s^*(w_{n+1}'))+ \sum_{i\in B}\tau_i (v(a(s_i)) - v(s^*(w_{n+1}'))) \\
    \leq& \sum_{i\in [n+1]} \tau_i' v(s^*(w_i')) + \tau_B \leq V(a^*(\varphi, \mu),\varphi,\mu) +\frac{2\epsilon}{c_1C^2}(1+\frac{1}{C})\,,
\end{align*}
where the last inequality holds due to $\tau_B \leq \frac{\mu'+q}{C}$.

Since $\mu'+q \in [\tau_B C, \tau_B (1-C)]$, in both of the second case and the third case, we have $\tau_i \leq (1+\frac{\abs{q}}{C\tau_B -\abs{q}})\tau_i'$ for all $i\in B$.
Then we have
\begin{align*}
    V(a, p, \mu) = &\sum_{i\in [n]}\tau_i v(a(s_i)) \leq \sum_{i\in B}(1+\frac{\abs{q}}{C\tau_B -\abs{q}})\tau_i' v(a(s_i)) + \sum_{i\notin B}\tau_i' v(a(s_i))\\
    \leq& \sum_{i\in [n]} \tau_i' v(s^*(w_i')) + \frac{\abs{q}}{C\tau_B-\abs{q}} = V(a^*(\varphi, \mu),\varphi,\mu) +\frac{\abs{q}}{C\tau_B -\abs{q}}\,,
\end{align*}
and
\begin{align*}
    V(a, p, \mu) \leq \tau_B + \sum_{i\notin B} \tau_i v(a(s_i))= \tau_B + \sum_{i\notin B} \tau_i' v(s^*(w_i')) \leq V(a^*(\varphi, \mu),\varphi,\mu) +\tau_B\,.
\end{align*}
Since $\min(\tau_B, \frac{\abs{q}}{C\tau_B -\abs{q}}) \leq \sqrt{\frac{\abs{q}}{C}} + \frac{\abs{q}}{C}$, 
by combining these two inequalities together, we have
\begin{align*}
    V(a, p, \mu)\leq V(a^*(\varphi, \mu),\varphi,\mu) + \sqrt{\frac{\abs{q}}{C}} + \frac{\abs{q}}{C} \leq V(a^*(\varphi, \mu),\varphi,\mu) +2\sqrt{\frac{\epsilon}{c_1 C} (1+\frac{1}{C})} \,.
\end{align*}
when $\epsilon$ is small.

    \end{proof}

\end{proof}

\begin{proof}[Proof of Lemma~\ref{lmm:bayes-alt}]
    According to our definition of $p'(\mu)$, it induces a convex combination of two points in $\text{Ex}'$, saying $\mu = \tau \cdot \mu'_{i_{k}} + (1-\tau) \cdot \mu'_{i_{l}}$. 
Recall that all $\mu$ values associated with points in $\text{Ex}$ must be on the boundary between two intervals. 
Furthermore, by construction, any point in $\text{Ex}'$ have $\mu'$ values which are exactly $\beta$ different from some $\mu$ in $\text{Ex}$.
Hence $\mu'_{i_{k}}$ and $\mu'_{i_{l}}$ will be at least $\beta$-far from the edge of any interval. 
Therefore, 
Lemma~\ref{lmm:bayes-stable} implies that $p'(\mu)$ is a $(x \cdot c_{1}\beta,x)$-stable policy under $\mu$.
By combining with Lemma~\ref{lmm:stabilized}, we prove Lemma~\ref{lmm:bayes-alt}.
\end{proof}

\subsubsection{Proof of Theorem~\ref{thm:bayes-discrete}}
\thmbayesdiscrete*
\begin{proof}[Proof of Theorem~\ref{thm:bayes-discrete}]
    Recalling our definition of $p'(\mu)$, it induces a convex combination of two points in $\text{Ex}'$, saying $\mu = \tau \cdot \mu'_{i_{k}} + (1-\tau) \cdot \mu'_{i_{l}}$. Then the explicit form of $p'(\mu)$ is
    \begin{align*}
         p(s_{i_{k}}|y=1) = \frac{\tau \cdot \mu_{k}'}{\mu}\,, \quad p(s_{i_{k}}|y=0) = \frac{\tau \cdot (1-\mu_{k}')}{1-\mu}
    \end{align*}
    By rounding these two probabilities, we obtain a discretized signal scheme $p_{\delta}(\mu)$ with
    \begin{align*}
         p_\delta(s_{i_{k}}|y=1) = \delta \cdot \argmin_{k\in \{0,\ldots,1/\delta\}} \abs{k\delta -p(s_{i_{k}}|y=1) }\,, \quad p_\delta(s_{i_{k}}|y=0) =\delta \cdot \argmin_{k\in\{0,\ldots,1/\delta\}} \abs{k\delta -p(s_{i_{k}}|y=0) }\,.
    \end{align*}
    Let $\delta_1 = p_\delta(s_{i_{k}}|y=1) - p'(s_{i_{k}}|y=1) $ and $\delta_0 = p_\delta(s_{i_{k}}|y=0) - p'(s_{i_{k}}|y=0)$ denote the discretization errors with $\abs{\delta_0},\abs{\delta_1}<\delta$. We have the new distribution of posteriors $(\tau_\delta, \mu_{\delta,i_k}), (1-\tau_\delta, \mu_{\delta,i_l})$
    with
    \begin{align*}
     \tau_\delta &=p_\delta(s_{i_k}|y=1)\mu +p_\delta(s_{i_k}|y=0)(1-\mu) =  (\frac{\tau \cdot \mu_{i_{k}}'}{\mu} + \delta_1)\mu + (\frac{\tau \cdot (1-\mu_{i_{k}}')}{1-\mu} +\delta_0)(1-\mu) \\
     &=\tau + \delta_1 \mu + \delta_0 (1-\mu)\,,\\
        \mu_{\delta,i_k} &= \pi(y|s_{i,k}) = \frac{p_\delta(s_{i_k}|y=1)\mu}{\tau_\delta} = \frac{(\frac{\tau \cdot \mu_{i_{k}}'}{\mu} + \delta_1)\mu}{\tau + \delta_1 \mu + \delta_0 (1-\mu)}=\mu_{i_{k}}' + \frac{\delta_1\mu(1-\mu_{i_{k}}') - \delta_0(1-\mu)\mu_{i_{k}}'}{\tau + \delta_1 \mu + \delta_0 (1-\mu)}\,,\\
       \mu_{\delta,i_l} &= \frac{\mu - \tau_\delta \mu_{\delta,i_k}}{1-\tau_\delta} \,.
    \end{align*}
    Thus, we have $\abs{\tau_\delta - \tau} \leq \delta$ and $\abs{\mu_{\delta,i_k}-\mu_{i_{k}}'} \leq \frac{\delta}{|\tau -\delta|}$. Due to the symmetry, we have $\abs{\mu_{\delta,i_l}-\mu_{i_{l}}'} \leq \frac{\delta}{|(1-\tau) -\delta|}$.
    Then we consider two cases based on the value of $\tau$.
    \begin{itemize}
        \item $\tau <\sqrt{\delta}$ or $\tau > 1-\sqrt{\delta}$. W.l.o.g., we assume that $\tau > 1-\sqrt{\delta}$. Then we can show that $\abs{\mu_{\delta,i_k}-\mu_{i_{k}}'} \leq 2\delta$. Then $p_\delta$ is $((1-\sqrt{\delta})c_1(\beta-2\delta),\sqrt{\delta})$-stable. Since $\mu_{\delta,i_k}$ is at least $\beta-2\delta$ way from the edge and if the Agent chooses the strategy $a(s^*(\mu_{\delta,i_k}))$ is not $s^*(\mu_{\delta,i_k})$ itself given  the signal $s^*(\mu_{\delta,i_k})$, then 
        \begin{align*}
            U(a,p_\delta,\mu) \leq U(a^* (p_\delta,\mu),p_\delta,\mu) - \tau c_1(\beta-2\delta)\leq U(a^*(p_\delta,\mu),p_\delta,\mu) - (1-\sqrt{\delta})c_1(\beta-2\delta)\,.
        \end{align*}
        If the Agent follows the signal $a(s^*(\mu_{\delta,i_k})) = s^*(\mu_{\delta,i_k})$, then 
        \begin{align*}
            V(a,p_\delta,\mu) \geq V(a^*(p_\delta,\mu),p_\delta,\mu) - (1-\tau)\geq V(a^*,p_\delta,\mu) - \sqrt{\delta}\,.
        \end{align*}
        And also, since $\abs{\mu_{\delta,i_k}-\mu_{i_{k}}'} \leq 2\delta$, we have $s^*(\mu_{\delta,i_k}) = s^*(\mu_{i_{k}}')$. Thus, we have
        \begin{align*}
            &V(a^*(p_\delta,\mu),p_\delta,\mu)\\
            = &\tau_\delta v(s^*(\mu_{\delta,i_k})) +(1-\tau_\delta) v(s^*(\mu_{\delta,i_l})) \\
            =& \tau_\delta v(s^*(\mu_{i_{k}}')) +(1-\tau_\delta) v(s^*(\mu_{\delta,i_l}))\\
            \geq& \tau_\delta v(s^*(\mu_{i_{k}}'))\\
            \geq& (\tau-\delta) v(s^*(\mu_{i_{k}}')) \\
            \geq&  V(a^*(p'(\mu),\mu),p'(\mu),\mu) - (1-\tau) -\delta\\ 
            \geq& V(a^*(p'(\mu),\mu),p'(\mu),\mu) -2\sqrt{\delta}
        \end{align*}
        \item $\tau \in [\sqrt{\delta}, 1-\sqrt{\delta}]$. Then both $\abs{\mu_{\delta,i_k}-\mu_{i_{k}}'} \leq 2\sqrt{\delta}$ and $\abs{\mu_{\delta,i_l}-\mu_{i_{l}}'} \leq 2\sqrt{\delta}$. Let $\sqrt{\delta} < \frac{\beta}{4}$. 
        Then by Lemma~\ref{lmm:bayes-stable}, we have that $p_\delta$ is $(x\cdot c_1 \beta/2, x)$-stable for any $x\in [0,1]$.
        Since $s^*(\mu_{\delta,i_k}) = s^*(\mu_{i_{k}}')$ and $s^*(\mu_{\delta,i_l}) = s^*(\mu_{i_{l}}')$, we have
        \begin{align*}
            V(a^*,p_\delta,\mu)=V(a^*,p'(\mu),\mu).
        \end{align*}  
    \end{itemize}
    Hence, $p_\delta(\mu)$ is a $(\frac{3\beta}{C} + c_2 \sqrt{\epsilon} +2\sqrt{\delta},\epsilon, x \cdot c_{1}\beta/2, \max(x,\sqrt{\delta}))$-optimal stable policy under $\mu$ for any $x\in [0,1]$.
\end{proof}

\subsubsection{Proof of Lemma~\ref{lmm:implication-gap-asp}}
\implicationgap*
\begin{proof}
    This is because for each strategy $s\in \cS$, there exists $\mu_s\in [0,1]$ and $c_S>0$ such that $u(s,\mu_s) \geq \mu(s',\mu_s) +C_s$ for all $s'\neq s$ in $\cS$. Hence for all $\mu \in (\mu_s - \frac{C_s}{2}, \mu_s + \frac{C_s}{2})$, we have $u(s,\mu) \geq u(s,\mu_s)-\frac{C_s}{2} 
\geq \mu(s',\mu_s) +\frac{C_s}{2} \geq u(s',\mu)$, where the first and the last inequalities follow from the fact that  $u(s,\cdot)$ is a linear function with $\abs{\partial u(s,\cdot)} \leq 1$.
Let $C =\min_{s\in \cS} C_s$ denote the minimal width of the intervals in $\{S_1,\ldots, S_{n}\}$ observe that since each $C_s > 0$, $C>0$. 

Note that Assumption~\ref{asp-bayes:interval} also implies that, for any two different strategies $s,s'$, the slopes of $u(s,\cdot)$ and $u(s',\cdot)$, denoted by $\partial u(s,\cdot)$ and $\partial u(s',\cdot)$, are different. Otherwise, one of the strategies is dominated by the other one and cannot be strictly optimal at any prior $\mu$, which conflicts with Assumption~\ref{asp-bayes:interval}. 
\end{proof}
\section{Proofs from Section~\ref{sec:general}}

\subsection{Proof of Lemma~\ref{lmm:weakerasp}}
\lmmweakerasp*
\begin{proof}
    When Assumption~\ref{asp:no-correlation} holds, we have
    $$\frac{1}{n_{p,r}}\EEs{a_{1:T}}{\abs{\sum_{t\in (p,r)} U(a_t, p, y_t) -U(\hat \mu_{p,r}, p, y_t)}}\leq \cO\left(\frac{1}{\sqrt{n_{p,r}}}\right)\,,$$
    and
    $$\frac{1}{n^{p_0}_{r}}\EEs{a^{p_0}_{1:T}}{\abs{\sum_{t:r_t=r} U(a_t^{p_0}, p, y_t) - U(\hat \mu^{p_0}_{r}, p, y_t)}}\leq \cO\left(\frac{1}{\sqrt{n^{p_0}_{r}}}\right)\,.$$
    This directly implies the following.
    \begin{align*}
    &\nr(y_{1:T},p_{1:T}^\sigma,r^\sigma_{1:T}) \\
    = & \frac{1}{T}\EEs{a_{1:T}}{\sum_{t=1}^T U(a_t,p_t,y_t) - \max_{h:\cP_0 \times \cA\mapsto \cA} \sum_{t=1}^T U(h(p_t, r_t),p_t,y_t)}\\
       \leq  &\frac{1}{T}\EEs{a_{1:T}}{\sum_{(p,r)\in \cP_\cO\times \cA}\left(\sum_{t\in (p,r)} U(\hat \mu_{p,r}, p, y_t) - \max_{a
       \in \cA} \sum_{t\in (p,r)}U(a, p, y_t)\right)} + \cO(\sqrt{\abs{\cP'}\abs{\cA}/T})\\
        \leq & \cO(\sqrt{\abs{\cP'}\abs{\cA}/T})\,.
    \end{align*}
    Similarly, we have $\nr(y_{1:T},(p_0,\ldots,p_0),r^{p_0}_{1:T})\leq \cO(\sqrt{\abs{\cA}/T})$\,.
\end{proof}
\subsection{Proof of Lemma~\ref{lmm:regret-follow-recommendation}}
\lmmfollowrecommendation*
\begin{proof}[Proof of Lemma~\ref{lmm:regret-follow-recommendation}]
For proposed mechanism $\sigma^\dagger$, let $t\in (p,r)$ denote $t:(p_t,r_t) = (p,r)$. Let $n_{p,r} = \sum_{t=1}^T\ind{t\in (p,r)}$ denote the number of rounds in which $(p_t,r_t) = (p,r)$. 
Let $\hat y_{p,r} = \frac{1}{n_{p,r}}\sum_{t\in (p,r)} y_t$ denote the empirical distribution of states in these rounds and $\pip_{p,r} = \frac{1}{n_{p,r}}\sum_{t\in (p,r)} \pip_t$ denote the empirical distribution of the forecasts.
For constant mechanism $\sigma^{p_0}$, let $t\in (r)$ denote $t: r^{p_0}_t = r$. 
Let $\cE_{3,p_0} = \{\ind{a^*(p_0,\pi_t) = a}\}_{a\in \cA}$.
Let $\alpha(\cE_{3,p_0}) = \sum_{E\in \cE_{3,p_0}} \alpha(E)$ and $\alpha(\cE_4) = \sum_{E\in \cE_4} \alpha(E)$.
For any $p_0\in \cP_0$, we have
    \begin{align*}
        &\sum_{t=1}^T V(r_t,p_t,y_t) \\
        =& \sum_{(p,r)\in \cP_0\times \cA} \sum_{t\in (p,r)} V(r,p,y_t)\\
        =& \sum_{(p,r)\in \cP_0 \times \cA} n_{p,r} V(r,p,\hat y_{p,r}) \\
        \geq & \sum_{(p,r)\in \cP_0 \times \cA} n_{p,r} V(r,p,\pip_{p,r}) -\alpha(\cE_4)T\tag{$\cE_4$-bias}\\
        = & \sum_{(p,r)\in \cP_0 \times \cA} \sum_{t\in(p,r)} V(\brr(p,\pip_t),p,\pip_t) -\alpha(\cE_4)T 
        \tag{since $r_t= \brr(p,\pip_t)$}\\
        =& \sum_{(p,r)\in \cP_0 \times \cA} \sum_{t\in(p,r)} \max_{p'\in\cP_0}V(\brr(p',\pip_t),p',\pip_t) -\alpha(\cE_4)T \tag{since $p_t = p^*(\pi_t)$}\\
        \geq& \sum_{t=1}^T V(\brr(p_0,\pip_t),p_0,\pip_t)-\alpha(\cE_4)T\\
        =& \sum_{r}\sum_{t\in (r)} V(r,p_0,\pip_t) -\alpha(\cE_4)T\\
        \geq & \sum_{r}\sum_{t \in (r)} V(r,p_0,y_t) - \alpha(\cE_{3,p_0})T -\alpha(\cE_4)T\tag{$\cE_3$-bias}\\
        =& \sum_{t=1}^T V(r^{p_0}_t, p_0,y_t)-\alpha(\cE_{3,p_0})T -\alpha(\cE_4)T\,.
    \end{align*}
    According to Theorem~\ref{thm:forecast-bias}, we have $\alpha(\cE_{3,p_0}) = \tilde \cO(\abs{\cY}\sqrt{\abs{\cA}/T})$ and $\alpha(\cE_{4}) = \tilde \cO(\abs{\cY}\sqrt{\abs{\cP_0}\abs{\cA}/T})$. Then we are done with the proof.
\end{proof}

\subsection{Proof of Lemma~\ref{lmm:utility-difference}}
\lmmutilitydiff*
\begin{proof}[Proof of Lemma~\ref{lmm:utility-difference}]
For the proposed mechanism $\sigma^\dagger$, let 
    $$\ir^\dagger_{p,r} = \max_{h:\cA\mapsto \cA}\sum_{t\in (p,r)} \left(U(h(a_t),p_t, y_t) - U(a_t,p_t,y_t)\right)$$
    denote the contextual swap regret for the Agent over the subsequence in which $(p_t,r_t) = (p,r)$.
    Similarly, for the fixed mechanism $\sigma^{p_0}$, let 
    $$\ir^{p_0}_{r} = \max_{h:\cA\mapsto \cA}\sum_{t\in (r)} \left(U(h(a^{p_0}_t),p_0, y_t) - U(a^{p_0}_t,{p_0},y_t)\right)$$
    denote the contextual swap regret for the Agent over the subsequence in which $r^{p_0}_t =  r$.
    
    Similarly, let 
    $$\nr^\dagger_{p,r} = \sum_{t\in (p,r)} U(a_t,p_t, y_t) - \max_{a\in \cA} \sum_{t\in (p,r)} U(a,p_t,y_t)$$
    and  
    $$\nr^{p_0}_{r} = \sum_{t\in (r)} U(a^{p_0}_t,p_0, y_t) -\max_{a\in \cA} \sum_{t\in (r)} U(a,{p_0},y_t)$$
    denote the negative cross swap regrets for the Agent over the subsequence in which $(p_t,r_t) = (p,r)$ under the proposed mechanism $\sigma^\dagger$ and the subsequence in which $r^{p_0}_t =  r$ under the constant mechanism $\sigma^{p_0}$ respectively. For proposed mechanism $\sigma^\dagger$, let $t\in (p,r,a)$ denote $t:(p_t,r_t,a_t) = (p,r,a)$.
    For constant mechanism $\sigma^{p_0}$, let $t\in (r,a)$ denote $t: (r^{p_0}_t, a^{p_0}_t) = (r,a)$.
    We have
    \begin{align*}
    &\ugap(y_{1:T},p^\sigma_{1:T},r^\sigma_{1:T},a^\sigma_{1:T})\\
    =&\max_{h:\cP_0\times\cA\times \cA\mapsto \cA} \min_{h':\cP_0\times\cA \mapsto \cA}\sum_{t=1}^T (U(h(p_t,r_t, a_t),p_t, y_t) -U(h'(p_t,r_t),p_t, y_t))\\
    = & \sum_{(p,r)\in \cP_0\times \cA}\left( \max_{h:\cA\mapsto \cA}\sum_{t\in (p,r)}(U(h(a_t),p, y_t) -U(a_t,p, y_t)) + \min_{r'\in \cA} \sum_{t\in (p,r)}(U(a_t,p, y_t) - U(r',p, y_t) )\right)\\
    =&\sum_{(p,r)\in \cP_0\times \cA} \ir^\dagger_{p,r} + \nr^\dagger_{p,r}\,.
    \end{align*}
    Similarly, for constant mechanism $\sigma^{p_0}$, we have
    $$\ugap(y_{1:T},(p_0,\ldots,p_0),r^{p_0}_{1:T},a^{p_0}_{1:T}) =\sum_{r\in \cA}\ir^{p_0}_{r} +\nr^{p_0}_{r}\,.$$
    According to Assumption~\ref{asp:alignment}, we have
    $$\frac{1}{T}\sum_{t=1}^T (V(r_t,p_t, y_t) -V(a_t,p_t, y_t)) \leq M_1\cdot \ugap(y_{1:T},p_{1:T},r_{1:T},a_{1:T}) + M_2\,. $$
    Therefore,
    \begin{align*}
        \EEs{a_{1:T}}{\frac{1}{T}\sum_{t=1}^T  V(a_t,p_t, y_t)}\geq & \frac{1}{T}\sum_{t=1}^T V(r_t,p_t, y_t)- M_1\cdot \EEs{a_{1:T}}{\ugap(y_{1:T},p^\sigma_{1:T},r^\sigma_{1:T},a^\sigma_{1:T})} - M_2\\
        \geq & \frac{1}{T}\sum_{t=1}^T V(r_t,p_t, y_t) - M_1(\eint + \eneg) - M_2\,,
    \end{align*}
    and
    \begin{align*}
    \EEs{a^{p_0}_{1:T}}{\frac{1}{T}\sum_{t=1}^T V(a_t^{p_0},p_0,y_t)}\leq &\frac{1}{T}\sum_{t=1}^T V(r^{p_0}_t,p_0,y_t) + M_1\cdot \EEs{a^{p_0}_{1:T}}{\ugap(y_{1:T},(p_0,\ldots,p_0),r^{p_0}_{1:T},a^{p_0}_{1:T})} + M_2\\
        \leq & \frac{1}{T}\sum_{t=1}^T V(r^{p_0}_t,p_0,y_t) + M_1(\eint + \eneg) + M_2\,.
    \end{align*}
\end{proof}

\section{Proofs from Section~\ref{sec:imposs}}

\nostable*

\begin{proof}
Consider the following contract setting: there are two actions the Agent can take, $a_{1}$ and $a_{2}$. $a_{1}$ gives the Principal a value of $1$, and $a_{2}$ gives her a value of $2$. The cost of $a_{1}$ for the Agent is $\frac{1}{4}$, and the cost of $a_{2}$ is $\frac{1}{2}$. The Principal's contract space has only two linear contracts, $p_{1} = \frac{1}{4}$ and $p_{2} = \frac{1}{2}$. Thus, $p_{1}$ equally incentivizes $a_{1}$ and $a_{2}$, while $p_{2}$ strictly incentivizes $a_{2}$. 

Intuitively, we will show that $p_{1}$ is not stable, as the Agent could tiebreak in favor of $a_{1}$ instead of $a_{2}$ and significantly decrease the Principal's payoff. Furthermore, $p_{2}$ is not optimal, as if the Agent were tiebreaking in favor of $a_{2}$, the Principal would have rather played $p_{1}$. We formalize this below.

Note that the payoffs for the Principal and Agent are independent of the state of nature, and thus of the prior $\pi$. Furthermore, $a^{*}(p_{1},\pi) = a_{2}$, and $a^{*}(p_{1},\pi) = a_{2}$, $\forall \pi$. Let us first assume for contradiction that $p_{1}$ is a $(\beta,\gamma)$-stable optimal policy where $\gamma = o(1)$ and $\beta > 0$. This means that either

$$U(a,p_{1},\pi) \leq U(a^{*}(p_{1},\pi),p_{1},\pi) - \beta$$
or
$$V(a,p_{1},\pi) \geq V(a^{*}(p_{1},\pi),p_{1},\pi) - \gamma$$

For the first condition, we get that

\begin{align*}
& U(a_{1},p_{1},\pi) \leq U(a_{2},p_{1},\pi) - \beta \\
& \Rightarrow p_{1}f(a_{1}) - c(a_{1}) \leq p_{2}f(a_{1}) - c(a_{2}) - \beta \\ 
& \Rightarrow \frac{1}{4} - \frac{1}{4} \leq \frac{1}{2} - \frac{1}{2} - \beta\\ 
& \Rightarrow \beta \leq 0
\end{align*}

This derives a contradiction, so the second condition must be satisfied.

For the second condition, we get that 

\begin{align*}
& V(a_{1},p_{1},\pi) \geq V(a_{2},p_{1},\pi) - \gamma \\
& \Rightarrow (1 - \frac{1}{4})\cdot 1 \geq (1 - \frac{1}{4}) \cdot 2 - \gamma \\
& \Rightarrow \gamma \geq \frac{3}{4} 
\end{align*}

 This also derives a contradiction. Therefore neither condition is satisfied, so $p_{1}$ is not a $(c,\epsilon,\beta,\gamma)$-stable optimal policy for any $\gamma = o(1)$ and $\beta > 0$.

Next, consider $p_{2}$. Let us assume for contradiction that $p_{2}$ is a $(c,\epsilon,\beta,\gamma)$-stable optimal policy where $c = o(1)$ and $\epsilon = 0$. 

Then, 
\begin{align*}
& V(a^*(p_{2},\pi),p_{2},\pi)\geq V(a^*(p_1,\pi,0),p_1,\pi)-c \\
& \Rightarrow V(a_{2},p_{2},\pi)\geq V(a_{2},p_1,\pi)-c \\
& \Rightarrow 1 \geq \frac{3}{2} - c \\
& \Rightarrow c \geq \frac{1}{2} 
\end{align*} 
This derives a contradiction.

As neither $p_{1}$ nor $p_{2}$ are $(c,\epsilon,\beta,\gamma)$-stable optimal policies for $c = o(1)$, $\epsilon \geq 0$, $\beta > 0$ and $\gamma = o(1)$, this completes our proof.
\end{proof}

\necessity*

\strongnecessity*

We will prove these propositions in conjunction. Our proof assumes the existence of and makes use of the learning algorithm $\mathcal{L}^{*}$, and we derive results for both propositions, depending on which guarantees $\mathcal{L}^{*}$ has.

\begin{proof}

Consider a repeated linear contracting problem with two states of nature, $M$ and $H$, and let the realized state sequence be $y_{1:T}$. The Agent's per-round action space is $\mathcal{A} = \{work, shirk\}$. The Principal's per-round policy space is discretized according to $\cP_\delta = \{0, \delta, 2\delta, \ldots, \floor{\frac{1}{\delta}}\delta\}$, the set of all $\delta$-discretized linear contracts. We assume $\delta$ is such that $0.5, 0.6 \in \cP_\delta$. If the state of nature in a given round is $M$, the task will be completed if and only if the Agent plays $work$. If the state of nature is $H$, the task will not be completed regardless. The Principal gets payoff $2$ if the task is completed. It costs the Agent $0$ to shirk and $1$ to work.




For any mechanism $\sigma$, we will construct an algorithm $\mathcal{L}$ for the Agent that gives the Principal high regret. Unlike standard learning algorithms, $\mathcal{L}$ has access to the entire state sequence. Towards defining this algorithm, we will first define two simpler algorithms that will be used as a subroutines which use knowledge of $y_{2:T}$. We will call these algorithms $a^{*}$ and $b^{*}$.  

$a^{*}$ plays $work$ if $y_{t} = M$ and $shirk$ if $y_{t} = H$. 

$b^{*}$ plays $work$ if $y_{t} = M$ and plays $shirk$ w.p. $\frac{4}{5}$ and $work$ w.p. $\frac{1}{5}$ if $y_{t} = H$. 

Furthermore, let us pick an algorithm which always achieves sublinear Contextual Swap Regret for all states of nature sequences against $\sigma$, and call it $noreg$. We know that $noreg$ must exist, by our assumption that some $\cL^{*}$ exists. If there is a learning algorithms in this setting which achieve sublinear negative regret for all sequences against $\sigma$, we will pick such an algorithm. For some $y_{1:T}$, let $m_{y,t}$ be the number of medium states seen in the first $t$ rounds. Let $balanced_{t}=true$ if, on round $t$, $|m_{m,t} - m_{h,t}| \leq \sqrt{12T\ln\left(2\left(1+\log_{2}\left(T\right)\right)^{2}\right)}$. Furthermore, let $balanced_{all}$ be the event that $balanced_{t} = true$ for all $t \leq T$. Intuitively, this condition checks whether the history of nature states is roughly balanced between $M$ and $H$ at each  round.

We are finally ready to define $\mathcal{L}$. This algorithm uses $a^{*}$, $b^{*}$, $noreg$ and $balanced_{t}$ to exploit knowledge about the states of nature fully, but does so deliberately imperfectly so as not to incur negative regret.

In $\mathcal{L}$, if the very first state of nature of $y$ is $M$, then $\mathcal{L}$ plays $a^{*}$ until the Principal ever plays a contract which is not $(0.5,r_{t}=work)$, and then it plays $noreg$ for the rest of the game. If the very first state of nature of $y$ is $H$, then it plays $b^{*}$ until the Principal ever plays a contract which is not $(0.6,r_{t}=work)$, and then it plays $noreg$ for the rest of the game. Furthermore, if the state sequence ever invalidates the balanced condition, the algorithm immediately begins playing $noreg$ for the rest of the game.

\begin{algorithm}
\caption{$\mathcal{L}$}\label{alg:L1}
\begin{algorithmic}
\State $t \gets 1$
\State{Play $shirk$ on the first round}
\State{Observe $y_{1}$}
\State $t \gets 2$
\If{$y_{1} = M$}
\State $\bar{p} \gets 0.5$
\State $alg \gets a^{*}$
\Else 
\State $\bar p \gets 0.6$
\State $alg \gets b^{*}$
\EndIf
\While{$(t \leq T)$, $(p_{t} = \bar p)$, $(r_{t} = work)$ and $balanced_{t}$}
    \State{Play according to $alg$}
    \State{$t \gets t+1$}
\EndWhile
\While{$(t \leq T)$} 
    \State{Play $noreg$ with the entire history of play in mind}
    \State{$t \gets t+1$}
\EndWhile
\end{algorithmic}
\end{algorithm}

The intuition is as follows: if the number of $M$ and $H$ states is approximately equal, the Principal gets a higher payoff when the Agent plays according to $a^{*}$ or $b^{*}$ than when he plays according to $noreg$. But the Agent himself is roughly indifferent between these algorithms. Therefore if the Principal's mechanism causes $noreg$ to be played when $a^{*}$ or $b^{*}$ could have been played, the Principal will have non-vanishing policy regret. Of course, if the number of $M$ and $H$ states is not approximately equal, there is no guarantee on the performance of $a^{*}$ or $b^{*}$. However, if this is ever the case, $\cL$ will switch to playing $noreg$ to ensure that it continues to satisfy the assumptions on its performance.


We prove that $\cL$ ensures the Principal high regret in Lemma~\ref{lem:nonvanishing}. To do this, we introduce a distribution $y^{*}$ which is i.i.d. between $M$ and $H$ in each round. We use the fact that, in expectation over $y^{*}$, the Principal payoff under $noreg$ is $o(T)$, and the Principal payoff when the Agent is playing either $a^{*}$ or $b^{*}$ is $\Omega(T)$ (Lemma~\ref{lem:payoffUBextended}). This implies that there is at least one sequence under which this difference is realized, or in other words, there is a sequence where the Principal has significant regret when $noreg$ is played rather than $a^{*}$ or $b^{*}$. The final piece we need is that such a sequence exists where $balanced_{all}$ is satisfied, in order that the Agent is actually playing $a^{*}$ or $b^{*}$. Because the probability of a sequence from $y^{*}$ not satisfying the balanced condition approaches $0$ with $T$ (Lemma~\ref{lem:iteratedlog}), we can show that such a sequence must exist.

Next, we turn to proving that $\cL$ has vanishing Contextual Swap regret in Lemma~\ref{lem:contextualnegative}. Towards this, use the fact that if $balanced_{all}$ is true and the Principal is playing in a way that causes the Agent to play $a^{*}$ or $b^{*}$, $a^{*}$ and $b^{*}$ have bounded swap regret (Lemma~\ref{lem:boundreg}). We use this with the fact that $\cL$ switches to playing $noreg$ when either $balanced_{all}$ is not true or the Principal misbehaves to show that $\cL$ always has vanishing swap regret. Combining Lemmas~\ref{lem:nonvanishing} and~\ref{lem:contextualnegative} completes the proof of Proposition~\ref{lem:lb-easy}.

Finally, in the case where $noreg$ also has bounded negative regret, we show in Lemma~\ref{lem:negative} that $\cL$ has bounded negative regret as well, completing the proof of Proposition~\ref{lem:lb-hard}.
\end{proof}

\begin{lemma}
For any Principal mechanism, there is a sequence of states of nature such that $\mathcal{L}$ will ensure the Principal non-vanishing policy regret.
\label{lem:nonvanishing}
\end{lemma}

\begin{proof}
Consider any Principal mechanism $\sigma$. On round $t = 1$, before observing any information about the nature states, the mechanism must provide the first policy. There are two cases:
\begin{itemize}
\item The mechanism provides the contract $0.5$ and the recommendation $work$ $w.p. \leq \frac{1}{2}$. Then, we will evaluate the expected regret of $\sigma$ over the distribution of nature states which begin with $y_{1} = M$ and then are distributed according to $y^{*}_{2:T}$. In the first round, with probability at least $\frac{1}{2}$, the Agent immediately begins playing $noreg$. Alternately, if the Principal had played $(0.5,work)$ in the first round (and throughout the entire game), the Agent would have played $a^{*}$. We can compute the regret of the Principal to this alternate policy sequence, in expectation over $y^{*}_{2:T}$. 

\begin{align*}
& \mathbb{E}_{y_{2:T}^{*},\mathcal{L},\sigma}[\sum_{t = 1}^{T}V(a_{t}^{\sigma},(0.5,work),y_{t})] - \mathbb{E}_{y_{2:T}^{*},\mathcal{L},\sigma}[\sum_{t = 1}^{T}V(a_{t}^{\sigma},p_{t}^{\sigma},y_{t})] \\
& = \mathbb{P}(balanced_{all})\mathbb{E}_{y_{2:T}^{*},\mathcal{L}}[\sum_{t = 1}^{T}V(\mathcal{L},(0.5,work),y_{t})|balanced_{all}]\\ & + \mathbb{P}(\neg balanced_{all})\mathbb{E}_{y_{2:T}^{*},\mathcal{L},\sigma}[\sum_{t = 1}^{T}V(a_{t}^{\sigma},(0.5,work),y_{t})|\neg balanced_{all}]  - \mathbb{E}_{y_{2:T}^{*},\mathcal{L},\sigma}[\sum_{t = 1}^{T}V(a_{t}^{\sigma},p_{t}^{\sigma},y_{t})] \\
& \geq \frac{3}{4}\mathbb{E}_{y_{2:T}^{*},\mathcal{L},\sigma}[\sum_{t = 1}^{T}V(a_{t}^{\sigma},(0.5,work),y_{t})|balanced_{all}] - \mathbb{E}_{y_{2:T}^{*},\mathcal{L},\sigma}[\sum_{t = 1}^{T}V(a_{t}^{\sigma},p_{t}^{\sigma},y_{t})] \tag{By Lemma~\ref{lem:iteratedlog}} \\
& \geq \frac{3}{4} \mathbb{E}_{y_{2:T}^{*},a^{*},\sigma}[\sum_{t = 1}^{T}V(a_{t}^{\sigma},(0.5,work),y_{t})|balanced_{all}]\\  & - \mathbb{P}(\sigma_{1}\neq (0.5,work)) \cdot \mathbb{E}_{y_{2:T}^{*},noreg,\sigma}[\sum_{t = 1}^{T}V(noreg^{^{\sigma}},\sigma,p_{t}^{\sigma},y_{t})] \\ & - \mathbb{P}(\sigma_{1} = (0.5,work)) \cdot \mathbb{E}_{y_{2:T}^{*},\mathcal{L},\sigma}[\sum_{t = 1}^{T}V(a_{t}^{\sigma},p_{t}^{\sigma},y_{t})] \\
& \geq \frac{3}{4} \mathbb{E}_{y_{2:T}^{*},a^{*},\sigma}[\sum_{t = 1}^{T}V(a_{t}^{\sigma},(0.5,work),y_{t})] - o(T) \\ & - \mathbb{P}(\sigma_{1}\neq (0.5,work)) \cdot \mathbb{E}_{y_{2:T}^{*},noreg,\sigma}[\sum_{t = 1}^{T}V(noreg^{^{\sigma}},\sigma,p_{t}^{\sigma},y_{t})] \\ & - \mathbb{P}(\sigma_{1} = (0.5,work)) \cdot \mathbb{E}_{y_{2:T}^{*},\mathcal{L},\sigma}[\sum_{t = 1}^{T}V(a_{t}^{\sigma},p_{t}^{\sigma},y_{t})] \tag{By Lemma~\ref{lem:balanced_close}}\\
& = \frac{3}{4}(\frac{T}{2} - \frac{1}{2}\cdot \frac{T}{2}) - \mathbb{P}(\sigma_{1}\neq (0.5,work)) \cdot \mathbb{E}_{y_{2:T}^{*},noreg,\sigma}[\sum_{t = 1}^{T}V(noreg^{^{\sigma}},\sigma,p_{t}^{\sigma},y_{t})] \\ & - \mathbb{P}(\sigma_{1} = (0.5,work)) \cdot \mathbb{E}_{y_{2:T}^{*},\mathcal{L},\sigma}[\sum_{t = 1}^{T}V(a_{t}^{\sigma},p_{t}^{\sigma},y_{t})] \tag{By the definition of $a^{*}$ over $y^{*}$} \\
& = \frac{3}{4}\cdot \frac{T}{4} - \mathbb{P}(\sigma_{1}\neq (0.5,work)) \cdot o(T) - \mathbb{P}(\sigma_{1} = (0.5,work)) \cdot \mathbb{E}_{y_{2:T}^{*},\mathcal{L},\sigma}[\sum_{t = 1}^{T}V(a_{t}^{\sigma},p_{t}^{\sigma},y_{t})] \tag{By Lemma~\ref{lem:PrincipalpayoffUB}} \\
& \geq \frac{3}{4} \cdot \frac{T}{4} - \mathbb{P}(\sigma_{1}\neq (0.5,work)) \cdot o(T) - \mathbb{P}(\sigma_{1} = (0.5,work)) \cdot (\frac{T}{4} + o(T)) \tag{By Lemma~\ref{lem:payoffUBextended}} \\
& \geq \frac{3}{4} \cdot \frac{T}{4} - \frac{1}{2} \cdot o(T) - \frac{1}{2} \cdot (\frac{T}{4} + o(T)) = \frac{3T}{16} - \frac{T}{8} - o(T) \\
& = \Omega(T) \\
\end{align*}

The expected total regret over this distribution of sequences against $\mathcal{L}$ is $\Omega(T)$. Therefore, there must be at least one sequence beginning with $M$ that has regret of $\Omega(T)$.

\item The mechanism provides the contract $0.6$ and the recommendation $work$ w.p. $\leq \frac{1}{2}$. Then, we will evaluate the expected regret of $\sigma$ over the distribution of nature states which begin with $y_{1} = H$ and then are distributed according to $y^{*}_{2:T}$. There is at least a $\frac{1}{2}$ probability that after the first round, the Agent immediately begins playing $noreg$. Alternately, if the Principal had played $(0.6,work)$ in the first round (and throughout the entire game), the Agent would have played $b^{*}$. We can compute the regret of the Principal to this alternate policy sequence, in expectation over $y^{*}_{2:T}$:

\begin{align*}
& \mathbb{E}_{y_{2:T}^{*},\mathcal{L},\sigma}[\sum_{t = 1}^{T}V(a_{t}^{\sigma},(0.6,work),y_{t})] - \mathbb{E}_{y_{2:T}^{*},\mathcal{L},\sigma}[\sum_{t = 1}^{T}V(a_{t}^{\sigma},p_{t}^{\sigma},y_{t})] \\
& = \mathbb{P}(balanced_{all})\mathbb{E}_{y_{2:T}^{*},\mathcal{L}}[\sum_{t = 1}^{T}V(\mathcal{L},(0.6,work),y_{t})|balanced_{all}]\\ & + \mathbb{P}(\neg balanced_{all})\mathbb{E}_{y_{2:T}^{*},\mathcal{L},\sigma}[\sum_{t = 1}^{T}V(a_{t}^{\sigma},(0.6,work),y_{t})|\neg balanced_{all}]  - \mathbb{E}_{y_{2:T}^{*},\mathcal{L},\sigma}[\sum_{t = 1}^{T}V(a_{t}^{\sigma},p_{t}^{\sigma},y_{t})] \\
& \geq \frac{3}{4}\mathbb{E}_{y_{2:T}^{*},\mathcal{L},\sigma}[\sum_{t = 1}^{T}V(a_{t}^{\sigma},(0.6,work),y_{t})|balanced_{all}] - \mathbb{E}_{y_{2:T}^{*},\mathcal{L},\sigma}[\sum_{t = 1}^{T}V(a_{t}^{\sigma},p_{t}^{\sigma},y_{t})] \tag{By Lemma~\ref{lem:iteratedlog}} \\
& \geq \frac{3}{4} \mathbb{E}_{y_{2:T}^{*},b^{*},\sigma}[\sum_{t = 1}^{T}V(a_{t}^{\sigma},(0.6,work),y_{t})|balanced_{all}]\\  & - \mathbb{P}(\sigma_{1}\neq (0.6,work)) \cdot \mathbb{E}_{y_{2:T}^{*},noreg,\sigma}[\sum_{t = 1}^{T}V(noreg^{^{\sigma}},\sigma,p_{t}^{\sigma},y_{t})] \\ & - \mathbb{P}(\sigma_{1} = (0.6,work)) \cdot \mathbb{E}_{y_{2:T}^{*},\mathcal{L},\sigma}[\sum_{t = 1}^{T}V(a_{t}^{\sigma},p_{t}^{\sigma},y_{t})] \\
& \geq \frac{3}{4} \mathbb{E}_{y_{2:T}^{*},b^{*},\sigma}[\sum_{t = 1}^{T}V(a_{t}^{\sigma},(0.6,work),y_{t})] - o(T) \\ & - \mathbb{P}(\sigma_{1}\neq (0.6,work)) \cdot \mathbb{E}_{y_{2:T}^{*},noreg,\sigma}[\sum_{t = 1}^{T}V(noreg^{^{\sigma}},\sigma,p_{t}^{\sigma},y_{t})] \\ & - \mathbb{P}(\sigma_{1} = (0.6,work)) \cdot \mathbb{E}_{y_{2:T}^{*},\mathcal{L},\sigma}[\sum_{t = 1}^{T}V(a_{t}^{\sigma},p_{t}^{\sigma},y_{t})] \tag{By Lemma~\ref{lem:balanced_close}}\\
& = \frac{3}{4}(\frac{T}{5}) - \mathbb{P}(\sigma_{1}\neq (0.6,work)) \cdot \mathbb{E}_{y_{2:T}^{*},noreg,\sigma}[\sum_{t = 1}^{T}V(noreg^{^{\sigma}},\sigma,p_{t}^{\sigma},y_{t})] \\ & - \mathbb{P}(\sigma_{1} = (0.6,work)) \cdot \mathbb{E}_{y_{2:T}^{*},\mathcal{L},\sigma}[\sum_{t = 1}^{T}V(a_{t}^{\sigma},p_{t}^{\sigma},y_{t})] \tag{By the definition of $b^{*}$ over $y^{*}$} \\
& = \frac{3}{4}\cdot \frac{T}{5} - \mathbb{P}(\sigma_{1}\neq (0.6,work)) \cdot o(T) - \mathbb{P}(\sigma_{1} = (0.6,work)) \cdot \mathbb{E}_{y_{2:T}^{*},\mathcal{L},\sigma}[\sum_{t = 1}^{T}V(a_{t}^{\sigma},p_{t}^{\sigma},y_{t})] \tag{By Lemma~\ref{lem:PrincipalpayoffUB}} \\
& \geq \frac{3T}{20} - \mathbb{P}(\sigma_{1}\neq (0.6,work)) \cdot o(T) - \mathbb{P}(\sigma_{1} = (0.6,work)) \cdot (\frac{T}{4} + o(T)) \tag{By Lemma~\ref{lem:payoffUBextended}} \\
& \geq \frac{3T}{20} - \frac{1}{2} \cdot o(T) - \frac{1}{2} \cdot (\frac{T}{4} + o(T)) = \frac{3T}{20} - \frac{T}{8} - o(T) \\
& = \Omega(T) \\
\end{align*}

For an equivalent argument to the first case, there must be at least one sequence beginning with $H$ that has regret of $\Omega(T)$

\end{itemize}
\end{proof}

Next, we show that these learning algorithm which can guarantee the Principal non-vanishing policy regret also satisfies our assumption that the Agent achieves no swap regret.

\begin{lemma}
$\mathcal{L}$ will have vanishing contextual swap regret  \label{lem:contextualnegative}
\end{lemma}

\begin{proof}

We will prove that $\mathcal{L}$ has vanishing contextual swap regret against any mechanism $\sigma$. Because this set of mechanisms includes all constant mechanisms, we now only need to prove this one stronger claim instead of two claims to satisfy assumption~\ref{asp:no-internal-reg}.
Let $t_{b}$ be the first round when the Agent defects to begin playing $noreg$. Then, the contextual swap regret of $\mathcal{L}$ against any sequence $y$ (not necessarily drawn from $y^{*}$) can be expressed as

\begin{align*}
&T \cdot \ir(y_{1:T},p_{1:T},r_{1:T}) = \E_{\cL,\sigma}\left[\max_{h:\cP\times \cA\times \cA\mapsto \cA}\sum_{t=1}^T (U(h(p_t^{\sigma},r_t^{\sigma},a_t^{\sigma}),p_t^{\sigma}, y_t) - U(a_t^{\sigma},p_t^{\sigma},y_t^{\sigma})) \right]\  \\
& = \E_{\cL,\sigma}\left[\max_{h:\cP\times \cA\times \cA\mapsto \cA}\sum_{t=1}^{t_{b}} (U(h(p_t^{\sigma},r_t^{\sigma},a_t^{\sigma}),p_t^{\sigma}, y_t) - U(a_t^{\sigma},p_t^{\sigma},y_t)) \right]\ + \\ & \E_{\cL,\sigma}\left[\max_{h:\cP\times \cA\times \cA\mapsto \cA}\sum_{t=t_{b}+1}^T (U(h(p_t^{\sigma},r_t^{\sigma},a_t^{\sigma}),p_t^{\sigma}, y_t) - U(a_t^{\sigma},p_t^{\sigma},y_t)) \right]\   \\
& = \E_{\cL,\sigma}\left[\max_{h:\cP\times \cA\times \cA\mapsto \cA}\sum_{t=1}^{t_{b}} (U(h(p_t^{\sigma},r_t^{\sigma},a_t^{\sigma}),p_t^{\sigma}, y_t) - U(a_t^{\sigma},p_t^{\sigma},y_t)) \right]\ + \\ & \E_{noreg,\sigma}\left[\max_{h:\cP\times \cA\times \cA\mapsto \cA}\sum_{t=t_{b}+1}^T (U(h(p_t^{\sigma},r_t^{\sigma},a_t^{\sigma}),p_t^{\sigma}, y_t) - U(a_t^{\sigma},p_t^{\sigma},y_t)) \right]\ \tag{By the fact that $\mathcal{L}$ begins playing $noreg$ at $t_{b}+1$}  \\
& \leq o(T) + \E_{noreg,\sigma}\left[\max_{h:\cP\times \cA\times \cA\mapsto \cA}\sum_{t=t_{b}+1}^T (U(h(p_t^{\sigma},r_t^{\sigma},a_t^{\sigma}),p_t^{\sigma}, y_t) - U(a_t^{\sigma},p_t^{\sigma},y_t)) \right]\ \tag{By Lemma~\ref{lem:boundreg}}  \\
& \leq o(T) \tag{By the fact that $noreg$ has bounded contextual swap regret} \\
\end{align*}

Thus, $\ir(y_{1:T},p_{1:T},r_{1:T}) \leq \frac{o(T)}{T} = o(1)$

\end{proof}

\begin{lemma}
As long as $noreg$ has vanishing negative regret, $\mathcal{L}$ will have vanishing negative regret. \label{lem:negative}
\end{lemma}
\begin{proof}

Let $t_{b}$ be the first round in which the Agent begins playing $noreg$.
We can split up the negative regret of the Agent as follows:

\begin{align*}
&T \cdot \nr(y_{1:T},p_{1:T}^\sigma,r^\sigma_{1:T}) =\EEs{\cL,\sigma}{\sum_{t=1}^T U(a_t^\sigma,p_t^\sigma,y_t)- \max_{h:\cP_0 \times \cA\mapsto \cA} (h(p_t^\sigma, r_t^\sigma),p_t^\sigma,y_t)} \\
& =\EEs{\mathcal{L},\sigma}\sum_{t=1}^{t_{b}} U(a_t^\sigma,p_t^\sigma,y_t)- \max_{h:\cP_0 \times \cA\mapsto \cA} {\sum_{t=1}^{t_{b}} U(h(p_t^\sigma, r_t^\sigma),p_t^\sigma,y_t)} \\
& + \EEs{noreg,\sigma}\sum_{t=t_{b}+1}^T U(a_t^\sigma,p_t^\sigma,y_t)- \max_{h:\cP_0 \times \cA\mapsto \cA} {\sum_{t=t_{b+1}}^{T} U(h(p_t^\sigma, r_t^\sigma),p_t^\sigma,y_t)} \\
& \leq \EEs{\mathcal{L},\sigma}\sum_{t=1}^{t_{b}} U(a_t^\sigma,p_t^\sigma,y_t)- \max_{h:\cP_0 \times \cA\mapsto \cA} {\sum_{t=1}^{t_{b}} U(h(p_t^\sigma, r_t^\sigma),p_t^\sigma,y_t)} + o(T) \tag{By the fact that $noreg$ has vanishing negative regret.} \\
& \leq o(T) \tag{
By Lemma~\ref{lem:boundreg}.}
\end{align*}

Thus, $\nr \leq \frac{o(T)}{T} = o(1)$.

\end{proof}

\label{app:imposs}
\begin{lemma}
$\mathbb{E}_{y^{*},\mathcal{L},\sigma}[\sum_{t = 1}^{T}V(a_{t}^{\sigma},(0.5,work),y_{t})] \leq \mathbb{E}_{y^{*},\mathcal{L},\sigma}[\sum_{t = 1}^{T}V(a_{t}^{\sigma},(0.5,work),y_{t})|balanced_{all}] + o(T)$
\label{lem:balanced_close}
\end{lemma}

\begin{proof}

In this proof we use the fact that the distributions $y^{*}|balanced_{all}$ and $y^{*}$ are very close to each other to show that the Principal's expected payoff must be similar under both.

\begin{align*}
&\mathbb{E}_{y^{*},\mathcal{L},\sigma}[\sum_{t = 1}^{T}V(a_{t}^{\sigma},(0.5,work),y_{t})] \\
& =\mathbb{P}(balanced_{all})\mathbb{E}_{y^{*},\mathcal{L},\sigma}[\sum_{t = 1}^{T}V(a_{t}^{\sigma},(0.5,work),y_{t})|balanced_{all}] \\ & + \mathbb{P}(\neg balanced_{all})\mathbb{E}_{y^{*},\mathcal{L},\sigma}[\sum_{t = 1}^{T}V(a_{t}^{\sigma},(0.5,work),y_{t})|\neg balanced_{all}] \\
& \leq \mathbb{P}(balanced_{all})\mathbb{E}_{y^{*},\mathcal{L},\sigma}[\sum_{t = 1}^{T}V(a_{t}^{\sigma},(0.5,work),y_{t})|balanced_{all}] \\ & + T^{-\frac{1}{10}}\mathbb{E}_{y^{*},\mathcal{L},\sigma}[\sum_{t = 1}^{T}V(a_{t}^{\sigma},(0.5,work),y_{t})|\neg balanced_{all}]  \tag{By Lemma~\ref{lem:iteratedlog}}\\ 
& \leq \mathbb{P}(balanced_{all})\mathbb{E}_{y^{*},\mathcal{L},\sigma}[\sum_{t = 1}^{T}V(a_{t}^{\sigma},(0.5,work),y_{t})|balanced_{all}] + T^{-\frac{1}{10}} \cdot T \\ 
& \leq \mathbb{E}_{y^{*},\mathcal{L},\sigma}[\sum_{t = 1}^{T}V(a_{t}^{\sigma},(0.5,work),y_{t})|balanced_{all}] + o(T) \\ 
\end{align*}
\end{proof}

\begin{lemma}
If $y_{2:T} \sim y^{*}_{2:T}$, with probability at least $1 - T^{\frac{1}{10}}$, $balanced_{all}=true$. Furthermore, $balanced_{all}$ implies that the difference between the number of $M$ and $H$ states is $o(T)$. \label{lem:iteratedlog}
\end{lemma}

\begin{proof}
Let us consider $y^{*}_{2:T}$ to be a sequence of independent, identically distributed random variables $S$, where the value is $1$ when the state is $M$ and $-1$ otherwise. Then they have mean $0$ and variance $1$. The absolute value of the difference between the number of $M$ states and the number of $H$ states is now exactly equal to $|S_{T}| = |\sum_{i=1}^{T}y_{i}|$.

This is now a Rademacher random walk. By an application of the nonasymptotic version of the Law of Iterated Logarithm in \cite{balsubramani2015sharp}, we have that with probability $\geq 1 - T^{-\frac{1}{10}}$, for all $t \leq T$  simultaneously,

\begin{align*}
& |S_{t}| \leq \sqrt{3t(2log(log(\frac{5}{2}t))+ log(2T^{\frac{1}{10}))}} \\
& \leq \sqrt{3T(2log(log(\frac{5}{2}T))+ log(2T^{\frac{1}{10}))}} = o(T) \\
\end{align*}

\end{proof}

\begin{lemma}
In expectation over $y^{*}$, the expected payoff of any mechanism $\sigma$ against $noreg$ is at most $o(T)$. \label{lem:PrincipalpayoffUB}
\end{lemma}

\begin{proof}
Let $s_{m,w}$ be the number of rounds in which the state is medium and the Agent works, and define $s_{h,w}$, $s_{m,s}$, and $s_{h,s}$ accordingly. Let us assume for contradiction that the Principal receives expected payoff of at least $c \cdot T$. Then,

\begin{align*}
&\mathbb{E}_{y_{*},noreg,\sigma}[\sum_{t = 1}^{T}V(a_{t}^{\sigma},p_{t}^{\sigma},y_{t})]\geq c\cdot T \\ & \Rightarrow \mathbb{E}_{y^{*},noreg,\sigma}[\sum_{t=1}^{T}((2 - 2p_{t}) \cdot \mathbbm{1}[m,w])] \geq c \cdot T  \\
& \Rightarrow \mathbb{E}_{y^{*},noreg,\sigma}[2s_{m,w}]  - c \cdot T \geq  \mathbb{E}_{y^{*},noreg,\sigma}[\sum_{t=1}^{T} 2p_{t} \cdot \mathbbm{1}[m,w])] \\
\end{align*}

However, by assumption, we also have that 

\begin{align*}
& \mathbb{E}_{y^{*},noreg,\sigma}[\sum_{t=1}^{T}U(a_{t}^{\sigma},p_{t}^{\sigma},y_{t})] \geq -o(T) \tag{By the fact that the Agent could play $shirk$ every round and get $0$} \\
& \Rightarrow \mathbb{E}_{y^{*},noreg,\sigma}[\sum_{t=1}^{T}(2p_{t} \cdot \mathbbm{1}[m,w]) - s_{m,w} - s_{h,w}] \geq -o(T)  \\
& \Rightarrow \mathbb{E}_{y^{*},noreg,\sigma}[2 \cdot s_{m,w} - c\cdot T - s_{m,w} - s_{h,w}] \geq -o(T)\tag{Using the Principal payoff expression} \\
& \Rightarrow \mathbb{E}_{y^{*},noreg,\sigma}[s_{m,w} - s_{h,w}] \geq c \cdot T  -o(T) > 0 \tag{For sufficiently large $T$} \\
\end{align*}

As $\sigma$ does not take the states of nature as input, we know that $y_{t} \sim y^{*}$ is independent of $(p_{t},r_{t})$. Furthermore, as $noreg$ does not take the states of nature as input, we know that $a_{t}$, conditioned on $(p_{t},r_{t})$, is independent of $y_{t} \sim y^{*}$. Putting these together, we get that $a_{t}$ is independent of $y_{t}$. Therefore,

\begin{align*}
& \mathbb{E}_{y^{*},noreg,\sigma}[s_{m,w} - s_{h,w}] \\
& = \sum_{t=1}^{T}\mathbb{P}_{y^{*},noreg,\sigma}(y_{t} = M, a_{t} = work) - \sum_{t=1}^{T}\mathbb{P}_{y^{*},noreg,\sigma}(y_{t} = H, a_{t} = work) \\
& = \sum_{t=1}^{T}\mathbb{P}_{y^{*}_{1:t-1},noreg,\sigma}(a_{t} = work) \cdot \mathbb{P}_{y^{*}_{t}}y_{t} = M) - \sum_{t=1}^{T}\mathbb{P}_{y^{*}_{1:t-1},noreg,\sigma}(a_{t} = work) \cdot \mathbb{P}_{y^{*}_{t}}y_{t} = H) \tag{By the independence of $a$ and $y$} \\
& = \frac{1}{2}\sum_{t=1}^{T}\mathbb{P}_{y^{*}_{1:t-1},noreg,\sigma}(a_{t} = work)  - \frac{1}{2}\sum_{t=1}^{T}\mathbb{P}_{y^{*}_{1:t-1},noreg,\sigma}(a_{t} = work)  = 0 \\
\end{align*}

This derives a contradiction, proving our claim.

\end{proof}

\begin{lemma}
If $y_{1} = M$ then no Principal mechanism can get expected payoff more than $\frac{T}{4} + o(T)$ payoff against $\mathcal{L}$, in expectation over $y^{*}_{2:T}$. If $y_{1} = H$ then no Principal mechanism can get expected payoff more than $\frac{T}{5} + o(T)$ payoff against $\mathcal{L}$, in expectation over $y^{*}_{2:T}$. 
\label{lem:payoffUBextended}
\end{lemma}

\begin{proof}
First, assume $y_{1} = M$. Furthermore, let $t'$ be the first round in which the Principal mechanism $\sigma$ does not play $(0.5,work)$. Then, the payoff of the Principal is

\begin{align*}
& \mathbb{E}_{y_{2:T}^{*},\cL,\sigma}[\sum_{t = 1}^{t'}V(a_{t}^{\sigma},(0.5,work),y_{t})] + \mathbb{E}_{y_{2:T}^{*},\cL,\sigma}[\sum_{t = t'}^{T}V(a_{t}^{\sigma},p_{t}^{\sigma},y_{t})] \\
& = \mathbb{E}_{y_{2:T}^{*},(a^{*})^{\sigma}}[\sum_{t = 1}^{t'}V(a_{t}^{\sigma},(0.5,work),y_{t})] + \mathbb{E}_{y_{2:T}^{*},noreg,\sigma}[\sum_{t = t'}^{T}V(noreg^{\sigma},\sigma,p_{t}^{\sigma},y_{t})] \\
& = \frac{t'}{2} - \frac{t'}{4} + \mathbb{E}_{y_{2:T}^{*},noreg,\sigma}[\sum_{t = t'}^{T}V(noreg^{\sigma},\sigma,p_{t}^{\sigma},y_{t})] \\
& = \frac{t'}{4} + o(T)  \tag{By Lemma~\ref{lem:PrincipalpayoffUB}}\\
& \leq \frac{T}{4} + o(T) \\
\end{align*}

The analysis is similar for $y_{1} = H$. Let $t'$ be the first round in which the Principal mechanism $\sigma$ does not play $(0.6,work)$. Then, the payoff of the Principal is

\begin{align*}
& \mathbb{E}_{y_{2:T}^{*},\cL,\sigma}[\sum_{t = 1}^{t'}V(a_{t}^{\sigma},(0.6,work),y_{t})] + \mathbb{E}_{y_{2:T}^{*},\cL,\sigma}[\sum_{t = t'}^{T}V(a_{t}^{\sigma},p_{t}^{\sigma},y_{t})] \\
& = \mathbb{E}_{y_{2:T}^{*},(b^{*})^{\sigma}}[\sum_{t = 1}^{t'}V(a_{t}^{\sigma},(0.6,work),y_{t})] + \mathbb{E}_{y_{2:T}^{*},noreg,\sigma}[\sum_{t = t'}^{T}V(noreg^{^{\sigma}},\sigma,p_{t}^{\sigma},y_{t})] \\
& = \frac{2t'}{5} - \frac{t'}{5} + \mathbb{E}_{y_{2:T}^{*},noreg,\sigma}[\sum_{t = t'}^{T}V(noreg^{^{\sigma}},\sigma,p_{t}^{\sigma},y_{t})] \\
& = \frac{t'}{5} + o(T)  \tag{By Lemma~\ref{lem:PrincipalpayoffUB}}\\
& \leq \frac{T}{5} + o(T) \\
\end{align*}

\end{proof}

\begin{lemma}
For any prefix of play of length $T' \leq T$, as long as $balanced_{all}=true$ and the Principal plays only $0.5,work$ for all $\sigma$,
$$\E_{a^{*},\sigma}[\ir(y_{1:T},p_{1:T},r_{1:T})] \leq o(T)$$
and 
$$\E_{a^{*},\sigma}[\nr(y_{1:T},p_{1:T}^\sigma,r^\sigma_{1:T})] \leq o(T)$$

Similarly, for any prefix of play of length $T' \leq T$, as long as $balanced_{all}=true$, the Principal plays only $0.6,work$, for all $\sigma$,
$$\E_{b^{*},\sigma}[\ir(y_{1:T},p_{1:T},r_{1:T})] \leq o(T)$$
and
$$\E_{b^{*},\sigma}[\nr(y_{1:T},p_{1:T}^\sigma,r^\sigma_{1:T})] \leq o(T)$$
\label{lem:boundreg}
\end{lemma}

\begin{proof}

For the first case of $(0.5,work)$ and $a^{*}$, the Agent is always mapping $M$ to $work$ and $H$ to $shirk$. As $work$ gets payoff $2p - 1 \geq 0$ under $m$ and $shirk$ gets $0$, while $work$ gets $-1$ under $H$ and $shirk$ gets $0$, this is the optimal mapping. Therefore the Contextual Swap Regret is $0$. Now we can upper bound the negative regret:

\begin{align*}
& \E_{a^{*},\sigma}\left[\max_{h:\cP\times \cA\mapsto \cA}\sum_{t=1}^{T'} U(h(p_t^{\sigma},r_t^{\sigma}),p_t^{\sigma}, y_t)) - U(a_{t}^{\sigma},p_t^{\sigma}, y_t)) \right] \\
& =\E_{a^{*},\sigma}\left[\max_{a \in \cA}\sum_{t=1}^{T'} U(h(0.5,work),0.5, y_t)) - U(a_{t}^{\sigma},0.5, y_t)) \right] \tag{By the fact that the Principal is making a fixed (policy, recommendation) pair across all $t\leq T'$} \\
&= \E_{a^{*},\sigma}\left[\max_{a \in \cA}\sum_{t=1}^{T'} U(h(0.5,work),0.5, y_t)) \right] \tag{By definition of $a^{*}$} \\
&= \max(\E_{a^{*},\sigma}[\sum_{t=1}^{T'} U(work,0.5, y_t)], \E_{a^{*},\sigma}[\sum_{t=1}^{T'} U(shirk,0.5, y_t)) \\
&= \max(\E_{a^{*},\sigma}[\frac{1}{2}m_{y,T'} - h_{y,T}], 0)  \\
&\leq \max(\E_{a^{*},\sigma}[\frac{1}{2}h_{y,T'} + o(T) - h_{y,T}], 0) \tag{By the fact that $balanced_{t}$ is true over the entire prefix.}  \\
& \leq o(T)
\end{align*}

For the second case of $(0.6, work)$ and $b^{*}$, let us use $m_{y,T'}$ to refer to the number of $m$ states in the sequence, and $h_{y,T'}$ to refer to the number of $h$ states:

\begin{align*}
& \E_{b^{*},\sigma}\left[\max_{h:\cP\times \cA\times \cA\mapsto \cA}\sum_{t=1}^{T'} U(h(p_t^{\sigma},r_t^{\sigma},a_t^{\sigma}),p_t^{\sigma}, y_t)) - U(a_{t}^{\sigma},p_t^{\sigma}, y_t)) \right] \\
& =\E_{b^{*},\sigma}\left[\max_{h: \cA\mapsto \cA}\sum_{t=1}^{T'} U(h(0.6,work,a_t^{\sigma}),0.6, y_t)) - U(a_{t}^{\sigma},0.6, y_t)) \right] \tag{By the fact that the Principal is making a fixed (policy, recommendation) pair across all $t\leq T'$} \\
&= \E_{b^{*},\sigma}\left[\max_{h: \cA\mapsto \cA}\sum_{t=1}^{T'} U(h(0.6,work,a_t^{\sigma}),0.6, y_t)) - \frac{1}{5}(m_{y,T'} - h_{y,T'}) \right] \tag{By definition of $b^{*}$} \\
&= \E_{b^{*},\sigma}[\max_{a \in \cA}\sum_{t=1}^{T'} U(a,0.6, y_t))\mathbbm{1}[a_{t}^{\sigma} = work]] \\& + \E_{b^{*},\sigma}[\max_{a \in \cA}\sum_{t=1}^{T'} U(a,0.6, y_t))\mathbbm{1}[a_{t}^{\sigma} = shirk]] - \E_{b^{*},\sigma}[\frac{1}{5}(m_{y,T'} - h_{y,T'}) ]  \\
&= \max(\E_{b^{*},\sigma}[\sum_{t=1}^{T'} U(work,0.6, H))\mathbbm{1}[a_{t}^{\sigma} = shirk]], 0) - \frac{1}{5}(m_{y,T'} - h_{y,T'})  \tag{By the fact that $b^{*}$ only shirks when $y = H$, and that shirking always guarantees payoff $0$.}  \\
&= \max(m_{y},T' - h_{y,T'}, 0) - \frac{1}{5}(m_{y,T'} - h_{y,T'})   \tag{By the distribution of the states conditioned on $b^{*}$ playing $work$} \\
& \leq o(T) \tag{By the fact that $balanced_{t}$ is true over the entire prefix.} 
\end{align*}

Thus, in the second case the Contextual Swap Regret is upper bounded. Finally, we need that the Negative Regret is upper bounded: 

\begin{align*}
& \E_{b^{*},\sigma}\left[\max_{h:\cP\times \cA\mapsto \cA}\sum_{t=1}^{T'} U(h(p_t^{\sigma},r_t^{\sigma}),p_t^{\sigma}, y_t)) - U(a_{t}^{\sigma},p_t^{\sigma}, y_t)) \right] \\
& =\E_{b^{*},\sigma}\left[\max_{a \in \cA}\sum_{t=1}^{T'} U(h(0.6,work),0.6, y_t)) - U(a_{t}^{\sigma},0.6, y_t)) \right] \tag{By the fact that the Principal is making a fixed (policy, recommendation) pair across all $t\leq T'$} \\
&= \E_{a^{*},\sigma}\left[\max_{a \in \cA}\sum_{t=1}^{T'} U(h(0.6,work),0.6, y_t)) - \frac{1}{5}(m_{y,T'} - h_{y,T'}) \right] \tag{By definition of $b^{*}$} \\
&= \max(\E_{b^{*},\sigma}[\sum_{t=1}^{T'} U(work,0.6, y_t)], \E_{a^{*},\sigma}[\sum_{t=1}^{T'} U(shirk,0.6, y_t)) - \E_{a^{*},\sigma}[(m_{y,T'} - h_{y,T'}) ]  \\
&= \max(\E_{b^{*},\sigma}[\frac{1}{2}m_{y,T'} - h_{y,T}], 0) - \E_{b^{*},\sigma}[(m_{y,T'} - h_{y,T'}) ]  \\
& \leq o(T) \tag{By the fact that $balanced_{t}$ is true over the entire prefix.}  \\
\end{align*}

\end{proof}
\end{document}